\documentclass[12pt]{amsart} 
\usepackage{hyperref}
\usepackage{amssymb}
\usepackage{amsmath, enumerate}
\usepackage{mathrsfs}
\usepackage[arrow,matrix,curve]{xy}
\usepackage{graphicx}
\usepackage{amscd,verbatim}
\usepackage{mathtools}
\usepackage{comment}

\usepackage{blindtext}
\usepackage{lipsum}
\usepackage{bm}
\usepackage{color,soul}
\usepackage{tikz-cd}
\usepackage[T1]{fontenc}
\usepackage{fullpage,enumerate}

\newcommand{\bea}{\begin{eqnarray}}
\newcommand{\eea}{\end{eqnarray}}
\newcommand{\beq}{\begin{equation}}
\newcommand{\eeq}{\end{equation}}

\newcommand\wt{\widetilde}

\newtheorem{theorem}{Theorem}[section]

\newtheorem{proposition}[theorem]{Proposition}

\theoremstyle{definition}
\newtheorem{definition}{Definition}[section]

\theoremstyle{remark}
\newtheorem{remark}[theorem]{Remark}

\newtheorem{example}[theorem]{Example}

\numberwithin{equation}{section}
\numberwithin{theorem}{section}

\newcommand\cA{\mathcal{A}}

\newcommand\cE{\mathcal{E}}

\newcommand\cH{\mathcal{H}}

\newcommand\cG{\mathcal{G}}

\newcommand\cL{\mathcal{L}}

\newcommand\cP{\mathcal{P}}

\newcommand\cS{\mathcal{S}}
\newcommand\cV{\mathcal{V}}
\newcommand\cW{\mathcal{W}}

\newcommand{\CC}{{\mathbb C}}

\newcommand{\RR}{{\mathbb R}}
\newcommand{\TT}{{\mathbb T}}
\newcommand{\ZZ}{{\mathbb Z}}

\newcommand{\HH}{{\mathbb H}}
\newcommand{\PP}{{\mathbb P}}

\newcommand{\mV}{{\mathscr V}}

\newcommand{\vect}[1]{\bm{{#1}}}
\newcommand{\im}{\text{i}}
\newcommand{\floor}[1]{\lfloor #1 \rfloor}

\newcommand{\eul}{\mathfrak{Eul}}
\newcommand{\coeul}{\mathfrak{coEul}}
\newcommand{\kerv}{\mathfrak{Kerv}}
\newcommand{\cokerv}{\mathfrak{coKerv}}

\usepackage[title]{appendix}
\begin{document}

\title[Differential topology of semimetals]{Differential topology of semimetals}

\author{Varghese Mathai}
\address{Department of Pure Mathematics, University of Adelaide,
Adelaide 5005, Australia}
\email{mathai.varghese@adelaide.edu.au}

 \author{Guo Chuan Thiang}
\address{Department of Pure Mathematics, University of Adelaide,
Adelaide 5005, Australia}
\email{guochuan.thiang@adelaide.edu.au}







\begin{abstract}
The subtle interplay between local and global charges for topological semimetals exactly parallels that for singular vector fields. Part of this story is the relationship between cohomological semimetal invariants, Euler structures, and ambiguities in the connections between Weyl points. Dually, a topological semimetal can be represented by Euler chains from which its surface Fermi arc connectivity can be deduced. These dual pictures, and the link to topological invariants of insulators, are organised using geometric exact sequences. We go beyond Dirac-type Hamiltonians and introduce new classes of semimetals whose local charges are subtle Atiyah--Dupont--Thomas invariants globally constrained by the Kervaire semicharacteristic, leading to the prediction of torsion Fermi arcs.
\end{abstract}


  \maketitle
  \date{}

\setcounter{tocdepth}{2}
\tableofcontents


\section{Introduction}
As theoretically predicted in \cite{Wan} and experimentally discovered in \cite{Xu1,Xu2,Lv,Zhang}, two-band solid state systems in 3D can host topologically protected ``Weyl semimetallic phases'' in which the quasiparticle excitations at the band crossings (``Weyl points'') share some features with Weyl fermions from relativistic quantum mechanics. The experimental signature, namely ``Fermi arcs'' of surface states which connect the projected Weyl points in the surface Brillouin zone, is just as remarkable. Based on this initial success, much effort has been put into the general study of topological semimetallic phases in the hope of predicting and eventually realising new exotic fermionic quasiparticles in condensed matter systems.

Most proposals have focused on \emph{local} aspects in the sense of finding new types of topological obstructions to locally opening up gaps in semimetal band crossings. Staying in the two-band case, there are generalisations of the basic Weyl semimetal phase to ``Type-II'' ones \cite{XutypeII,Sol}, as well as ``quadratic'' ones \cite{Huang}. One can consider model Hamiltonians with more than two bands and in a different number of spatial dimensions \cite{ZhangLian}. Following the example set by topological insulator theory, one may also introduce antiunitary symmetry constraints such as time-reversal and charge-conjugation \cite{Zhao1}. To circumvent certain difficulties in achieving non-trivial obstructions, point symmetries were also introduced into the game, and a host of possibilities arise \cite{Bradlyn}. Our results in this paper have a different focus, concentrating on (1) isolating in a conceptually simple but rigorous way the general mathematical mechanism allowing for local semimetallic charges, (2) providing the full \emph{global} topological characterisation of semimetal band structures, Fermi arcs, and the relation to topological insulator invariants, and (3) introducing a new family of semimetals whose topological invariants have a very different character to those commonly used in the literature. Figure \ref{fig:Flowchart} summarises the physical and mathematical concepts which we introduce in this paper.

 \begin{figure}
 \centering
        \includegraphics[width=0.9\textwidth]{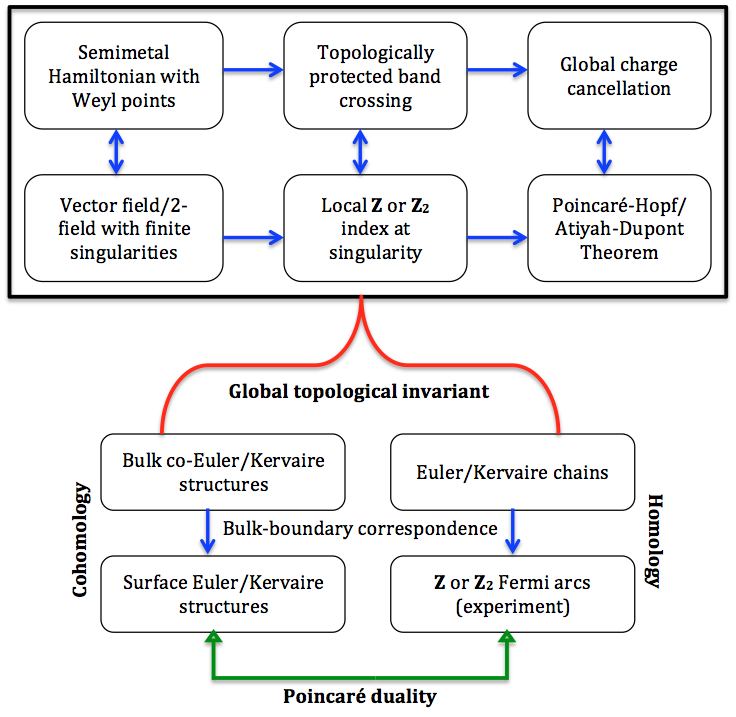}
        \caption{Semimetals have Bloch Hamiltonians with band crossings at certain isolated Weyl points in the Brillouin torus $T$, corresponding to vector fields over $T$ with finite isolated singularities at the Weyl points. These band crossings may be topologically protected, with their $\ZZ$ topological charges corresponding to the local indices of the vector field at the singularities. The charge-cancellation condition corresponds to the Poincar\'{e}--Hopf theorem. The global semimetal band structure can be studied in two complementary ways (Definition \ref{defn:semimetaltopinv}) which are dual to each other and related by Poincar\'{e} duality. Semimetals have characteristic classes in cohomology, identifiable as co-Euler structures, while in homology they are represented by the dual Euler chains. Under a bulk-boundary correspondence, these invariants are mapped to surface invariants, which in homology is just the projection of the Euler chain onto the Fermi arcs seen in experiments. These ideas can be generalised to semimetal Hamiltonians specified by tangent 2-fields, which now have $\ZZ_2$-charges. Euler structures and chains are replaced by Kervaire ones, and the Poincar\'{e}--Hopf theorem replaced by a theorem of Atiyah--Dupont.}\label{fig:Flowchart}
 \end{figure}
 
With regards to (1), we see, with much hindsight, that the local topological charges protecting basic types of semimetal crossings are completely inherited from the singular vector field which specifies the Hamiltonian \cite{Witten,MTFermi}. The Poincar\'{e}--Hopf theorem \cite{MilnorDT} then shows how the Brillouin zone topology forces a global charge cancellation condition, and the prediction of Fermi arcs becomes a corollary of the bulk-boundary correspondence. This brings us to point (2). While the 3D Weyl semimetal may be intuitively understood using a Stokes' Theorem argument to predict ``jumps'' in 2D Chern numbers (or ``weak invariants'') across the Weyl points, there are several subtleties involved, not the least of which is the fact that the local charges (equal to the jumps) fail to completely characterise the semimetal, and cannot predict, a priori, the Fermi arc topology \cite{MTFermi}. It is here that the finer topology of the Brillouin torus comes into play: the fact that $\TT^d$ has non-trivial cycles means that ``the (Fermi) arc joining two Weyl points'' is an ambiguous notion, even up to deformations of the arc (contrast with the simply-connected spheres). The resolution of this ambiguity requires additional global data which captures the full topology of semimetal band structures, and is closely related to the concept of Euler structures introduced by Turaev to resolve ambiguities in torsions of Reidemeister type \cite{TuraevEuler}; both hinge on the fact that the first homology of the underlying manifold may be non-zero. By considering a homological version of \emph{Euler structures} \cite{Burg,Hutchings}, we find that the language of \emph{Euler chains} very concisely represents all the \emph{topological} features of a semimetal. The Euler chain representation directly predicts the \emph{topology} of the Fermi arcs, and the latter are the experimental signature of a topological semimetal. Furthermore, the notion of topological equivalence between semimetals becomes clear in this language, as a kind of homotopy between their defining vector fields. A semimetal Hamiltonian determines a Fermi arc, e.g.\ through a transfer matrix formalism \cite{Hatsugai2,Avila,Dwivedi2}, and this determination descends to the level of topological invariants in the sense that equivalent Hamiltonians give rise to equivalent Fermi arcs. 

Having isolated the basic mathematical principles underlying the geometry and topology of semimetals, we move on to (3): the generalisation of our theory for 2-frames or tangent 2-fields along the lines of the Atiyah--Dupont theorem \cite{AD} which is the analogue of the Poincar\'{e}--Hopf theorem\footnote{The latter accounts for the Nielsen--Ninomiya no-go theorem \cite{NN} excluding chiral fermions in lattice gauge theory.} in this context, where the  \emph{Kervaire semicharacteristic} plays the role of the Euler characteristic.
It is a topological invariant of a very different nature to those that had been previously utilised in the literature on topological phases (e.g.\ Chern classes and related characteristic classes, $K$-theory, winding numbers/degrees).
This leads us to define some new mathematical notions --- \emph{Kervaire structures} and \emph{Kervaire chains} --- which are two mutually dual objects characterising a new class of $\ZZ_2$-topological semimetals introduced for the first time in this paper, predicting in principle that torsion Fermi arcs can be experimentally found.

\subsection{Outline}
In Section \ref{sec:physback}, we give a quick overview of the use of cohomological characteristic classes to study topological insulating phases, and then connect it to semimetal phases via the Mayer--Vietoris technique of ``patching together'' the non-singular portion (avoiding the Weyl points) and topologically trivial portions (small neighbourhoods of the Weyl points). The basic ``semimetal Mayer--Vietoris sequence'' is given a geometric interpretation as an extension problem, parallel to the physical problem of deforming a semimetal band structure into a globally insulating one, in Section \ref{sec:extension}. In Section \ref{sec:Euler}, we introduce the concept of Euler structures and the Euler chain representation of a semimetal. We illustrate with diagrams the basic intuitive ideas, and apply the Euler chain concept to Fermi arc topology and their rewirings. In Section \ref{sec:higherdimension}, we study Dirac-type Hamiltonians parametrised by singular vector fields, emphasising the merits of an abstract coordinate-free analysis with regards to symmetry considerations. In particular, we justify the generality of Dirac-type Hamiltonians through natural symmetry constraints, in order to avoid introducing ``spurious topology'' through ad-hoc parametrisations of toy model Hamiltonians. In Section \ref{sec:semimetalgerbe}, a geometric connection to gerbes is made in the 4D case (we refer the reader to Appendix \ref{appendix} for a primer on gerbes and the construction of the basic gerbe on $\text{SU}(2)$), and in Section \ref{sec:5Dsemimetal}, the connection to quaternionic valence bundles in 5D is made, and the geometric interpretation as an extension problem is  in Section \ref{sec:extension}. In Section \ref{sec:torsionsemimetal}, we introduce a new class of semimetals modelled on Hamiltonians which are bilinear in gamma matrices. The topology of such semimetals is characterised by the notion of Kervaire structures and chains, and we predict an interesting phenomenon of $\ZZ_2$ surface Fermi arcs. Finally we suggest some directions for future work in Section \ref{sec:Outlook}.

\bigskip

{\bf Conventions.} In this paper, $T$ will always be a compact connected smooth $d$-dimensional manifold without boundary, playing the role of a general parameter manifold for a family of Hamiltonians. In specific instances, it may have an orientation, metric, spin${}^c$ structure etc., and may be the Brillouin zone for a family of Bloch Hamiltonians, as additionally specified. All (co)-homology groups have integer coefficients unless otherwise indicated. The real Clifford algebra $Cl_{r,s}$ has $r$ generators squaring to $-1$ and $s$ generators squaring to $+1$, corresponding to an orthonormal basis $e_1,.\ldots,e_{r+s}$ in $\RR^{r+s}$ for the bilinear form $(-,\ldots,-,+\ldots +)$; $Cl_n$ refers to $Cl_{0,n}$, and the complex Clifford algebra $\CC l_n$ has $n$ anticommuting generators squaring to $+1$.


\section{Physical background: Topological phases, cohomology}\label{sec:physback}
\subsection{Topological band insulators and cohomology groups}\label{sec:2DBerry}
2D Chern insulators are classified by a first Chern number, equal to the integrated Berry curvature 2-form of the valence bands over the Brillouin zone $T=\TT^2$. When the Fermi energy $E_F$ lies in a spectral gap, the Fermi projection onto the valence bundle $\cE_F$ may be defined. Over $\TT^2$ the $\cE_F$ are classified by their rank and first Chern class, with the latter generally considered to be the interesting invariant. Since the first Chern class of $\cE_F$ is that of the determinant line bundle, the classification problem of 2D Chern insulators is formally equivalent to that of U(1) line bundles over the Brillouin zone for which the \emph{integral cohomology group} $H^2(\TT^2,\ZZ)\cong\ZZ$ provides a complete answer. In terms of differential forms, $H^2(\TT^2,\ZZ)\neq 0$ measures the failure of the (Berry) curvature 2-form $\mathcal{F}$ of $\cE_F$ to be globally exact (the Berry connection 1-form $\cA$ only exists locally). The integral of $\mathcal{F}$ over $\TT^2$ yields an integer-valued first Chern number, which may be interpreted as a topological obstruction to globally defining a basis of valence Bloch eigenstates for $\cE_F$.

A typical way to construct a valence line bundle with first Chern class $n\in \ZZ$ is to consider $2\times 2$ traceless\footnote{We normalize the Fermi level to $0$.} Bloch Hamiltonians $H(k)=\vect{h}(k)\cdot\vect{\sigma},\, k\in\TT^2$, parametrised by a nowhere-vanishing smooth vector field $\vect{h}$ on $\TT^2$. Here $\vect{\sigma}=(\sigma_1,\sigma_2,\sigma_3)$ are the Pauli matrices. The spectrum is $\pm |\vect{h}(k)| \neq 0$, and the spectrally-flattened Hamiltonian $\widehat{H}$ is parametrised by the unit vector map $\hat{\vect{h}}:\TT^2\rightarrow S^2\subset\RR^3$. Note that $H$ and $\widehat{H}$ determine the same valence line subbundle $\cE_F$ of the Bloch bundle $\TT^2\times\CC^2$. The map $\hat{\vect{h}}$ has an integer degree given by the formula $\hat{\vect{h}}_*([\TT^2])=\text{deg}(\hat{\vect{h}})[S^2]$, where $[\TT^2], [S^2]$ denote the fundamental classes. The degree is a homotopy invariant, and all degrees do occur. By identifying $S^2$ with $\CC\PP^1$ (the Bloch sphere construction), we see that $\cE_F$ is given by the pullback of the tautological (Hopf) line bundle over $\CC\PP^1$ under $\hat{\vect{h}}$. Due to the low-dimensionality of $\TT^2$, all maps into $\CC\PP^\infty=K(\ZZ,2)$ are approximated by maps into $\CC\PP^1$, and so the first Chern class of $\cE_F$ is precisely the degree of $\hat{\vect{h}}$.

In higher dimensions, and also in the presence of additional symmetries (e.g.\ time-reversal or point group symmetries) constraining the form of the Hamiltonians, topological band insulators are classified by more complicated invariants. In the mathematical physics literature, these have included higher degree cohomology invariants playing the role of characteristic classes for the valence bundles, e.g.\ higher Chern/instanton numbers, symplectic bundle invariants \cite{ASSS,ASSS2,Hatsugai}, Fu--Kane--Mele/FKMM-type invariants \cite{dNG}, or generalised cohomology invariants ($K$-theory) \cite{FM,Thiang}. We are primarily interested in $3, 4$ or $5$ spatial dimensions, and in those dimensions, the relevant ``Berry curvature form'' $\mathcal{F}$ can be a higher-degree form \cite{Gaw,Gaw2,LinYau}, and so a higher degree cohomology group $H^n(X,\ZZ)$ comes into play. Cohomology groups with other coefficients also arise quite generally in obstruction theory, which can be used to study semimetal-insulator transitions. All these (generalised) cohomological invariants are also homotopy invariants, and are closely related to degree theory (as we saw from the 2D Chern insulator described above).

\subsection{Topological semimetals}
Although the prototypical topological semimetal is the Weyl semimetal in $d=3$, the general theory presented in this paper works in higher dimensions. This is not merely of theoretical interest. In the case of topological insulators, there are concrete proposals and experiments in which topological phases of physical systems in three or fewer dimensions formally realise those in $d>3$. Some examples are quasicrystalline systems \cite{Kraus}, time-periodic (Floquet) insulators \cite{Lindner,Gaw} and their photonic counterparts \cite{Rechtsman}, and general ``virtual'' topological insulators \cite{Prodan}. It is in this spirit that we embark on the mathematical study of semimetallic phases for general $d$. 

\subsubsection{Local aspects of semimetal topology}
In a semimetal, the Fermi level does not lie in a spectral gap, but instead passes through band crossings at some Weyl submanifold $W$ of the momentum space manifold $T$. Such crossings may arise as ``accidental degeneracies'' \cite{vN,Herring} and, for ``Dirac-type Hamiltonians'' in $d$-dimensions occur generically at points. A quick way to see this is to consider Dirac-type Bloch Hamiltonians of the form
\begin{equation}
H(k)=\vect{h}(k)\cdot\vect{\gamma}, \qquad k\in T;\label{Dirac-type-Hamiltonian}
\end{equation}
Here each $k\mapsto\vect{h}(k)$ is a smooth assignment of $d$-component vectors, and $\vect{\gamma}=(\gamma_1,\ldots,\gamma_d)$ is the vector of traceless Hermitian $\gamma$-matrices representing irreducibly and self-adjointly the $d$-generators of a Clifford algebra (for $d=3$, these are the usual $2\times 2$ Pauli matrices). By the Clifford algebra relations $\gamma_i\gamma_j+\gamma_j\gamma_i=2\delta_{ij}$, the square of $H(k)$ is the scalar $|\vect{h}(k)|^2$, so the eigenvalues of $H(k)$ are $\pm |\vect{h}(k)|$, each with degeneracy equal to half the dimension of the $\gamma$-matrices. This means that a band crossing at $k\in T$ requires $d$ conditions $\vect{h}(k)=\vect{0}$, and we see that the locus $W$ of band crossings generically has codimension $d$. When $T$ is a $d$-manifold, $W$ is generically a collection of isolated ``Weyl points'', the (rather misleading) terminology arising from the $d=3$ case in which the low-energy excitations around a band crossing are essentially described by the Weyl equation \cite{Wan,Turner}. Furthermore, each Weyl point $w\in W$ serves as a generalised {\em monopole}, and may be assigned a {\em local topological charge} via the unit-vector map restricted to a small enclosing $d-1$-sphere $S_w$,
\begin{equation} 
\hat{\vect{h}}_w = \frac{\vect{h}(k)}{ |\vect{h}(k)|} : S_w \longrightarrow
S^{d-1}\subset\RR^d, \qquad d=3,4,5,\label{unitvector}
\end{equation}
and defining the local charge at $w$ to be the degree of $\hat{\vect{h}}_w$, equal to the local index $\text{Ind}_w(\vect{h})$ of $\vect{h}$ at $w$. The local charge information of all the crossings for the Hamiltonian $\vect{h}\cdot\vect{\gamma}$ is concisely summarised by the 0-chain $\cW_{\vect{h}}\coloneqq \sum_{w\in W}\text{Ind}_w(\vect{h}) w \in C_0(T,\ZZ)$, where $C_i(\cdot,\ZZ)$ denotes the set of singular $i$-chains.

The unit vector map $\hat{\vect{h}}$ on $T\setminus W$, and  its restrictions \eqref{unitvector} to $S_w$, can be used to pullback the Hopf line bundle with connection when $d=3$, the basic gerbe with connection when $d=4$ and the quaternionic Hopf line bundle with connection when $d=5$ (see Appendix \ref{appendix} for their definitions). As explained in Section \ref{sec:extension}, these geometrical structures arise naturally when analysing Dirac-type Hamiltonians, their band crossings, and the topology of their valence bundles.

\subsubsection{Global aspects of semimetal topology}
Interesting phenomena arise when we study the band crossings \emph{globally}. For instance, if $\vect{h}$ can be identified as a tangent vector field over $T$, as is often the case in physical models, then the classical Poincar\'{e}--Hopf theorem guarantees that $\sum_{w\in W}\text{Ind}_w(\vect{h})=\chi(T)=0$ where $\chi(T)$ is the Euler characteristic of $T$, which is zero for $T=\TT^d$ or $d$ odd. This global ``charge-cancellation'' condition is well-known in lattice gauge theory (where $T=\TT^d$) as the Nielsen--Ninomiya theorem \cite{NN} and discussed in the semimetal context in \cite{Witten,MTFermi}. Together with the bulk-boundary correspondence, it predicts the appearance of surface Fermi arcs connecting the projected Weyl points on a surface Brillouin zone. These arcs are an experimental signature of topological non-triviality in a semimetal band structure, and have been discovered recently \cite{Xu1,Lv}.

More generally, $\vect{h}$ could be a section of some oriented rank-$d$ vector bundle $\cE$ over a compact oriented $d$-manifold $T$. The local topological charges at the zeros of $\vect{h}$ are defined as before, as is the Euler characteristic of $\cE$ (evaluating the Euler class on the fundamental class). When the Euler characteristic $\chi$ vanishes (and also when $\chi\neq 0$ following \cite{Burg}), a notion of \emph{Euler structures} \cite{TuraevEuler} can be defined. The set of Euler structures of $T$ form a torsor over $H_1(T,\ZZ)$, and it provides an elegant way to understand some global aspects of Fermi arc topology, including certain ambiguities first studied in \cite{MTFermi}, the ``rewiring'' of Fermi arcs \cite{Dwivedi,LFF}, and Weyl point creation/annihilation (Section \ref{sec:Eulerheuristic}). This is already interesting when $T=\TT^d$, for which there are non-trivial Euler structures: the various different ways of connecting two Weyl points by an arc are permuted amongst each other under $H_1(T,\ZZ)$. We remark that Euler structures had previously been been used to clarify certain ambiguities in the Reidemeister torsion and Seiberg--Witten invariants of 3-manifolds. Their appearance in the semimetal context indicates a similar subtlety and richness in the structure of semimetals. 

We take the following as a working definition:
\begin{definition}\label{defn:abstractHamiltonian}
Let $T$ be a compact, oriented $d$-dimensional Riemannian spin$^{c}$ manifold with $d\geq 3$. An \emph{abstract Dirac-type Hamiltonian} is a smooth vector field $\vect{h}$ over $T$, which is \emph{insulating} if $\vect{h}$ is non-singular (nowhere zero), and \emph{semimetallic} if $\vect{h}$ has a finite set $W$ of singularities.
\end{definition}
This definition is motivated by the construction of concrete Bloch Hamiltonians from such vector fields, generalising \eqref{Dirac-type-Hamiltonian}, given in Section \ref{sec:higherdimension}.

\subsection{Mayer--Vietoris principle: connecting topological insulators and semimetals}
In \cite{MTFermi}, we carried out an analysis of the global structure of semimetal band structures and Fermi arcs. The Mayer--Vietoris principle provided the key to understanding the connection between topological semimetals and insulators (Fig.\ \ref{fig:metal-or-insulator}). In particular, we showed that the local charge information at Weyl points needs to be supplemented with some global data to fully characterise semimetal band structure and Fermi arcs. We summarise the key constructions in \cite{MTFermi}, then provide several different ways to understand them.

 \begin{figure}
 \centering
        \includegraphics[width=0.7\textwidth]{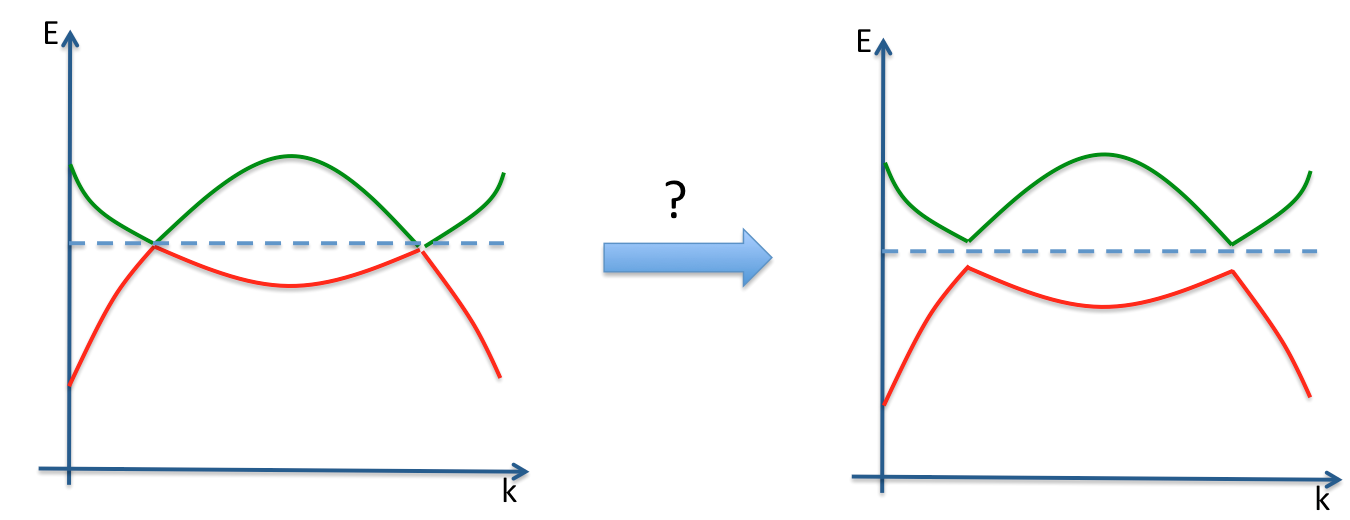}
        \caption{There are local obstructions to opening up a gap at a band crossings, related by a global consistency condition.}\label{fig:metal-or-insulator}
 \end{figure}

{\bf Definition of $W, S_W, D_W$.} Let $T$ be a compact oriented $d$-manifold (usually the Brillouin torus $\TT^d$ in concrete models) with $d\geq 3$, and consider a smooth family of (concrete) $n\times n$ Bloch Hamiltonians $T\ni k\mapsto H(k)$. The spectrum comprises $n$ bands, and we assume that band crossings at the Fermi level (normalized to $0$) occur on a finite set $W\subset T$ of isolated ``Weyl points''. On the complement $T\setminus W$, the Bloch Hamiltonians are gapped, and there is no difficulty in defining the valence bundle $\cE_F$ over $T\setminus W$ and its characteristic classes in $H^*(T\setminus W)$. For each $w\in W$, choose (mutually disjoint) open balls $D_w$ containing $w$, and $(d-1)$-spheres $S_w\subset D_w$ surrounding $w$. Let $S_W=\coprod_{w\in W}S_w$ and $D_W=\coprod_{w\in W}D_w$ be the respective disjoint unions. Thus $T$ is covered by $D_W$ and $T\setminus W$, and the intersection of the covering sets is $S_W\times (-\epsilon,\epsilon)$ which is homotopy equivalent to $S_W$ (Figure \ref{fig:MV-subspaces} illustrates the setting for two Weyl points).

 \begin{figure}
 \centering
        \includegraphics[width=0.8\textwidth]{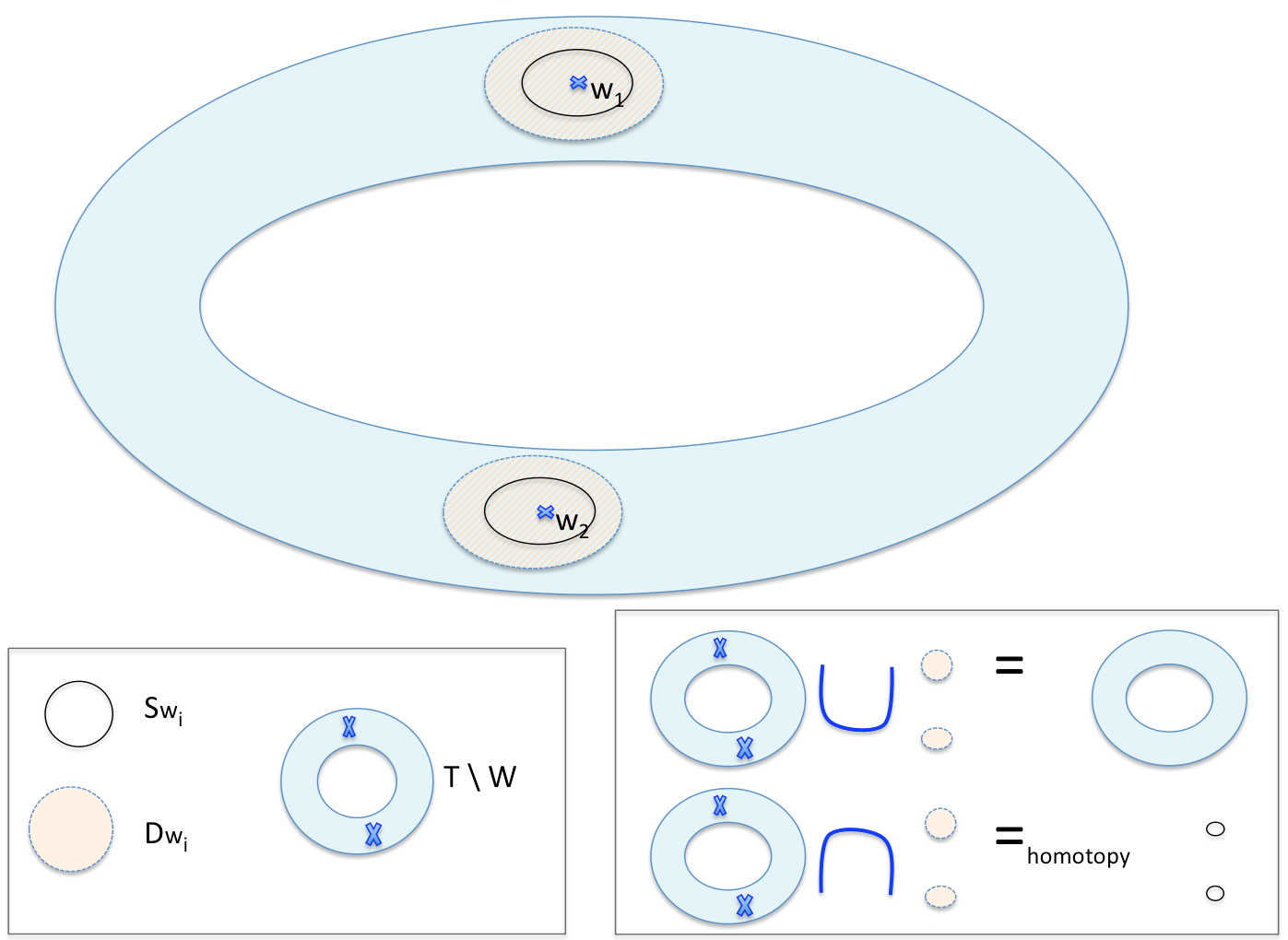}
        \caption{Covering of $T$ used in the semimetal MV-sequence, for two Weyl points $w_1, w_2$.}\label{fig:MV-subspaces}
 \end{figure}

The MV sequence links the cohomology groups of $T$ with those of the covering subspaces and the intersection $S_W$, by an exact sequence
\begin{equation}
\cdots \rightarrow H^{*-1}(S_W)\rightarrow H^*(T)\rightarrow H^*(T\setminus W)\oplus H^*(D_W)\rightarrow H^*(S_W)\rightarrow H^{*+1}(T)\rightarrow\cdots.\nonumber
\end{equation}
The maps which increase cohomological degree are the MV connecting maps, while the others are (differences of) restriction maps. Note that the $H^*(D_W)$ terms vanish for $*>0$ since $D_W$ is a disjoint union of contractible sets.

When a semimetallic Bloch Hamiltonian is specified abstractly by a vector field $\vect{h}$ as in Definition \ref{defn:abstractHamiltonian}, the zero set of $\vect{h}$ will correspond to the Weyl points $W$, and we have the subsets $S_W, D_W$ exactly as above. Then $\vect{h}$ defines a topological invariant for the semimetal which lives in $H^{d-1}(T\setminus W)$ (Section \ref{sec:Euler}). The most interesting part of the MV sequence is at $*=d-1$, which we call the cohomological \emph{semimetal MV sequence},
\begin{equation}
\cdots 0\longrightarrow \underbrace{H^{d-1}(T)}_{\text{(weak) insulator}}\overset{\iota^*}{\longrightarrow} \underbrace{H^{d-1}(T\setminus W) }_{\text{insulator/semimetal}}
\overset{\beta}{\longrightarrow} \underbrace{H^{d-1}(S_W)}_{\text{local charges}}\overset{\Sigma}{\longrightarrow}\underbrace{H^d(T)}_{\text{total charge}}\longrightarrow 0\cdots,\label{MVsequence}
\end{equation}
Here $H^{d-1}(S_W)=\oplus_{w\in W} H^{d-1}(S_w)$ is the direct sum of the local charge groups for each $w$, and the last MV connecting map $\Sigma$ is the ``total charge operator'' which adds up the local charges \cite{BottTu}. For integer coefficients, $H^{d-1}(S_w)\cong\ZZ\cong H^d(T)$ and the charges are $\ZZ$-valued, but we will also consider $\ZZ_2$ coefficients in Section \ref{sec:torsionsemimetal}.

\subsection{Basic two-band Weyl semimetal in 3D}
Let us be more explicit in the basic 3D situation, with $2\times 2$ Bloch Hamiltonians parametrised by $k\in\TT^3$ and $W=\{w_1,w_2\}$. The MV-sequence reads
\begin{equation}
\cdots 0\rightarrow H^{2}(\TT^3)\xrightarrow{\iota^*} H^{2}(\TT^3\setminus W) 
\xrightarrow{\beta} H^{2}(S^2_{w_1}\coprod S^2_{w_2})\xrightarrow{\Sigma} H^3(\TT^3)\rightarrow 0\ldots
\end{equation}
Consider local charges $(+q,-q)\in H^{2}(S^2_{w_1}\coprod S^2_{w_2})$ which cancel, then there exists (by exactness) a line bundle $\cE_F$ over $\TT^3\setminus W$ whose restriction to $S_{w_i}$ has Chern class given by the local charge at $w_i$.

 \begin{figure} 
 \centering
        \includegraphics[width=0.6\textwidth]{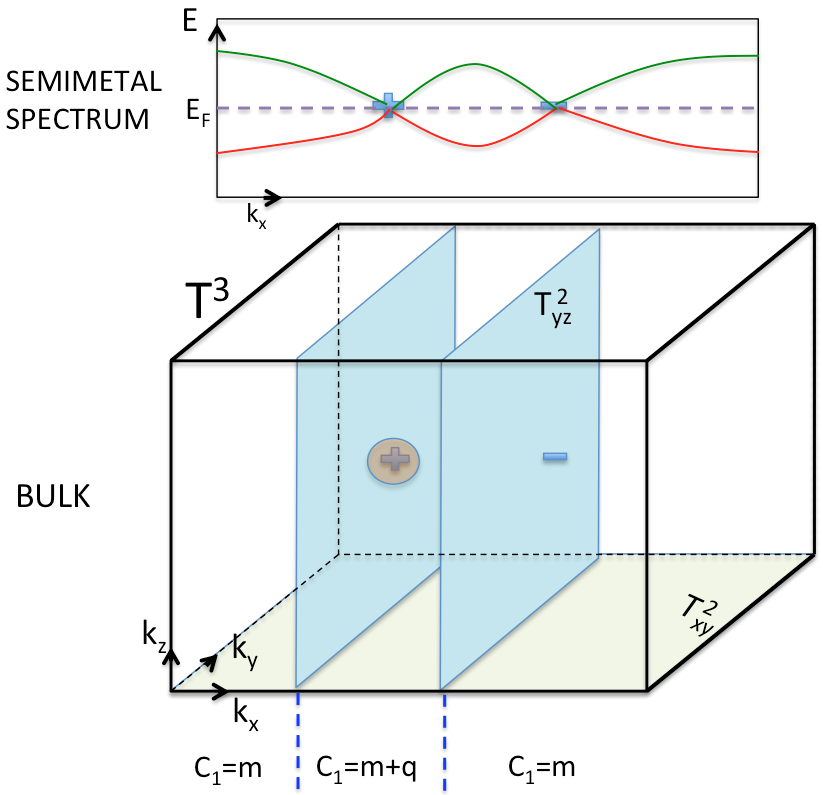}
        \caption{For each $k_x$ away from $W=\{+,-\}$ (with charges $(+q,-q)$), $\mathcal{E}_F$ has a first Chern number $c_1(k_x)\equiv c_1^{yz}(k_x)$ on the 2D subtorus in the $y$-$z$ direction (blue). Consider the region $X$ between the two blue 2-tori, with a small ball around $+$ removed. This is a 3-manifold whose boundary $\partial X$ consists of the two blue 2-tori (oppositely oriented) together with a small 2-sphere around $+$. Note that the top and bottom faces are identified, as are the front and back faces so they do not contribute to $\partial X$. Everywhere on $X$, the valence line bundle is well-defined, and its Berry curvature 2-form $\mathcal{F}$ is closed. By Stokes' theorem, $0=\int_X \text{d}\mathcal{F}=\int_{\partial X}\mathcal{F}$, so the difference between the first Chern numbers for the left 2-torus and the right 2-torus is precisely the first Chern number on the small 2-sphere, which is $+q$. This argument is easily generalised to higher dimensions and higher curvature forms (Remark \ref{remark:higherjumps}).}\label{fig:Chern-jump}

 \end{figure}

{\bf Intuitive idea behind relation between local charges and first Chern numbers of a Weyl semimetal.} Suppose the vector field $\vect{h}$ on $\TT^3$ vanishes at $w_1, w_2$ with local indices $+q,-q$ respectively. For simplicity, we first assume that $w_i$ have the same $k_y, k_z$-coordinates, as in Figure \ref{fig:Chern-jump}. Take a trivial Bloch bundle $\mathscr{S}_0=\TT^3\times\CC^2$, the Pauli matrices $\sigma_1,\sigma_2,\sigma_3$ with respect to a trivialisation, and let $H(k)=\vect{h}(k)\cdot\vect{\sigma}$. Then on $\TT^3\setminus W$, the valence line bundle $\cE_F$ is well-defined. At each $k_x$ away from the Weyl points, $\cE_F$ has first Chern numbers $c_1^{yz}(k_x)$ obtained by pairing $c_1(\cE_F)$ with the fundamental class of the 2-torus at $k_x$ (integrating the Berry curvature over the 2-torus). The integer $c_1^{yz}$ is sometimes called a ``weak invariant'' for ``weak'' topological insulators supported on lower-dimensional subtori of the Brillouin torus. As a function of $k_x$, $c_1^{yz}$ remains constant unless a Weyl point is traversed, whence it jumps by an amount equal to the local charge there. A possible function $c_1^{yz}$ is given at the bottom of Figure \ref{fig:Chern-jump}, and different choices of $\vect{h}$ can give rise to different $m$. 

Note that even with $m$ and the local charges $\pm q$ still do not completely specify the class of $\cE_F$ --- there are two other independent first Chern numbers $c_1^{xy}, c_1^{zx}$ which are constant functions of $k_z$ and $k_y$ respectively (since both Weyl points are traversed simultaneously as either $k_z$ or $k_y$ is varied). In general, the $w_i$ need not have the same $k_y, k_z$-coordinates, and so there are similar jumps in the functions $c_1^{xy}(k_z)$ and $c_1^{zx}(k_y)$, see Figures \ref{fig:Intersection}-\ref{fig:Intersection2}. There is a simple way to concisely and invariantly (i.e.\ coordinate-free) capture all the jumps in the various first Chern numbers, through a Poincar\'{e} dual picture in terms of Euler chains (see Section \ref{sec:semimetalEulerchain} and Remark \ref{remark:higherjumps}).

{\bf All first Chern numbers can appear in a Weyl semimetal.}
Let us also consider $\cE'_F$ which arise as a subbundle of a possibly non-trivial Bloch bundle $\mathscr{S}$. Such $\cE'_F$ are easily obtained by tensoring $\mathscr{S}_0$ with a line bundle $\cL$ with $c_1(\cL)\in H^2(\TT^3)$. The ``twisted'' Hamiltonian $H'(k)=(\vect{h}(k)\cdot\vect{\sigma})\otimes 1$ acts on the ``twisted'' Bloch bundle $\mathscr{S}=\mathscr{S}_0\otimes \cL$, and has valence line bundle $\cE'_F=\cE_F\otimes\cL$ (defined over $\TT^3\setminus W$). Therefore $c_1(\cE'_F)=c_1(\cE_F)+\iota^*(c_1(\cL))$ and, in particular, the first Chern numbers $c_1^{yz}(\cE'_F)$ increases by $c_1^{yz}(\cL)$ and similarly for $c_1^{xy}, c_2^{zx}$. In this way, all elements of $H^2(\TT^3\setminus W)$ can be realised as valence line bundles of some semimetallic Bloch Hamiltonian.

\subsection{Geometric interpretation of MV sequence}\label{sec:extension}
In $d=3$, the exact sequence \eqref{MVsequence} has a direct interpretation as a solution to a geometrical extension problem, which we can interpret in terms of insulator-semimetal transitions. Let us recall that $H^2(T,\ZZ)$ classifies line bundles over $T$ (e.g.\ determinant valence bundles). So for $T=\TT^3$, the exactness of \eqref{MVsequence} tells us the following:
\begin{enumerate}
\item A collection of line bundles $\cL_w\rightarrow S_w$ with Chern classes $c_w$ may be extended to a line bundle over $T\setminus W$ (comes from a semimetal) iff $\sum_{w\in W}c_w=0$ (local charges cancel).
\item A line bundle $\cL\rightarrow T\setminus W$ extends to all of $T$ iff each local charge vanishes. Thus a semimetal (whose valence bundle represents a class in $H^2(T\setminus W)$) can be gapped into an insulator iff all the local charges are zero. A \emph{topological semimetal} with a topologically protected crossing at $w$ has non-zero local charge there.
\item The insulator invariants in $H^2(T)$ (``weak'' Chern insulators) appear faithfully in $H^2(T\setminus W)$. Thus two semimetals with the same local charge (same image under $\beta$) can be \emph{globally inequivalent}, differing by some element in (the image under $\iota^*$ of) $H^2(T)$ (the weak invariants).
\end{enumerate}

An important consequence of Point 3 is that the local charge information is not sufficient to determine the global band structure of the semimetal. As explained in \cite{MTFermi}, this ambiguity has crucial consequences for the determination of the resulting surface Fermi arcs. We will have more to say about this in Section \ref{sec:BBC}, see also Figure \ref{fig:Bulk2boundary}.

More generally, in $d\geq 3$, a Dirac-type Bloch Hamiltonian in the sense of Definition \ref{defn:abstractHamiltonian} determines by restriction smooth maps $\hat{\vect{h}}_w: S_w \longrightarrow S^{d-1}, w\in W$, and the semimetal MV sequence can be interpreted at this level. In Appendix \ref{appendix} we describe the relevant geometric structures on $S^{d-1}$, $d=3,4,5$, and these can be pulled back to $S_W$ via all the $\hat{\vect{h}}_w$. 
\begin{enumerate}
\item In $d=3$, let $\cL$ be the Hopf line bundle over $S^2$. Consider a collection of pullback line bundles $\hat{\vect{h}}_w^*(\cL)$ over $S_w$, such that $\sum_{w\in W} c_1(\hat{\vect{h}}_w^*(\cL))=0=
\sum_{w\in W} \deg(\hat{\vect{h}}_w)$, where $\deg(\hat{\vect{h}}_w)$ denotes the degree of the map $\hat{\vect{h}}_w$. Although they individually may not extend to $T\setminus \{w\}$, together $\coprod_{w\in W}  \hat{\vect{h}}_w^*(\cL)$ over $S_W$  does extend to $T\setminus W$ by exactness of the MV sequence. See the Appendix for more information.

\item In $d=4$, the $H^3$ part of the semimetal Mayer--Vietoris sequence 
\begin{equation}
\cdots 0\rightarrow H^3(T)\xrightarrow{\iota^*} H^3(T\setminus W) 
\xrightarrow{\beta} H^3(S^3_W)\xrightarrow{\Sigma} H^4(T)\rightarrow 0\ldots\label{MV4D}
\end{equation}
has the following analogous geometric interpretation. We refer the reader to Appendix \ref{appendix} for a primer on gerbes and the construction of the basic gerbe on $\text{SU}(2)$.
A collection of pullback basic gerbes $\cG_w\rightarrow S_w$ with Dixmier--Douady numbers $DD_w$ may be extended to a gerbe over $T\setminus W$ (comes from a gerbe semimetal) iff $\sum_{w\in W}DD_w=0$ (local DD charges cancel).
The pullback basic gerbe $\cG\rightarrow T\setminus W$ extends to all of $T$ iff each local DD charge vanishes. Thus a gerbe semimetal (with topological invariant in $H^3(T\setminus W)$) can be gapped into an insulator iff all the local DD charges are zero. A \emph{topological gerbe semimetal} with a topologically protected crossing at $w$ has non-zero local DD charge there.
The insulator invariants in $H^3(T)$ (``weak'' gerbe insulators) appear faithfully in $H^3(T\setminus W)$. Thus two gerbe semimetals with the same local DD charge (same image under $\beta$) can be \emph{globally inequivalent}, differing by some element in (the image under $\iota^*$ of) $H^3(T)$. 
Conversely, the pullback gerbes $\hat{\vect{h}}_w^*(\cG)$ over $S_w$, where $\cG$ is the basic gerbe over $S^3\cong \text{SU}(2)$, such that $\sum_{w\in W} DD(\hat{\vect{h}}_w^*(\cG))=0=
\sum_{w\in W} \deg(\hat{\vect{h}}_w)$. Although the $\hat{\vect{h}}_w^*(\cG)$ may not individually extend to $T\setminus \{w\}$, the gerbe $\coprod_{w\in W} \hat{\vect{h}}_w^*(\cG)$  does extend to $T\setminus W$ by exactness of the MV sequence. See Section
\ref{sec:semimetalgerbe} and the Appendix for more details.

\item In $d=5$, the analysis is analogous to the case of line bundles done above. One replaces the line bundle $\cL$ over $\CC \PP^1$ by the quaternionic line bundle $\cH$ over $\HH \PP^1$ and one uses the 
MV sequence in \eqref{MV5D}, where the first Chern class is replaced by the 2nd Chern class. See  Section \ref{sec:5Dsemimetal} and the Appendix for more details.


\end{enumerate}

Although there is a MV sequence analysis in the case of Kervaire semimetals (Section \ref{sec:Kervairestructures}), we are unable to formulate a geometric extension problem in this context (see Section \ref{sec:Outlook}).

\subsubsection{Adiabatic connections}
In Section \ref{sec:2DBerry}, we briefly alluded to the relation between Chern numbers and Berry connections associated to valence subbundles for a family of Hamiltonians over $T$. In general, a rank-$n$ eigen-subbundle has a $\text{U}(n)$ adiabatic (non-abelian Berry) connection, whose curvature gives the real (de Rham) Chern classes of the subbundle \cite{Simon}. For $n=1$, the curvature 2-form is involved in computing first Chern numbers. In the presence of time-reversal symmetry, the subbundle is complex even-dimensional and in fact a quaternionic symplectic bundle with quaternionic dimension $m=\frac{n}{2}$; the Berry connection becomes a $\text{Sp}(m)$ one, see \cite{ASSS,ASSS2,Hatsugai} and Section \ref{sec:quaternionbundle}. For $m=1$, the 4-form obtained from squaring the curvature enters in the computation of second Chern numbers. In a similar vein, the additional imposition of a chiral symmetry allows a ``semimetal gerbe'' to be associated to the family of Hamiltonians, see Section \ref{sec:semimetalgerbe}. There is a ``Berry connection 2-form'' with 3-form curvature, which computes the Dixmier--Douady invariants of the gerbe.

\subsection{Dual MV sequence and Fermi arcs}\label{sec:dualMV}
Fermi arcs are more naturally analyzed in homology, and for this reason, we Poincar\'{e}-dualise \eqref{MVsequence}
\begin{equation}
\cdots 0\rightarrow \underbrace{H_1(T)}_{\text{Euler structures}}\rightarrow \underbrace{H_1(T,W)}_{\text{Euler chains}}\overset{\partial}{\rightarrow} \underbrace{H_0(S_W)\cong H_0(W)}_{\text{local charges}}\rightarrow \underbrace{H_0(T)}_{\text{total charge}}\rightarrow 0\cdots,\label{dualMVsequence}
\end{equation}
with the physical meanings of the first two groups explained in Section \ref{sec:Euler}. Roughly speaking, the Poincar\'{e} dual of a semimetal invariant in $H^2(T\setminus W)$ is an Euler chain in the \emph{relative} homology group $H_1(T,W)$, whose boundary is the 0-chain of local charge data.

Equation \eqref{dualMVsequence} can also be exhibited directly as a Mayer--Vietoris sequence for \emph{Borel--Moore} homology $H^{\text{BM}}_*$ \cite{BM,Iversen}. The groups $H^{\text{BM}}_*(\cdot)$ may be defined for locally compact spaces $X$ using \emph{locally finite} chains. For compact spaces, $H^{\text{BM}}_*$ and the usual (singular) homology $H_*$ coincide. However, $H^{\text{BM}}_0(D_w)=0$ whereas $H_0(D_w)\cong\ZZ$, for $D_w$ an open ball containing $w$. There are isomorphisms $H_1(T,W)\cong H^{\text{BM}}_1(T\setminus W)$, $H_*(X)\cong H^{\text{BM}}_{*+1}(X\times\RR)$, Poincar\'{e} duality $H^{d-1}(T)\cong H^{\text{BM}}_1(T\setminus W)$, and restriction maps to $H^{\text{BM}}_*$ of open subsets, see Chapter 2.6 of \cite{Chriss}. The Poincar\'{e} dual of \eqref{MVsequence} written in terms of $H^{\text{BM}}_*$ is
\begin{equation}
\cdots 0\longrightarrow H^{\text{BM}}_1(T)\longrightarrow H^{\text{BM}}_1(T\setminus W) \longrightarrow H^{\text{BM}}_1(S_W\times (-\epsilon,\epsilon)) \longrightarrow H^{\text{BM}}_0(T)\longrightarrow 0\cdots,\label{dualMVsequenceBM}
\end{equation}
which is precisely the MV-sequence for $H^{\text{BM}}_*$ (IX.2.3 of \cite{Iversen}), with respect to the open cover $\{D_W, T\setminus W\}$ of $T$; c.f.\ the ``localization'' long exact sequence, IX.2.1 of \cite{Iversen} and 2.6.10 of \cite{Chriss}. The sequences \eqref{dualMVsequenceBM} and \eqref{dualMVsequence} are the same on account of $H^{\text{BM}}_1(S_W\times (-\epsilon,\epsilon))\cong  H_0(S_W)$. The language of Borel--Moore homology has certain advantages, but for this paper we stay in the more elementary setting of relative homology so as to avoid having to introduce too much technical machinery.

In \cite{MTFermi}, we explained how Fermi arcs represent \emph{relative homology} classes in $H_1(\widetilde{T},\widetilde{W})$, where $\pi:(T,W)\rightarrow(\widetilde{T},\widetilde{W})$ projects out one torus direction, so $\widetilde{T}$ is the surface Brillouin zone. The Euler chain representing the topological data of a semimetal is projected onto a surface Fermi arc, and this is the Poincar\'{e} dual picture of the bulk-boundary correspondence, see Section \ref{sec:BBC}.


\section{Euler chains and semimetal Hamiltonians}\label{sec:Euler}
Since we parametrise Dirac-type Bloch Hamiltonians by vector fields $\vect{h}$ on $T$, we would also like to analyse the semimetal MV-sequence \eqref{MVsequence} in terms $\vect{h}$ and its homotopies. In the insulating case, we might want to maintain the gap condition and so consider homotopies through non-singular (i.e.\ nowhere vanishing) vector fields. In the semimetal case, we want to at least maintain the gap condition everywhere except perhaps at a finite set of isolated Weyl points. We may require the set of Weyl points and their charges to be kept fixed, or we may allow them to move around and be created/annihilated in pairs. In Section \ref{sec:Eulerheuristic}, we explain why it is important to keep track of the ``history'' of the creation/annihilation of Weyl points as part of the topological data of a semimetal band structure. This motivates the mathematical notion of Euler structures, which we introduce in Section \ref{sec:smoothEuler}, as well as the dual notion of Euler chains for keeping track of Weyl point connectivity, introduced in Section \ref{sec:Eulerchain}. The Euler chain representation of a semimetal is then explained in Section \ref{sec:semimetalEulerchain}.

\subsection{Weyl pair creation and annihilation}\label{sec:Eulerheuristic}
Heuristically, a pair of Weyl points with opposite charges $+1$ and $-1$ may be created at a particular point in $T$ where a valence and conduction band cross. Then the Weyl points are moved apart in $T$ to form two separate band crossings, each of which is topologically protected and cannot be gapped out. In the reverse direction, if two separate Weyl points with opposite charges come together at a single point, their charges annihilate and a gap can be opened at that point. What is particularly interesting is that a $T$ with some non-trivial 1-cycle (such as $T=\TT^3$) admits a scenario in which a pair of of Weyl points is created at a single point, moved apart in opposite directions along the 1-cycle, and then annihilated when they meet again. This creation/annihilation process is \emph{global}, and the history traces out the 1-cycle (Fig.\ \ref{fig:Zero-creation}-\ref{fig:Zero-annihilation}). 

 \begin{figure}
 \centering
        \includegraphics[width=0.7\textwidth]{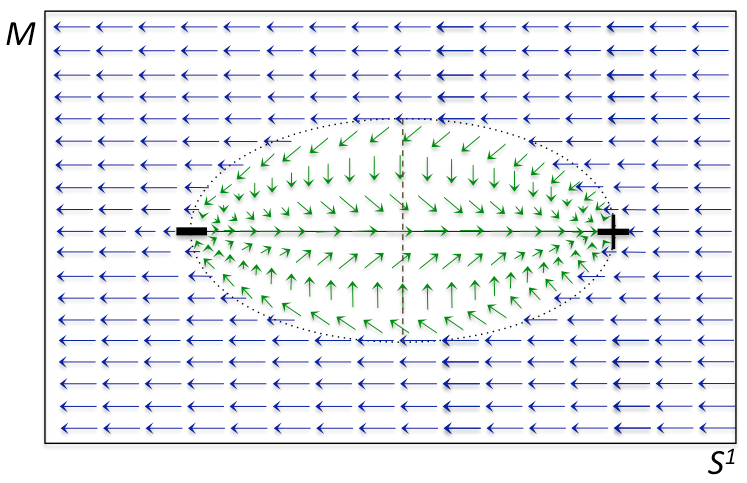}
        \caption{A typical way to locally introduce a pair of opposite-degree zeroes into a vector field. Here, the compact oriented $d$-manifold $T$ is $S^1\times M$, and we perturb an initial global non-vanishing vector field (blue) by modifying it within a ball containing the desired zeroes (black $\pm$ symbols). The region enclosed by the dotted line is a tubular neighbourhood to the open black line connecting $-$ to $+$, so that each cross-section (dashed line) is an open ball $D^{d-1}$ in $M$. For each cross-section, we assign tangent directions (green arrows, elements of $S^{d-1}$) to the closed ball $\overline{D}^{d-1}$ by the surjective map $\overline{D}^{d-1}\rightarrow S^{d-1}$ which collapses $\partial\overline{D}^{d-1}$ to the point corresponding to the blue arrow. On a small $d-1$-sphere $S^{d-1}_{\pm}$ surrounding $\pm$, the assignment of tangent directions is a degree $\pm 1$ map $S^{d-1}_{\pm}\rightarrow S^{d-1}$ as required. This is essentially the Pontryagin--Thom construction, see \cite{MilnorDT}.}\label{fig:Zero-creation}
 \end{figure}
 
 \begin{figure}
 \centering
        \includegraphics[width=0.7\textwidth]{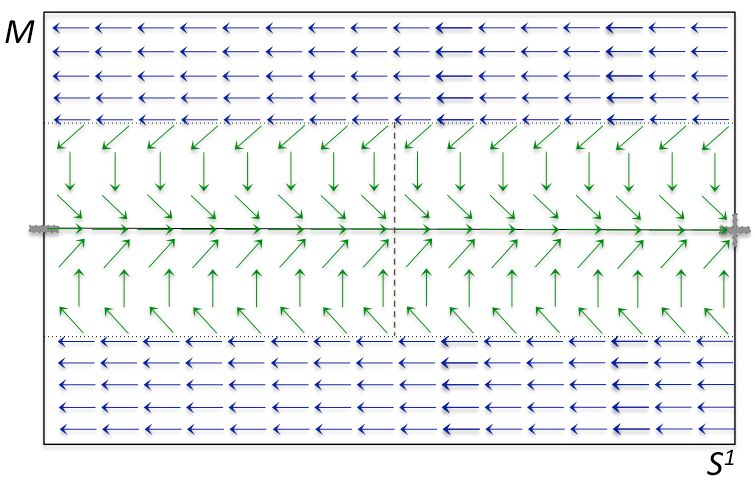}
        \caption{If the locally created zeroes of Fig.\ \ref{fig:Zero-creation} are moved apart in opposite directions along the cycle $S^1$, and then annihilated when they meet (grey $\pm$ symbols), the result is a global vector field modified (as in the green arrows) from the initial one (blue arrows) within a tubular neighbourhood $S^1\times D^{d-1}$ to the trajectory of the zeroes.}\label{fig:Zero-annihilation}
 \end{figure}

Even though both the initial and final band structures have no band crossings, they can generally become topologically inequivalent during this process. For example, take $T=\TT^3$ (so $M=\TT^2$) in Fig.\ \ref{fig:Zero-creation}-\ref{fig:Zero-annihilation}. The initial constant vector field $\vect{h}$ over $\TT^3$ corresponds to an initial gapped Hamiltonian $H$ with trivial valence band $\hat{\vect{h}}^*(\cL)$, where $\cL$ is the Hopf line bundle over $S^2\cong\CC\PP^1$ as in Sec.\ \ref{sec:extension}. On the other hand, the final vector field $\vect{h}'$ is such that the unit vector map $\hat{\vect{h}'}:\TT^3\rightarrow S^2$ restricts to a degree-1 map $\TT^2\rightarrow S^2$ on each cross-section $\TT^2$ to the Weyl points' trajectory $S^1$. Thus the valence band of the final gapped Hamiltonian $H'$, which is $\hat{\vect{h}'}^*(\cL)$, ``acquires'' a weak Chern class supported in the $\TT^2$ directions (which is Poincar\'{e} dual to $S^1$). We learn that a local Weyl point creation-annihilation process does not change the weak invariants of a Hamiltonian, but a global one can. This motivates, and provides a physical interpretation of, a certain notion of equivalence of vector fields (thus of Hamiltonians) defined formally in the next subsection. 

\subsection{Euler structures}\label{sec:smoothEuler}

{\bf Smooth Euler structures.} Let $T$ be a compact oriented $d$-manifold with $d\geq 2$, and suppose $T$ has Euler characteristic $\chi(T)=0$ (this is automatic in $d=3$) so that it has a global non-vanishing vector field. 
\begin{definition}[\cite{TuraevEuler}]\label{defn:smootheuler}
Two non-singular smooth vector fields $\vect{h}, \vect{h}'$ on $T$ are said to be \emph{homotopic} (also called \emph{homologous}\footnote{The stronger notion of homotopy on $T$ through non-singular vector fields can also be analysed \cite{TuraevEuler}, but we do not consider this in this paper.} in \cite{TuraevEuler}) if for some open ball $D$, the fields $\vect{h}, \vect{h}'$ are homotopic on $T\setminus D$ in the class of non-singular vector fields on $T\setminus D$. The set $\mathfrak{vect}(T)$ of homotopy classes of non-singular vector fields is called the set of \emph{smooth Euler structures} on $T$.
\end{definition}
As explained in Section 5 of \cite{TuraevEuler}, the first obstruction to such a homotopy between $\vect{h},\vect{h}'$ is an element (written as $\vect{h}/\vect{h}'$) of $H^{d-1}(T,\pi_{d-1}(S^{d-1}))\cong H^{d-1}(T,\ZZ)$, canonically isomorphic to $H_1(T,\ZZ)$ by Poincar\'{e} duality. It turns out that there is a natural free and transitive action of the group of obstructions $H^{d-1}(T,\ZZ)\cong H_1(T,\ZZ)$ on $\mathfrak{vect}(T)$. 

\begin{remark}
If we pick a reference nonsingular $\vect{h}_{\text{ref}}$ as the zero, then we can identify $\mathfrak{vect}(T)$ with $H^{d-1}(T,\ZZ)$. For $T=\TT^d$, a natural choice for $\vect{h}_{\text{ref}}$ is a constant length vector field pointing along a torus cycle (this gives a trivial insulating phase); reference to such a $\vect{h}_{\text{ref}}$ will be implicit in this case. We remark that the space of spin${}^c$ structures for $T$ is a $H^2(T,\ZZ)$-torsor, and that when $d=3$, a non-singular $\vect{h}_{\text{ref}}$ determines a spin${}^c$ structure through the unit vector field $\hat{\vect{h}}_{\text{ref}}$ \cite{TuraevSpinc}, and vice-versa. 
\end{remark}

\subsubsection{Another picture of Euler structures}\label{sec:coEuler}
An equivalent definition of Euler structures, which is useful for our generalisation to Kervaire structures in Section \ref{sec:Kervairestructures}, was given in \cite{TuraevSpinc} as follows. 
\begin{definition}\label{defn:coEuler}
Let $T$ be a compact oriented $d$-manifold with $\chi(T)=0$, and let $\cS \equiv\cS(\cE)$ be the unit sphere bundle for the tangent bundle $\cE\rightarrow T$ (with a Riemannian metric). Define the set of \emph{Euler structures} to be the subset  $\widetilde{\coeul}(T)\subset H^{d-1}(\cS ,\ZZ)$ for which the restriction to each (oriented) fibre $\cS _k\cong S^{d-1}$ generates $H^{d-1}(S^{d-1},\ZZ)$. 
\end{definition}
Then there is a free and transitive action of $H_1(T,\ZZ)\cong H^{d-1}(T,\ZZ)$ on $\wt{\coeul}(T)$ under pullback to $\cS $ and addition in $H^{d-1}(\cS ,\ZZ)$. To pass from the first picture of smooth Euler structures as $\mathfrak{vect}(T)$ to this second picture, notice that the unit vector field $\hat{\vect{h}}$ for a non-singular $\vect{h}$ is a map $T\rightarrow \cS$ and so a $d$-cycle in $\cS $. Orient $\cS $ by requiring the intersection number of every $\cS_k$ with this cycle to be $+1$. Then the $d$-cycle $\hat{\vect{h}}$ Poincar\'{e} dualises to a $(d-1)$-cocycle which we write as $\text{PD}_{\vect{h}}$, and the latter represents an Euler structure in $\wt{\coeul}(T)\subset H^{d-1}(\cS,\ZZ)$ in the second picture. To emphasize the cohomology definition of Euler structure, we will sometimes use the term \emph{co-Euler structure}, to distinguish it from a homological definition to be given in Section \ref{sec:Eulerchain}.

As with $\mathfrak{vect}(T)$, there is an identification of the affine space $\widetilde{\coeul}(T)$ with $H^{d-1}(T,\ZZ)$ upon fixing a reference non-singular $\vect{h}_{\text{ref}}$. Writing $p:\cS \rightarrow T$ for the projection, we note that $\hat{\vect{h}}_{\text{ref}}^*\circ p^*=\text{id}$ in the Gysin sequence
\begin{equation}
	0\longrightarrow H^{d-1}(T,\ZZ)\overset{p^*}{\underset{\hat{\vect{h}}_{\text{ref}}^*}{\rightleftarrows}}H^{d-1}(\cS ,\ZZ)\overset{p_*}{\longrightarrow}H^0(T,\ZZ)\longrightarrow 0,\label{gysin}
\end{equation}
with $p_*$ the pushforward, or integration over the $S^{d-1}$ fibers.

The subset $\wt{\coeul}(T)\subset H^{d-1}(\cS ,\ZZ)$ comprises $[\text{PD}_{\vect{h}_{\text{ref}}}]$ and all its translates by $p^*(H^{d-1}(T,\ZZ))$, so $\hat{\vect{h}}_{\text{ref}}^*$ gives a bijection from $\wt{\coeul}(T)\rightarrow H^{d-1}(T,\ZZ)$ taking $[\text{PD}_{\vect{h}_{\text{ref}}}]$ to the identity element.

\subsection{Euler chains and singular vector fields}\label{sec:Eulerchain}
There is a dual picture of Euler structures involving vector fields with finite singularities, which is closely related to the intuition behind Fermi arcs. We sketch the construction here for the case of non-degenerate zeroes, see \cite{Hutchings,Burg,Molina} for the general case.

For a vector field $\vect{h}$ on $T$ with finite singularity set $W_{\vect{h}}$, let $\cW_{\vect{h}}$ denote the singular 0-chain $\cW_{\vect{h}}\coloneqq \sum_{w\in W_{\vect{h}}} \text{Ind}_w(\vect{h}) w\in C_0(T,\ZZ)$ which encodes the local charge information. We can also think of $\cW_{\vect{h}}$ as an element of $H_0(W_{\vect{h}},\ZZ)$ whose weights sum to zero. 
\begin{definition}[\cite{Hutchings}]
Let $T$ be a compact $d$-manifold with $d\geq 2$ and $\chi(T)=0$. An \emph{Euler chain} for a 0-chain $\cW$ in $T$ whose weights sum to zero, is a 1-chain $l\in C_1(T,\ZZ)$ such that $\partial l$ is $\cW$. An Euler chain $l$ defines a relative homology class in $H_1(T,W,\ZZ)$ where $W$ is the set of points that $\cW$ is defined on\footnote{$\cW$ is allowed to be zero on points in $W$.}. Let $\eul(T,\cW)\subset H_1(T,W,\ZZ)$ denote the subset of relative homology classes of 1-chains in $T$ whose boundary is $\cW$. 
\end{definition}
From the exact sequence
\begin{equation}
0=H_1(W,\ZZ)\rightarrow H_1(T,\ZZ)\rightarrow H_1(T,W,\ZZ)\overset{\partial}{\rightarrow} H_0(W,\ZZ)\overset{\Sigma}{\rightarrow} H_0(T,\ZZ)\rightarrow H_0(T,W,\ZZ)=0,\nonumber
\end{equation}
we see that $\eul(T,\cW)$ is a coset of $H_1(T,\ZZ)$ in $H_1(T,W,\ZZ)$ and so an affine space for $H_1(T,\ZZ)$. In particular, the difference of $[l],[l']\in\eul(T,\cW)$ is an element of $H_1(T,\ZZ)$. Note that any 0-chain $\cW$ with total weight zero can be realised as the 0-chain of charges $\cW_{\vect{h}}$ for some singular vector field $\vect{h}$, by the (converse of the) Poincar\'{e}--Hopf theorem. Euler chains can be thought of as extra global data encoding how the local charge configuration for a singular vector field is ``connected''. As explained in Section \ref{sec:semimetalEulerchain}, Euler chains are in a precise sense Poincar\'{e} dual to cohomological semimetal invariants. 

{\bf Non-degenerate homotopies.} Let $\vect{h}, \vect{h}'$ be two smooth vector fields, then there is a canonical identification of $\eul(T,\cW_{\vect{h}})$ and $\eul(T,\cW_{\vect{h}'})$ as affine spaces, using a notion of \emph{non-degenerate homotopy} defined as follows. Let $p^*\cE$ be the pullback of the tangent bundle $\cE\rightarrow T$ under the projection $p:[0,1]\times T\rightarrow T$. A \emph{non-degenerate homotopy} between $\vect{h}$ and $\vect{h}'$ is a section $\breve{\vect{h}}$ of $p^*\cE$ transverse to the zero section, which restricts to $\vect{h}$ at $t=0$ and $\vect{h}'$ at $t=1$. Such a homotopy exists by perturbing, for instance, a linear homotopy, and allows for the movement $t\mapsto W_{\vect{h}_t}$ of the local charges $W_{\vect{h}_t}\coloneqq\vect{h}_t^{-1}\{0\}$ including the creation and annihilation of pairs of zeros with equal and opposite charges (Fig.\ \ref{fig:Euler-chain}). For generic $t$, the intermediate vector field $\vect{h}_t$ intersects the zero vector field transversally (e.g.\ transversality of $\vect{h}_t$ fails when a pair of Weyl points are created or annihilated). The zero set $\breve{\cW}$ of $\breve{\vect{h}}$ is a canonically oriented 1-submanifold with boundary $\cW_{\vect{h}'}^{(1)}-\cW_{\vect{h}}^{(0)}$, where the superscript indicates whether the 0-chain lies on $\{0\}\times T$ or on $\{1\}\times T$. Let $[l]\in \eul(T,\cW_{\vect{h}})$, then $\breve{\cW}+l$ is a 1-cycle on $[0,1]\times T$ relative to $\{1\}\times T$, since $\partial(\breve{\cW}+l)=\cW_{\vect{h}'}^{(1)}$. Because $H_1([0,1]\times T, \{1\}\times T)=0$, there is a 2-chain $P$ on $[0,1]\times T$ whose boundary is $\breve{\cW}+l$ relative to $\{1\}\times T$. Taking $l'\coloneqq \breve{\cW}+l-\partial P$, we see that $l'$ is a 1-chain on $\{1\}\times T$ with $\partial l'=\cW_{\vect{h}'}^{(1)}$, which defines a class $[l']\in \eul(T,\cW_{\vect{h}'})$. Intuitively, we can think a homotopy $\breve{\vect{h}}$ as providing the data of how an initial $l$ is to be ``carried along $\breve{\cW}$'' onto a final $l'$.

 \begin{figure}
 \centering
        \includegraphics[width=0.7\textwidth]{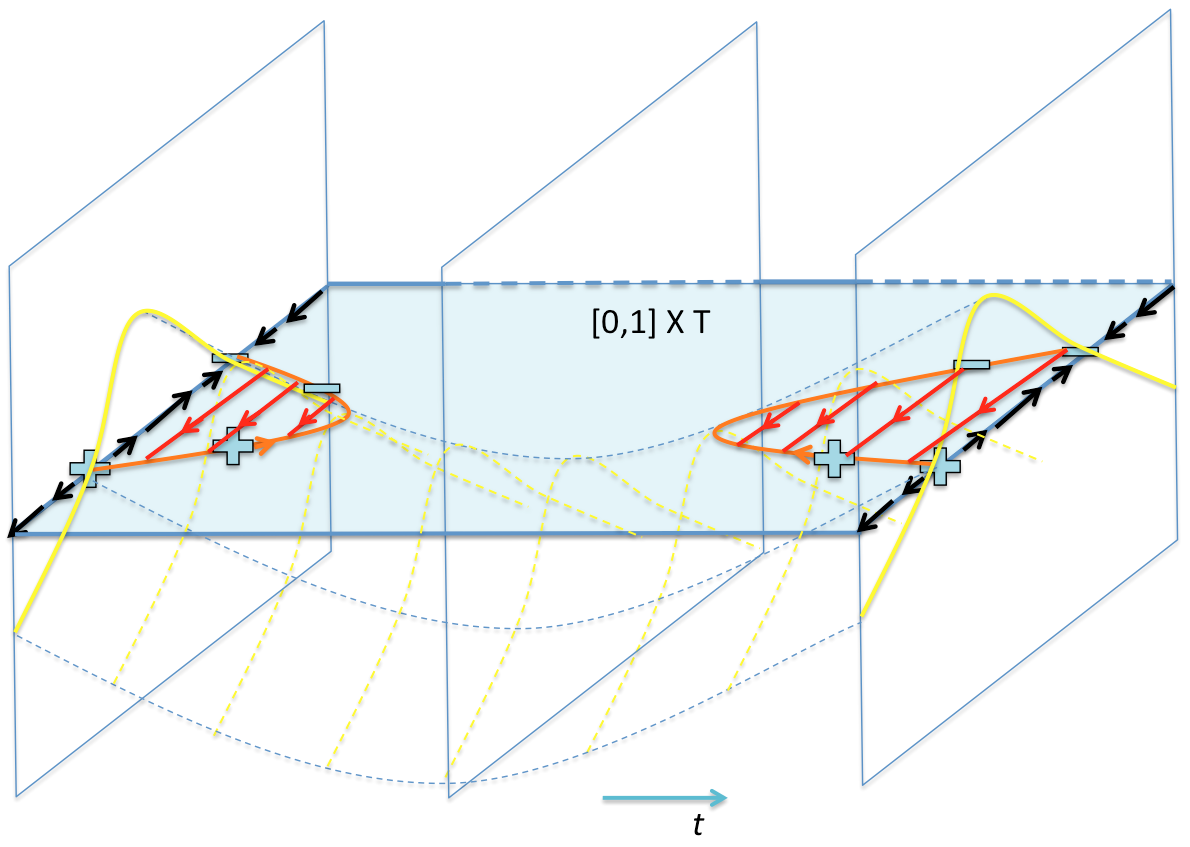}
        \caption{The homotopy parameter $t\in [0,1]$ goes horizontally, while all coordinates for the manifold $T$ are suppressed except one cyclic $S^1$-direction. The height function at each $t$ (yellow curve) indicates the length of a tangent vector on $\{t\}\times T$ in the aforementioned direction. Black arrows represent one component of the initial (left) vector field $\vect{h}$ and final (right) one $\vect{h}'$, which are taken to be equal in this example. The zero-section of the tangent bundle of $[0,1]\times T$ is the blue shaded area. A non-degenerate homotopy $\breve{\vect{h}}$ is shown, with the intermediate vector fields $\vect{h}_t$ (indicated by yellow dashed curves) generically having isolated zeroes with charges $-, +$. The orange oriented submanifold $\breve{\cW}$ comprises the singularities of $\breve{\vect{h}}$. The Euler chains for the $\vect{h}_t$ (red lines joining $-$ to $+$) are ``carried along'' $\breve{\cW}$, disappear when the Weyl points coalesce, then reappear when a pair of Weyl points are created. Under this homotopy, the homology class of the final Euler chain (for $\vect{h}'$) remains the same as that of the initial Euler chain (for $\vect{h}$), consistent with what happens if the obvious constant homotopy is chosen instead.}\label{fig:Euler-chain}
 \end{figure}

 \begin{figure}
 \centering
        \includegraphics[width=0.7\textwidth]{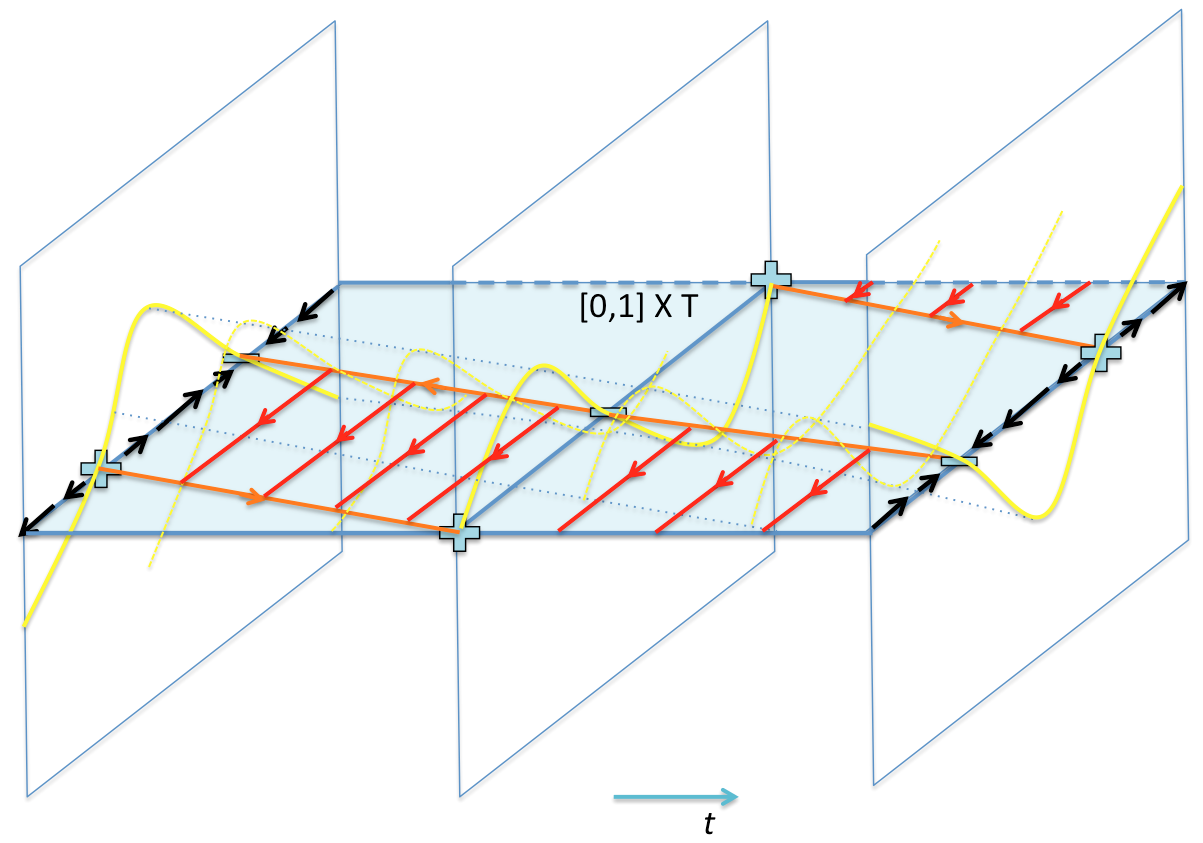}
        \caption{A homotopy $\breve{\vect{h}}$ which translates the initial vector field $\vect{h}$ by half a $S^1$-cycle, showing how the Euler chain is carried along the zero set (orange) of $\breve{\vect{h}}$. Rotation by a full cycle will return the Euler chain to its original position, with the same homology class.}\label{fig:Euler-chain2}
 \end{figure}

 \begin{figure}
 \centering
        \includegraphics[width=0.7\textwidth]{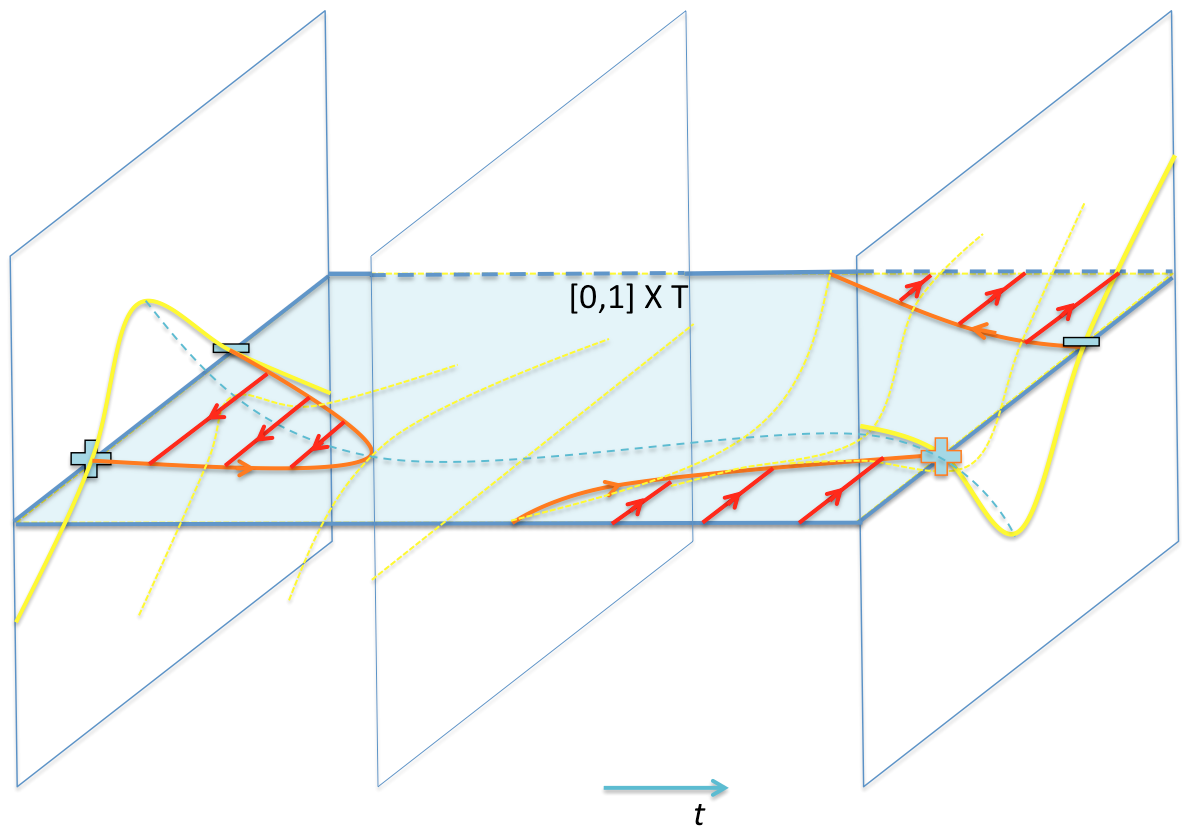}
        \caption{A homotopy $\breve{\vect{h}}$ which changes the homology class of an Euler chain. Note that there is no open ball of $T$ within which the singular set of $\vect{h}_t$ stays for all $t\in[0,1]$. Rather, such an ``enveloping open set'' contains an $S^1$ cycle.}\label{fig:Euler-chain3}
 \end{figure}

One shows \cite{Hutchings} that the above assignment $[l]\mapsto[l']$ is independent of the choices of homotopy and of $\Sigma$, descending to an affine map $\phi_{\vect{h}',\vect{h}}:\eul(T,\cW_{\vect{h}})\rightarrow \eul(T,\cW_{\vect{h}'})$. Furthermore, one has the properties $\phi_{\vect{h}'',\vect{h}}=\phi_{\vect{h}'',\vect{h}'}\circ\phi_{\vect{h}',\vect{h}}$, and $\phi_{\vect{h},\vect{h}}=\text{id}$, so each $\eul(T,\cW_{\vect{h}})$ is canonically isomorphic to a single affine space (which for oriented $T$ can be taken to be the space of smooth Euler structures $\mathfrak{vect}(T)$ in Definition \ref{defn:smootheuler}) over $H_1(T,\ZZ)$ \cite{Hutchings}. In particular, when a non-singular $\vect{h}_{\text{ref}}$ in $\mathfrak{vect}(T)$ has been chosen, a possibly singular $\vect{h}$ has an associated Euler chain $l$.

\begin{example}
If there is a $\breve{\vect{h}}$ such that the local charges $\cW_t$ are simply moved around smoothly and disjointly without creation/annihilation, i.e.\ there is a smooth 1-parameter family of diffeomorphisms $\varphi_t$ of $T$ taking $\cW_0$ to $\cW_t$, $t\in[0,1]$, then the zero set $\breve{W}$ of $\breve{\vect{h}}$ is simply a set of disjoint lines representing the trajectories of the local charges. An initial Euler chain $l$ is ``carried along $\breve{W}$'' to the final Euler chain $l'$ (Fig.\ \ref{fig:Euler-chain2}).
\end{example}

Suppose $\vect{h}, \vect{h}'$ have the same 0-chains of local charges, $\cW=\cW_{\vect{h}}=\cW_{\vect{h}'}$. Then the endomorphism $\phi_{\vect{h}',\vect{h}}:\eul(T,\cW)\rightarrow \eul(T,\cW)$ can be thought of as the homological change in the Euler chains, $[l]\mapsto [l']$, under a non-degenerate homotopy $\breve{\vect{h}}$ taking $\vect{h}$ to $\vect{h}'$. Specifically, $l'-l=\breve{\cW}-\partial P$ can be pushed forward under the projection $p$ to a 1-chain on $T$, also written $l'-l$, which is actually a cycle. Then $[l']-[l]\in H_1(T,\ZZ)$ is a ``homological difference'' between $\vect{h}$ and $\vect{h}'$; this difference was called a Chern--Simons class $\text{cs}(\vect{h},\vect{h}')$ in \cite{Burg,Molina}. For $\vect{h},\vect{h'}$ having different local charges, $\text{cs}(\vect{h},\vect{h}')$ is defined similarly, as a 1-chain modulo boundaries .

\begin{example}
Suppose further that there is a non-degenerate homotopy $\breve{\vect{h}}$ which stays non-singular outside some open ball in $T$ (necessarily containing $W_{\vect{h}}=W_{\vect{h}'}$). Then an initial Euler chain $l$ does not pick up any nonzero element of $H_1(T,\ZZ)$, so $[l']=[l]$ (Fig.\ \ref{fig:Euler-chain}, \ref{fig:Euler-chain3}). 
\end{example}

\subsection{Euler chain representation of a semimetal}\label{sec:semimetalEulerchain}
The upshot of introducing Euler chains is that we can use them to represent globally the topological data of a semimetal band structure, whenever vector fields are used to parameterise semimetal Hamiltonians. This is certainly the case for $2\times 2$ Hamiltonians in 3D, and more generally for Dirac-type Hamiltonians in higher dimensions. The singular points $W_{\vect{h}}$ of $\vect{h}$ are the band crossings of $H$, and the local charge information is contained in the 0-chain $\cW_{\vect{h}}$. Each $\vect{h}$ can be associated with an Euler chain class in $H_1(T,W_{\vect{h}},\ZZ)$ by the prescription of the previous subsection, and vector fields with the same local charges but homologically different Euler chains cannot he homotoped whilst keeping the singular set within an open ball (in particular, while keeping the singular set fixed). When we generalise to Kervaire chains for bilinear Hamiltonians later, it will be convenient to use an alternative prescription to associate an Euler (or Kervaire) chain to a semimetal Hamiltonian. 

This alternative prescription uses a definition of Euler structures which follows that given in Section \ref{sec:coEuler}, but which is defined in terms of \emph{singular} vector fields. Let $\cS$ be the sphere bundle of the tangent bundle of $T$, and $\cS|_{T\setminus W}, \cS|_{D_W}=S^{d-1}\times D_W, \cS|_{S_W}=S^{d-1}\times S_W$ be the restrictions of $\cS$ to the subspaces of $T$ appearing in our MV-sequence, and let $p$ denote the various bundle projections. All these sphere bundles have vanishing Euler class. There is also an MV-sequence for the cover $\{\cS|_{T\setminus W},\cS|_{D_W}\}$ of $\cS$, which has intersection $\cS|_{S_W}$, and we write $\widetilde{\Sigma}$ for its connecting maps. Let $\vect{h}_{\text{ref}}$ be a non-singular vector field on $T$, then $\hat{\vect{h}}_{\text{ref}}^*\circ p^*=\text{id}$.

Consider the semimetal MV-sequence and its dual homology sequence, where we have suppressed the integer coefficients,
\begin{equation*}
    \xymatrix{ 0\ar[r]&H^{d-1}(T)\ar[d]_{\text{PD}}\ar[r]^{\iota^*} & H^{d-1}(T\setminus W)\ar[d]_{\text{PD}}\ar[r]^{\beta} & H^{d-1}(S_W)\ar[d]_{\text{PD}}\ar[r]^{\Sigma} & H^d(T)\ar[d]_{\text{PD}}\ar[r] & 0\\
   0\ar[r] & H_1(T)\ar[r]& H_1(T,W)\ar[r]^{\partial} &  H_0(W) \ar[r]^{\Sigma} & H_0(T)\ar[r] & 0
    }.
\end{equation*}
The local charges can be thought of as a 0-chain $\cW$ on $W\subset T$, or dually as an element $\text{PD}(\cW)$ in $H^{d-1}(S_W)$; in either case they are required to have total charge zero. 
\begin{definition}
For a 0-chain $\cW$ of local charges on $T$, with total charge 0 and defined on $W\subset T$, we write $\coeul(T,\cW)$ for the image of $\eul(T,\cW)$ in $H^{d-1}(T\setminus W)$ under Poincar\'{e} duality; equivalently this is the inverse image of $\text{PD}(\cW)$ under $\beta$.
\end{definition}

We combine the cohomology Mayer--Vietoris (horizontal) and Gysin sequences (vertical) for $T$ and $\cS$,
\begin{equation}
    \xymatrix{     & 0\ar[d]& 0\ar[d] & 0\ar[d] &  \\
     0\ar[r] &H^{d-1}(T)\ar[d]_{p^*}\ar[r]^{\iota^*} &  H^{d-1}(T\setminus W)\ar[d]_{p^*}\ar[r]^{\beta} & H^{d-1}(S_W)\ar[d]_{p^*}\ar[r]^{\quad\Sigma}&  \cdots\\
       0\ar[r] & H^{d-1}(\cS) \ar[r]\ar[d]_{p_*} & H^{d-1}(\cS|_{T\setminus W})\oplus H^{d-1}(\cS|_{D_W}) \ar[r]\ar[d]_{p_*} &  H^{d-1}(\cS|_{S_W})
        \ar[r]^{\quad\widetilde{\Sigma}}\ar[d]_{p_*} &   \cdots\\
   0\ar[r] & H^0(T)\ar[r]\ar[d]& H^0(T\setminus W)\oplus H^0(D_W)\ar[r]\ar[d] &  H^0(S_W) \ar[r]\ar[d] & \cdots\\
    & 0& 0 &  0 & 
    },\label{MVGysin}
\end{equation}
where we have augmented the middle Gysin sequence by a $H^{d-1}(\cS|_{D_W})\xrightarrow{p_*} H^0(D_W)$ piece needed to make the bottom two horizontal MV sequences exact. The horizontal maps are restriction maps (or differences of) and so it is easy to see that the diagram commutes. The middle MV-sequence is, by the K\"{u}nneth theorem, 
\begin{equation}
0\rightarrow H^{d-1}(\cS) \rightarrow H^{d-1}(\cS|_{T\setminus W})\oplus H^{d-1}(\cS|_{D_W}) \rightarrow  H^{d-1}(S_W)\oplus H^0(S_W)
        \xrightarrow{\widetilde{\Sigma}=\widetilde{\Sigma}_1\oplus\widetilde{\Sigma}_2} \cdots\nonumber
\end{equation}
and so the restriction map for the $H^{d-1}(\cS|_{T\setminus W})$ factor has a ``local charge'' component $\tilde{\beta}:H^{d-1}(\cS|_{T\setminus W})\rightarrow H^{d-1}(S_W)$ and a second component which is less important for us. Furthermore the restriction map for the $H^{d-1}(\cS|_{D_W})\cong H^0(D_W)$ factor only lands in the  $H^0(S_W)$ factor. The relevant part of \eqref{MVGysin} is then the commuting diagram of short exact sequences
\begin{equation}
    \xymatrix{     & 0\ar[d]& 0\ar[d] & 0\ar[d] &  \\
     0\ar[r] &H^{d-1}(T)\ar[d]_{p^*}\ar[r]^{\iota^*\;\;} &  H^{d-1}(T\setminus W)\ar[d]_{p^*}\ar[r]^{\beta\;\;} & \text{ker}_\Sigma(H^{d-1}(S_W))\ar[d]_{p^*}\ar[r]&  0\\
       0\ar[r] & H^{d-1}(\cS) \ar[r]^{\iota^*\;\;}\ar[d]_{p_*}\ar[u]_{\hat{\vect{h}}_{\text{ref}}^*} & H^{d-1}(\cS|_{T\setminus W}) \ar[r]^{\tilde{\beta}\;\;}\ar[d]_{p_*}\ar[u]_{\hat{\vect{h}}_{\text{ref}}^*}  &  \text{ker}_{\widetilde{\Sigma}_1}(H^{d-1}(S_W))
        \ar[r]\ar[d] & 0\\
   0\ar[r] & H^0(T)\ar[r]^{\iota^*\;\;}\ar[d]& H^0(T\setminus W)\ar[r]\ar[d] &  0 & \\
    & 0& 0 &   & 
    }.\label{MVGysin2}
\end{equation}

\begin{definition}
Let $\cW$ be a 0-chain of local charges defined on the finite subset $W\subset T$ with total weight zero, and $\text{PD}(\cW)$ its Poincar\'{e} dual. We define $\widetilde{\coeul}(T,\cW)$ to be the set of $u\in H^{d-1}(\cS|_{T\setminus W},\ZZ)$ such that $p_*u=1$ and $\tilde{\beta}u=\text{PD}(\cW)$.
\end{definition}

\begin{proposition}
There is an identification of $\widetilde{\coeul}(T,\cW)$ and $\coeul(T,\cW)$ as affine spaces for $H^{d-1}(T,\ZZ)$.
\end{proposition}

\begin{proof}
By chasing through the diagram \eqref{MVGysin2}, we see that an element $u\in \widetilde{\coeul}(T,\cW)$ is of the form $p^*v+\iota^*s$ for some $v\in H^{d-1}(T\setminus W)$ with $\beta(v)=\text{PD}(\cW)$ and some $s\in H^{d-1}(\cS)$ with $p_*(s)=1$, i.e.\ $s\in\coeul(T,\cW)$ and $v\in\coeul(T)$. Clearly, the action of $H^{d-1}(T)$ on $\widetilde{\coeul}(T,\cW)$ by pullback $(p\circ\iota)^*$ and addition is free, and is also transitive since the difference of $u_1,u_2\in\widetilde{\coeul}(T,\cW)$ is the sum of $\iota^*(s_2-s_1)$ and $p^*(v_2-v_1)$, which is something in $(p\circ\iota)^*(H^{d-1}(T))$. Given a reference non-singular $\vect{h}_{\text{ref}}$, it is easy to check that the surjective map $\hat{\vect{h}}_{\text{ref}}^*: H^{d-1}(\cS|_{T\setminus W})\rightarrow H^{d-1}(T\setminus W)$ respects the local charges $\beta,\widetilde{\beta}$, and restricts to a $H^{d-1}(T)$-equivariant bijection between $\widetilde{\coeul}(T,\cW)$ and $\coeul(T,\cW)$.
\end{proof}

The following diagram summarises the various affine spaces for $H_1(T)\cong H^{d-1}(T)$:
\begin{equation}
\xymatrix{  & H_1(T)\ar[r] \ar[d]& \eul(T,\cW)\ar[l] \ar[d]& \subset \quad\quad\;\; H_1(T,W) \\
 & H^{d-1}(T)\ar[r]\ar[d] \ar[u]& \coeul(T,\cW)\ar[l] \ar[d]\ar[u]& \subset \quad H^{d-1}(T\setminus W) \\
H^{d-1}(\cS)\quad\supset & \widetilde{\coeul}(T)\ar[r] \ar[u]& \widetilde{\coeul}(T,\cW)\ar[l] \ar[u]& \subset \quad H^{d-1}(\cS|_{T\setminus W})
}.\nonumber
\end{equation}

\begin{remark}
If $\cW$ is the zero 0-chain, then $\widetilde{\coeul}(T,0)$ coincides under $\iota^*$ with $\widetilde{\coeul}(T)$.
\end{remark}

Let $\vect{h}$ be a vector field with local charges $\cW$, so $\hat{\vect{h}}$ defines a $d$-cycle on $\cS|_{T\setminus D_W}$ relative to the boundary $\cS|_{S_W}$. Its Poincar\'{e} dual is a $(d-1)$-cocycle on $\cS|_{T\setminus D_W}$ representing an element of $\widetilde{\coeul}(T,\cW)$. Pulling back under $\hat{\vect{h}}_{\text{ref}}$ gives an element $[w_{\vect{h}}]$ in $\coeul(T,\cW)$, whose Poincar\'{e} dual is an Euler chain $[l_{\vect{h}}]$ in $\eul(T,\cW)$. In this way, a semimetal Hamiltonian specified by a vector field $\vect{h}$ with finite singularities can be represented by an Euler chain for $\cW$.

\begin{definition}\label{defn:semimetaltopinv}
Suppose $H$ is a Dirac-type Hamiltonian specified by a smooth vector field $\vect{h}$ on $T$ with finite singularities $W$. It has a cohomological topological invariant $[w_{\vect{h}}]\in \coeul(T,\cW)\subset H^{d-1}(T\setminus W,\ZZ)$ as in the above paragraph, and its \emph{Euler chain representation} is the Poincar\'{e} dual $[l_{\vect{h}}]\in \eul(T,\cW)\subset H_1(T,W,\ZZ)$.
\end{definition}

\subsubsection{``Jumps'' in topological invariants when restricted to subtori} The Euler chain representation of a semimetal is extremely useful because it invariantly characterises and generalises the idea of ``jumps in Chern numbers'' as a Weyl point is traversed, even for complicated local charge configurations. More precisely, the natural pairing $H^{d-1}(T\setminus W,\ZZ)\times H_{d-1}(T\setminus W,\ZZ)\rightarrow\ZZ$ is well-defined for any semimetal with invariant $[\omega_{\vect{h}}]$, and any homology $(d-1)$-cycle $[P]$ in $T$ which avoids $W$. For example, we can take $[P]$ to be the fundamental classes of $(d-1)$-spheres surrounding Weyl points, or (when $T=\TT^d$) of $(d-1)$ subtori at a fixed coordinate. The pairing Poincar\'{e} dualises to the \emph{intersection pairing}\footnote{The intersection pairing is more precisely defined as a map $H^{\text{BM}}_1(T\setminus W)\times H_{d-1}(T\setminus W)\rightarrow H_0(T\setminus W)$, as in 2.6.17 of \cite{Chriss}. Here $H^{\text{BM}}_1(T\setminus W)$ is Borel--Moore homology, mentioned in Section \ref{sec:dualMV}, which is isomorphic to $H_1(T,W)$ and Poincar\'{e} dual to $H^{d-1}(T\setminus W)$.} between $[l_{\vect{h}}]$ and $[P]$, i.e.\ the signed number of intersections between $l_{\vect{h}}$ and $P$.

As a concrete example in $d=3$, for those $[P]$ which are represented by $2$-submanifolds $P$, the intersection number computes the integral over $P$ of the Berry curvature 2-form --- this yields the familiar first Chern numbers (``weak'' Chern invariants) and their discontinuities across the Weyl points (Fig.\ \ref{fig:Intersection}-\ref{fig:Intersection2}).

\subsubsection{Bulk-boundary correspondence and rewiring Fermi arcs}\label{sec:BBC}
For $T=\TT^d$, a projection $\pi:(\TT^d,W)\rightarrow(\widetilde{T}=\TT^{d-1},\widetilde{W})$ induces a homology projection $\pi_*:H_1(T,W)\rightarrow H_1(\widetilde{T},\widetilde{W})$, which we can think of as being Poincar\'{e} dual to the bulk-boundary correspondence \cite{MTFermi,MT1,MT3}. Under $\pi_*$, an Euler chain becomes a surface \emph{Fermi arc} --- a 1-chain on $\widetilde{T}$ whose boundary is the projected 0-chain of local charges. Figure \ref{fig:Bulk2boundary} provides an intuitive picture in $d=3$ in which $\pi$ projects out the $k_z$ coordinate, illustrating how surface Fermi arcs are determined\footnote{In \cite{MTFermi}, we used a different sign convention for which Fermi arcs point from $+$ to $-$, so the Fermi arcs there are \emph{oppositely oriented} to those in this paper.} from bulk Euler chains by $\pi_*$. The Euler chain description is coordinate-free, and determines the Fermi arc connectivity for the projection to \emph{any} surface.

 \begin{figure}
 \centering
        \includegraphics[scale=0.20]{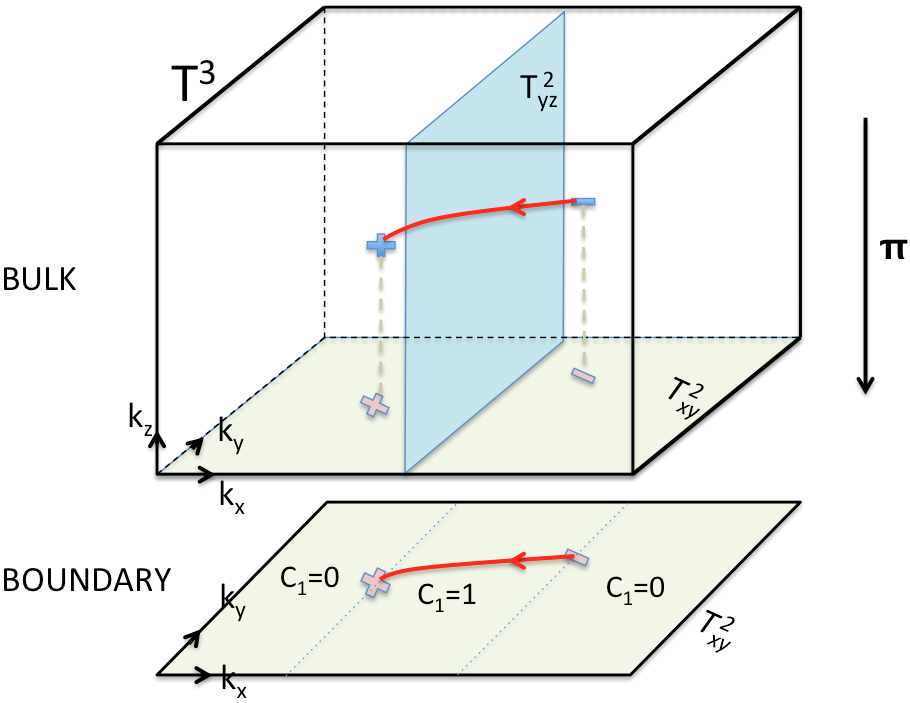}\quad\includegraphics[scale=0.20]{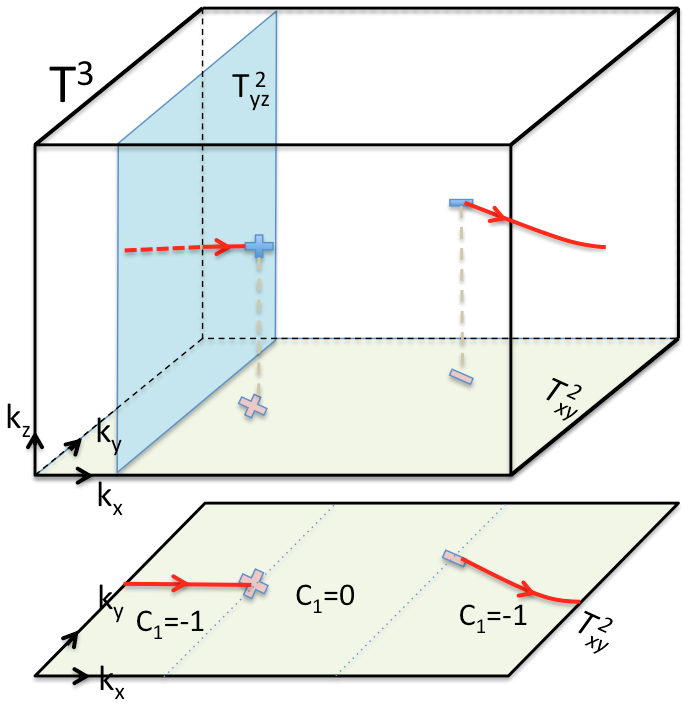}
        \caption{Red lines in the bulk Brillouin torus $\TT^3$ are Euler chains representing semimetals, encoding, in particular, how their first Chern numbers (weak invariants) $c_1^{yz}, c_1^{zx}$ vary with $k_x, k_y$. Only the values of $c_1^{yz}$ are indicated. Gapless surface states appear at points in the surface Brillouin torus $\TT^2$ where a first Chern number is non-zero, forming a Fermi arc connecting the projected Weyl points. The Fermi arc is determined from the Euler chain homologically by projection $\pi_*$}\label{fig:Bulk2boundary}
 \end{figure}

\begin{remark}
If we are able to access the complete specification of a semimetal Hamiltonian as an operator, then its surface Fermi arc is uniquely determined in principle. In practice, we have access to the Fermi arcs, and different Fermi arcs (as 1-chains) come from different Hamiltonians (as operators). The passage from Hamiltonian to Fermi arc can be modelled, for example, by transfer matrices \cite{Hatsugai2,Avila,Dwivedi2}. At the \emph{topological} level, a coarser but more appropriate question is how Fermi arcs may be used to distinguish \emph{topologically distinct} semimetal Hamiltonians. For this latter question, we note that \emph{any} Euler chain representative for a semimetal Hamiltonian $H$ maps under $\pi_*$ to a 1-chain which is topologically equivalent to the actual Fermi arc for $H$. Thus Fermi arcs which are topologically distinct must come from topologically distinct Hamiltonians; our slight abuse of language in calling the 1-chain a ``Fermi arc'' does not matter at the level of topological invariants.
\end{remark}

Certain scenarios involving tuning semimetal Hamiltonians and ``rewiring'' their surface Fermi arcs were considered in \cite{LFF,Dwivedi}. These essentially involve homotopies $\breve{\vect{h}}$ which fix the local charges $\cW=\cW_{\vect{h}}=\cW_{\vect{h}'}$ (at least up to a diffeomorphism moving the Weyl points), and so induce the identity map on the class of Euler chains in $\eul(T,\cW)\subset H_1(T,W)$. Under $\pi_*$, the surface Fermi arcs for $\vect{h}$ and $\vect{h}'$ may become ``rewired'', but only in such a way that their class in $H_1(\widetilde{T},\widetilde{W})$ remains unchanged. Rewirings which involve a change in $H_1(\widetilde{T},\widetilde{W})$ necessarily require the prefiguring Euler chain in $\eul(T,\cW)$ to change homology class, and this cannot be achieved by homotopies $\breve{\vect{h}}$ which fix the local charges $\cW$. Figures \ref{fig:Intersection}-\ref{fig:Intersection2} illustrate some examples of rewirings.

 \begin{figure}
 \centering
        \includegraphics[width=0.9\textwidth]{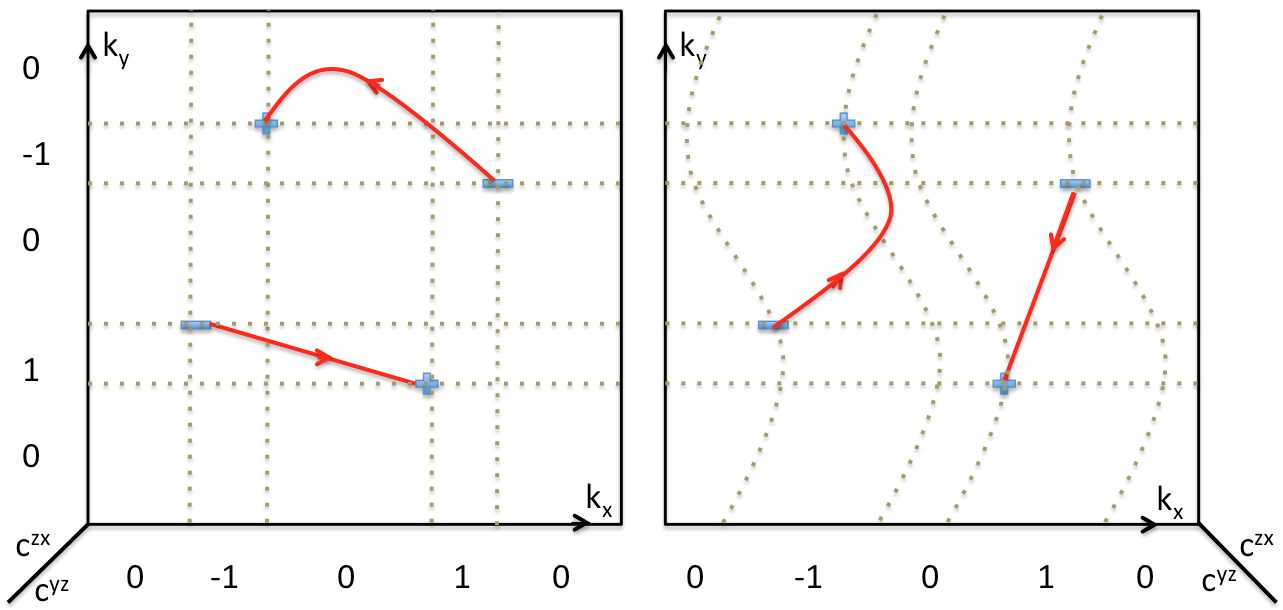}
        \caption{Black square represents $\TT^3$ with one coordinate suppressed; we can also interpret it as the projected surface Brillouin torus $\TT^2$. Two homologous Euler chains (the left and right pairs of red directed lines) having the same local charge configuration are shown, and we can also interpret them as the corresponding surface Fermi arcs. [L] First Chern numbers for 2D subtori in the $yz$ (resp.\ $zx$) directions at fixed $k_x$ (resp.\ $k_y$) are shown, and may be computed by counting the signed intersection between the Euler chain and the subtorus. The ``straightness'' of the Euler chain does not matter, and Chern numbers can also be calculated for 2-cycles other than the standard 2D subtori with a fixed coordinate [R].}\label{fig:Intersection}
 \end{figure}

 \begin{figure}
 \centering
        \includegraphics[width=0.9\textwidth]{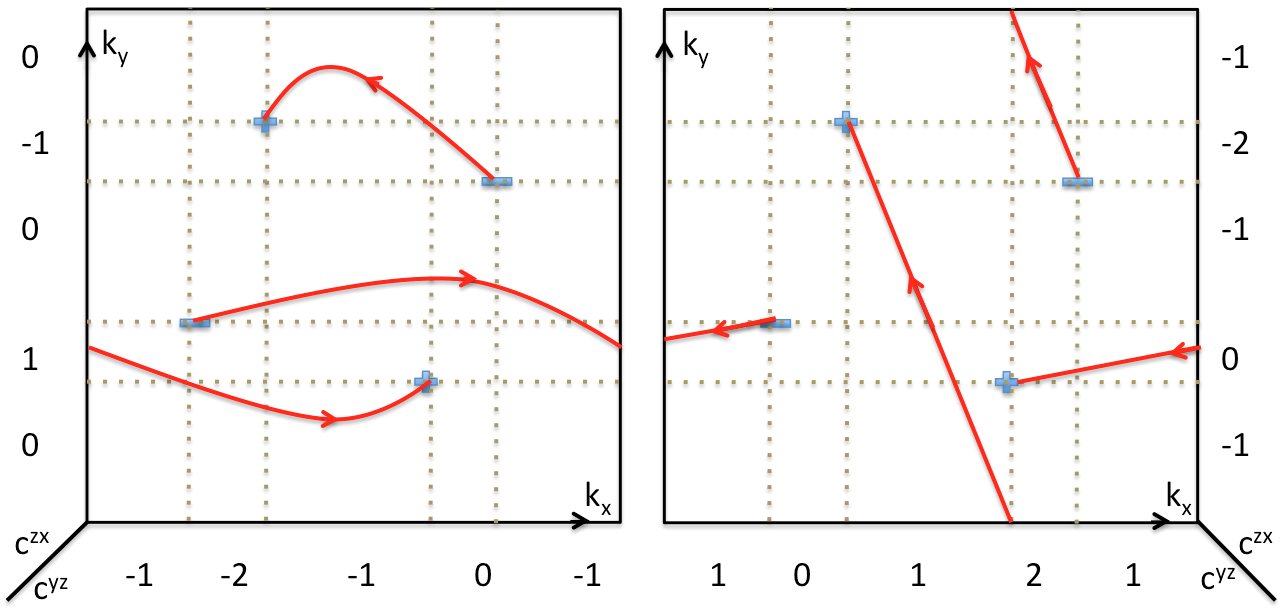}
        \caption{More complicated Euler chains in $\eul(T,\cW)\subset H_1(\TT^3, W)$ representing more semimetal invariants in $\coeul(T,\cW)\subset H^2(\TT^3\setminus W)$. The left and right Euler chains are not homologous.}\label{fig:Intersection2}
 \end{figure}

\begin{remark}\label{remark:higherjumps}
In $d>3$, the intuition afforded by Figures \ref{fig:Chern-jump} and \ref{fig:Bulk2boundary}-\ref{fig:Intersection2} carries over in much the same way. For example, in $d=4$, the blue surfaces represent hyperplane slices (3-tori) on which the semimetal gerbe of Section \ref{sec:semimetalgerbe} restricts and has a curvature 3-form. There are four independent slice directions. As a slice is moved transversely across a semimetal band crossing with local charge $q$, the Dixmier--Douady invariant of the gerbe on the slice jumps by $q$. In $d=5$, the slices are 4-tori and and it is the second Chern number of the $\mathsf{T}$-invariant semimetal (Section \ref{sec:5Dsemimetal}) which jumps by $q$.
\end{remark}

\begin{remark}\label{remark:nonzeroEuler}
It is also possible to define Euler structures for $T$ even if $\chi(T)\neq 0$ by introducing a basepoint $k_0\in T$ which ``contains $\chi(T)$'', see \cite{Burg}. An Euler chain for a 0-chain $\cW$ is then a 1-chain $l\in C_1(T,\ZZ)$ such that $\partial l=\cW-\chi(T)k_0$. Every vector field $\vect{h}$ on $T$ with isolated zeroes $W$ admits an Euler chain for its 0-chain $\cW_{\vect{h}}$ of local charges in this sense: take $l=\sum_{w\in W}\text{Ind}_w(\vect{h})(w-k_0)$, then $\partial l = \sum_{w \in W}\text{Ind}_w(\vect{h}) w-(\sum_{w\in W}\text{Ind}_w(\vect{h}))k_0=\cW_{\vect{h}}-\chi(T)k_0$ by the Poincar\'{e}--Hopf theorem. Euler structures for $T$ with basepoint $k_0$ are then defined as classes of pairs $(\vect{h},l)$ which are considered equivalent if $l_2=l_1+\text{cs}(\vect{h}_1,\vect{h}_2)$ up to a boundary. There is an action of $H_1(T,\ZZ)$ on these classes by addition to $l$, which is furthermore free and transitive.
\end{remark}


\section{Topological semimetal invariants in $d=3,4,5$}\label{sec:higherdimension}
\subsection{General remarks on Bloch bundles and gauge invariance}
Often, it is implicitly assumed that the Hamiltonians act on a \emph{trivial} Bloch bundle $\mathscr{S}=T\times\CC^n$ and that a trivialisation (choice of basis) has been given, so that the Dirac-type Hamiltonians of the form \eqref{Dirac-type-Hamiltonian} are well-defined globally. Although this default setting suffices to illustrate many interesting features of semimetals, it is worth mentioning that non-trivial $\mathscr{S}$ can arise physically and can still be handled mathematically. In general one has, after a Bloch--Floquet transform, a \emph{Hilbert bundle} $\mathscr{H}\rightarrow T$ in which the fibre $\mathscr{H}_k$ over $k\in T$ comprises the quasi-periodic Bloch wavefunctions with quasimomentum $k$ \cite{FM,RS,Kato,Kuchment}. One typically restricts attention to a low-energy subbundle $\mathscr{S}$ defined by a finite number of energy bands separated from the rest by spectral gaps. Subsequently, a \emph{Bloch bundle} $\mathscr{S}$ will refer to such a truncated finite-rank hermitian vector bundle over a manifold $T$ of (quasi)-momenta, on which the \emph{Bloch Hamiltonians} act fiberwise.

In the theory of topological band insulators, it is crucial that $\mathscr{S}$ or its subbundles can be non-trivializable, e.g.\ 2D Chern insulators are essentially specified by the first Chern class of a non-trivializable valence line bundle. Expressions such as \eqref{Dirac-type-Hamiltonian} should be understood as \emph{local} expressions, which continue to make sense on a Bloch bundle $\mathscr{S}$ as long as $\mathscr{S}$ has the structure of a Clifford module bundle (or spinor bundle) so that the $n\times n$ operators $\gamma_i$ are well-defined globally. In particular, the analysis of the spectrum of $H=\vect{h}\cdot\vect{\gamma}$ remains the same for \emph{any} $\mathscr{S}$ --- namely, the $\pm |\vect{h}(k)|$ eigenspaces degenerate (i.e.\ there is a band crossing) precisely at the zeroes of $\vect{h}$. On the other hand, the topology of the valence subbundle over $T\setminus W$, which is defined by the Fermi projection $\frac{1}{2}(1-\hat{\vect{h}}(k))$, does depend on $\mathscr{S}$ through the operators $\gamma_i$.

When $n>2$, the most general $n$-band Hamiltonian is \emph{not} of Dirac-type, and well-defined topological invariants should respect the $\text{U}(n)$ gauge invariance of the rank-$n$ Bloch bundle $\mathscr{S}$. Nevertheless, Dirac-type Hamiltonians are generic when certain time-reversal/particle-hole symmetries are imposed. Such additional symmetries restrict the gauge group, and the topological invariants which we will define for insulators/semimetals described by Dirac-type Hamiltonians, are gauge invariant in this restricted sense. Some examples of such gauge restrictions were studied in \cite{ASSS,ASSS2}, and we also analyse them in Section \ref{sec:higherdimension}.

\subsection{General Hamiltonians over a manifold}
The properties of Dirac-type Hamiltonians which are relevant for defining insulator/semimetal invariants can be abstracted as follows. An abstract (quantized) Dirac-type Hamiltonian is a section of the vector part of some Clifford algebra bundle over a momentum space manifold $T$, i.e.\ the quantization of a vector field $\vect{h}$ specifying the abstract Hamiltonian as in Definition \ref{defn:abstractHamiltonian}.  Pointwise in $T$, the spectrum of such a Hamiltonian (as a Clifford algebra element) can be found quite easily; it depends only on the vector field and has the crucial feature that it degenerates exactly at the zero set $W$ of the vector field. When such Hamiltonians are represented on some Clifford module bundle $\mathscr{S}$ (physically the Bloch bundle), they become concrete families of Hermitian operators with Fermi projections etc. The local topological invariant at a zero $w\in W$ is intrinsic to $\vect{h}$ rather than the particular choice of Bloch bundle $\mathscr{S}$, and exists even when $\mathscr{S}$ is not a trivial bundle. Furthermore, the symmetries and gauge structure for the Hamiltonians become more transparent at this abstract level, and this is especially so in higher dimensions.

{\bf Hamiltonians as quantized vector fields.}
For $d\geq 3$, let $\mathcal{E}$ be an oriented rank-$d$ real vector bundle over a compact oriented manifold $T$ with fibre metric $g$, along with a spin${}^c$ structure (e.g. take $T$ a spin${}^c$ manifold and $\mathcal{E}$ its tangent bundle, then there is a $H^2(T,\ZZ)$ worth of spin${}^c$ structures to choose from). The Clifford algebra bundle $Cl(\mathcal{E},g)$ is a ``quantization'' of the exterior algebra bundle $\bigwedge^*\mathcal{E}$ to allow for multiplication of vectors, both bundles having structure group $\text{SO}(d)$ \cite{LM}. With the spin${}^c$ structure, we can construct the (irreducible) spinor bundle $\mathscr{S}$ (the ``Bloch bundle'') of complex rank $n=2^{\floor{\frac{d}{2}}}$, with $\text{Spin}^c(d)$-invariant Hermitian inner product, and Clifford multiplication gives a \emph{real} bundle homomorphism $Cl(\mathcal{E},g)\rightarrow \text{End}(\mathscr{S})$. Thus there is a map $\tilde{c}:\bigwedge^*\cE\rightarrow Cl(\cE,g)\rightarrow\text{End}(\mathscr{S})$.

In particular, $\mathcal{E}\subset\bigwedge^*\mathcal{E}$ is identified with a subbundle of $Cl(\mathcal{E},g)$, and the map $\tilde{c}$ takes a section $\vect{h}\in\Gamma(\cE)$ to a concrete Dirac-type Hamiltonian. More explicitly, an orthonormal frame $\{e_1,\ldots,e_d\}$ for $\mathcal{E}$ becomes $\{c(e_1),\ldots,c(e_d)\}$ in $Cl(\mathcal{E},g)$ and satisfies the Clifford relations $c(e_i)c(e_j)+c(e_j)c(e_i)=2\delta_{ij}$. The $c(e_i)$ are represented on $\mathscr{S}$ as traceless self-adjoint endomorphisms $\tilde{c}(e_i)\equiv\gamma_i$ which inherit the familiar relations $\gamma_i\gamma_j+\gamma_j\gamma_i=2\delta_{ij}{\bf 1}_\mathscr{S}$.
Thus, a section $\vect{h}\in\Gamma(\mathcal{E})$ defines a (quantized) \emph{abstract Dirac-type Hamiltonian} $c(\vect{h})\in\Gamma (Cl(\mathcal{E},g))$ (c.f.\ Definition \ref{defn:abstractHamiltonian}), which is then represented by a \emph{concrete Dirac-type Hamiltonian} $\tilde{c}(\vect{h})=\vect{h}\cdot\vect{\gamma}$ acting on $\mathscr{S}$. The concrete \emph{Dirac-type Bloch Hamiltonian} above a point $k\in T$ is the operator $H(k)=\vect{h}(k)\cdot\vect{\gamma}(k)$ acting on the fiber $\mathscr{S}_k\cong \CC^n$. 

In the Clifford algebra bundle, it is already ($\text{SO}(d)$-invariantly) true that $c(\vect{h})^2=g(\vect{h},\vect{h})=|\vect{h}|^2$. It follows that the spectrum of a Bloch Hamiltonian is $\text{spec}(H(k))=\pm|\vect{h}(k)|$ with each eigenvalue $2^{\floor{\frac{d}{2}}-1}$-fold degenerate, independently of $\mathscr{S}$, and this is $\text{Spin}^c(d)$-invariant. Precisely at the zeroes of $\vect{h}$, the two $2^{\floor{\frac{d}{2}}-1}$-fold degenerate energy bands cross, and such a crossing is precisely protected by the topological index of $\vect{h}$ there.

More generally, we only require $\mathscr{S}$ to be a Clifford module bundle for $Cl(\mathcal{E},g)$, for which $c(\vect{h})\in\Gamma (Cl(\mathcal{E},g))$ acts self-adjointly. Such Clifford module bundles are in fact twisted versions of spinor bundles, obtained by tensoring with some vector bundle $\mathcal{V}$, so we shall mainly consider only irreducible spinor bundles. In Section \ref{sec:torsionsemimetal} we will also consider $\bigwedge^2\mathcal{E}\subset Cl(\mathcal{E},g)$ acting skew-adjointly on $\mathscr{S}$, giving rise to ``bilinear Hamiltonians''.

\begin{remark}
Physically, one might start off with a given Bloch bundle $\mathscr{S}$, which is a hermitian $\text{U}(n)$ vector bundle over $T$ obtained by Fourier/Bloch--Floquet transform. In order to have a notion of Clifford multiplication and Dirac-type Hamiltonians, we need some assumptions on $\mathscr{S}$ allowing it to have the structure of a (irreducible) Clifford module bundle for some $\mathcal{E}$. In the first place, $\mathscr{S}$ should have the correct (complex) rank.
In most model Hamiltonians considered in the physics literature, $T=\TT^d$ and the Bloch bundle $\mathscr{S}$ is assumed to be trivial. In this case, we take the trivial bundle $\cE$, identified with the tangent bundle of $\TT^d$, and $\vect{h}\in\Gamma(\cE)$ is a tangent vector field. Then the corresponding Dirac-type Bloch Hamiltonian $H(k)=\vect{h}(k)\cdot\vect{\gamma}$ has gamma matrices $\gamma_i$ which can be taken to be constant over $\TT^d$.
\end{remark}

\subsection{Remarks on $\mathsf{T}$ and $\mathsf{C}$ symmetries}
Formally, a family of $n\times n$ Bloch Hamiltonians over $T$ can also be regarded as a family of Hamiltonians on an $n$-level quantum mechanical system, with adiabatic phases/holonomies, etc. There is, however, a fundamental difference between $T$ as an adiabatic parameter space and as a quasi-momentum space, when antiunitary symmetries are introduced. The Brillouin torus $\TT^d$ is topologically the space of unitary characters for the translation group $\ZZ^d$ of a lattice, so the complex conjugation involved in an antiunitary symmetry operator induces an involution $\tau$ (a $\ZZ_2$-action) on $\TT^d$. Explicitly, when we parameterise $\TT^d$ by angles $k=(k_1, \ldots, k_d)$, the character $n\mapsto e^{\im n\cdot k}$ labelled by $k$ conjugates to the character labelled by $-k\equiv\tau(k)$. In the context of Bloch Hamiltonians, a fermionic time-reversal symmetry $\mathsf{T}$ is a lift of $\tau$ to a map on the Bloch bundle $\mathscr{S}$ which is antiunitary on fibres, and satisfies $\mathsf{T}^2=-1$. Compatibility of $\mathsf{T}$ with $H(k), k\in\TT^d$ is then the condition $\mathsf{T}(k)H(k)\mathsf{T}(k)^{-1}=H(\tau(k))$. Such bundles equipped with a ``Quaternionic'' structure $\mathsf{T}$ were studied in \cite{Dupont}, and have more recently been investigated in the context of topological insulators \cite{FM,dNG,MT2,Thiang}. For bosonic time-reversal, $\mathsf{T}$ squares to the identity instead, and is sometimes called a ``Real structure''.

Similarly, an antiunitary particle-hole symmetry $\mathsf{C}$ \emph{anticommutes} with the Bloch Hamiltonians in a way that respects $\tau$, namely, $\mathsf{C}(k)H(k)\mathsf{C}(k)^{-1}=-H(\tau(k))$. On the other hand, for a family of Hamiltonians parametrised by $T$, we simply have $\mathsf{T}(k)H(k)\mathsf{T}(k)^{-1}=H(k)$ and $\mathsf{C}(k)H(k)\mathsf{C}(k)^{-1}=-H(k),\, k\in T$, without the involution $\tau$. In either case, the squares $\mathsf{T}^2, \mathsf{C}^2$ may be $\pm 1$, and when both $\mathsf{T}$ and $\mathsf{C}$ are present, they may be assumed to commute \cite{FM}.

The same involution $\tau$ is induced on $\TT^d$ by a spatial inversion symmetry $\mathsf{P}$, since a translation by $n\in\ZZ^d$ becomes a translation by $-n$. Thus we can consider a time-reversal symmetry which concurrently effects spatial inversion\footnote{Note that this is not the same as requiring $\mathsf{T}$ and $ \mathsf{P}$ to \emph{separately} be symmetries, and indeed, the latter situation requires additional data specifying whether $\mathsf{P}, \mathsf{T}$ commute or anticommute \cite{Wigner,ShawLever}, leading in each case to different Wigner classes.}; we write $(\mathsf{T}\mathsf{P})$ for such a symmetry operator. Then we have $(\mathsf{T}\mathsf{P}(k))H(k)(\mathsf{T}\mathsf{P}(k))^{-1}=H(k),\, k\in\TT^d$ where the effect of $\tau$ now cancels; similarly for $(\mathsf{C}\mathsf{P})$. Thus the operator $(\mathsf{T}\mathsf{P})$ is an ordinary quaternionic structure, and the Bloch bundle $\mathscr{S}$ has invariants as a \emph{symplectic bundle}. 

{\bf Convention.} For this paper, we will simply write $\mathsf{T}, \mathsf{C}$ with the understanding that $(\mathsf{T}\mathsf{P}), (\mathsf{C}\mathsf{P})$ is meant whenever $T$ is a Brillouin torus. Thus we always have $\mathsf{T}(k)H(k)\mathsf{T}(k)^{-1}=H(k)$ and $\mathsf{C}(k)H(k)\mathsf{C}(k)^{-1}=-H(k),\,k\in T$.

\subsection{Abstract Weyl semimetal in 3D}
Let us explain the above abstractions in more detail for two-band Hamiltonians in $d=3$, in which $\text{dim}_\CC\mathscr{S}=2$. This is the setting for the basic Weyl semimetal. 

Since $\text{Spin}^c(3)=\text{U}(2)$, any rank-2 hermitian vector bundle $\mathscr{S}$ can arise as a spinor bundle: take the traceless Hermitian endomorphisms $\mathcal{H}_0$ of $\mathscr{S}$ with the Hilbert--Schmidt real inner product $\langle H_1(k), H_2(k) \rangle=\frac{1}{2}\text{tr}(H_1\circ H_2(k))$.  Then $\mathcal{H}_0$ becomes a real orientable rank-3 bundle with structure group $\text{SO}(3)=\text{PU}(2)$ liftable to $\text{Spin}^c(3)$. An orthonormal frame for $\mathcal{H}_0$ gives a set of Pauli operators $\{\sigma_1,\sigma_2,\sigma_3\}$, which is positively oriented if $\sigma_1\sigma_2=\im\sigma_3$. Such a positively-oriented frame has the form of the standard Pauli matrices in some local trivialization of $\mathscr{S}$. The Bloch Hamiltonian is given locally by $H(k)=\vect{h}(k)\cdot\vect{\sigma}$. Then $(\mathcal{H}_0,\langle\cdot,\cdot\rangle)$ has a Clifford algebra bundle and a choice of spin$^c$ structure yields a spinor bundle. $\mathscr{S}$ is some twisted version of this spinor bundle, and the twisting can be thought of as changing the spin$^c$ structure under the canonical action of $H^2(T,\ZZ)$ \cite{LM}.


Thus $\{\sigma_1,\sigma_2,\sigma_3\}$ may be thought of as $\{\tilde{c}(e_1),\tilde{c}(e_2),\tilde{c}(e_3)\}$ for an orthonormal frame $\{e_1,e_2,e_3\}$ in $(\mathcal{E},g)\cong(\mathcal{H}_0,\langle\,\cdot\,,\,\cdot\,\rangle)$. From this point of view, a vector field $\vect{h}\in\Gamma(\cE)$ is quantized to an \emph{abstract Hamiltonian} $c(\vect{h})\in Cl(\cE,g)$, which acts concretely as $\tilde{c}(\vect{h})=\vect{h}\cdot\vect{\sigma}$ on $\mathscr{S}$. Note that, in particular, the spectrum $\text{spec}(H(k))=\pm |\vect{h}(k)|$ is invariant under $\text{U}(2)=\text{Spin}^c(3)$ gauge transformations of $\mathscr{S}$.

\begin{remark}\label{rem:gauge3D} A gauge transformation of the Bloch bundle $\mathscr{S}$ implemented by local $2\times 2$ unitaries $U(k)$ turns $\vect{h}(k)\cdot\vect{\sigma}$ into the conjugate $U(k)(\vect{h}(k)\cdot\vect{\sigma})U^{-1}(k)$. The result of this conjugation is known to be equal to $\vect{h'}(k)\cdot\sigma$ where $\vect{h'}(k)$ is rotated from $\vect{h}(k)$ by the $\text{SO}(3)$ element $\pi(U(k))$ in $$1\rightarrow \text{U}(1)\rightarrow\text{U}(2)\xrightarrow{\pi} \text{PU}(2)=\text{SO}(3)\rightarrow 1.$$
Similarly, $\text{PSU}(2)=\text{SO}(3)$ in case the structure group of $\mathscr{S}$ can be reduced to $\text{SU}(2)$.
\end{remark}

\subsubsection{Symmetries of 3D semimetal}
Two band-Hamiltonians can more generally include a trace term, i.e.\ $H(k)=m(k)+\vect{h}(k)\cdot\vect{\sigma}(k)$, where $m(k)$ is not constant and so not eliminable by shifting the overall energy level. Indeed, Hamiltonians with such $m(k)$ terms (and constant $\vect{\sigma}$) can give rise to new phenomena such as ``Type-II Weyl semimetals'' \cite{XutypeII,Sol}. It is possible to distinguish the two situations, i.e.\ eliminate the $m(k)$ term, by imposing a natural symmetry constraint.

{\bf Quaternionic structure and $\mathsf{C}$-symmetry.}
Let us recall the construction of the spinor representation $S$ of $Cl_{0,3}$. Let $e_1,e_2,e_3$ be orthonormal vectors in $\RR^3$, $e_ie_j\equiv e_i\wedge e_j\in\Lambda^2(\RR^3)$, and $e_1e_2e_3\equiv e_1\wedge e_2\wedge e_3\in\Lambda^3(\RR^3)$. The even part $Cl_{0,3}^+\cong Cl_{2,0}$ is isomorphic to $\HH$, and for the complexified algebras, $\CC l_3^+\cong \CC l_2\cong M_2(\CC)=\HH\otimes_\RR\CC$. The spinor representation on $S=\CC^2$ is the (unique complex) irreducible representation of $\CC l_3^+\cong M_2(\CC)$, and is extended to $\CC l_3\supset Cl_{0,3}$ by setting a chirality condition $\tilde{c}(-\im e_1e_2e_3)=1$. $\text{Spin}(3)=\text{SU}(2)\cong\text{Sp}(1)$ lies in $Cl_{0,3}^+\cong \HH$ as the unit quaternions, and $S$ is an irreducible representation of $\text{Spin}(3)$. This is just the fundamental representation of $\text{SU}(2)\cong \text{Sp}(1)$ which is well-known to be quaternionic.

 The quantized bivectors $c(e_3e_1), c(e_3e_2)$ generate $Cl_{2,0}\cong Cl^+_{0,3}\cong\HH$. On $S$, they are represented by skew-Hermitian operators $\tilde{c}(e_3 e_1)=\im\sigma_2,\, \tilde{c}(e_3 e_2)=-\im\sigma_1$ (which generate a quaternion algebra), and with the chirality element $\tilde{c}(-\im e_1e_2e_3)$ being the identity operator on $S$, we recover the Pauli matrices $\tilde{c}(e_i)=\sigma_i$. Since $S$ is a quaternionic representation for $\text{Spin}(3)$, there is an antiunitary operator $\mathsf{\Theta}$ on $S$ commuting with $Cl^+_{0,3}\supset\text{Spin}(3)$ such that $\mathsf{\Theta}^2=-1$, from which we deduce that $\mathsf{\Theta}$ \emph{anticommutes} with $c(e_i)=\sigma_i$. For instance, we can take $\mathsf{\Theta}=\im\sigma_2\circ\kappa$ where $\kappa$ is complex conjugation in the basis where $\sigma_i$ are the standard Pauli matrices. This operator is often taken to be the fermionic time-reversal operator $\mathsf{T}$ on spin-$\frac{1}{2}$ systems. 

If $\cE$ has a spin (not just spin${}^c$) structure, there is also a notion of spinor bundles $\mathscr{S}$ associated to the spin structure \cite{LM}, and these have structure group $\text{Spin}(3)=\text{SU}(2)=\text{Sp}(1)$. Globalising the above constructions, these spinor bundles come with a quaternionic structure $\mathsf{C}$ anticommuting with $H=\tilde{c}(\vect{h})=\vect{h}\cdot\vect{\sigma}$.

{\bf Reduction to $\text{SU}(2)\cong\text{Sp}(1)$.}
Suppose we have a reduction of the structure group of the Bloch bundle from $\text{U}(2)$ to $\text{Spin}(3)=\text{SU}(2)$. This is possible exactly when $\text{det}(\mathscr{S})$ is trivial, i.e.\ $c_1(\mathscr{S})=c_1(\text{det}(\mathscr{S}))=0$ (this is a general result about reducing from $\text{U}(n)$ to $\text{SU}(n)$). Alternatively, since $\text{SU}(2)=\text{Sp}(1)$, the reduction from $\text{U}(2)$ to $\text{SU}(2)$ requires a quaternionic structure on $\mathscr{S}$, which provides an isomorphism between the line bundles $\mathcal{L}_1$ and $\overline{\mathcal{L}_2}$ for an orthogonal splitting $\mathscr{S}=\mathcal{L}_1\oplus\mathcal{L}_2$. Then $c_1(\mathscr{S})=c_1(\mathcal{L}_1)+c_1(\mathcal{L}_2)=c_1(\mathcal{L}_1)-c_1(\mathcal{L}_1)=0$. In this case, we can identify $\mathcal{H}_0$ with the tangent bundle $\mathcal{E}\rightarrow\TT^3$, as (trivial) $\text{SO}(3)$ bundles. 

Therefore, in the two-band case in 3D, the imposition of a $\mathsf{C}$-symmetry constraint is one possible way to justify (1) the restriction from arbitrary $2\times 2$ Bloch Hamiltonians to the \emph{traceless} ones, and (2) the triviality of the Bloch bundle $\mathscr{S}$.

\subsection{$\mathsf{T}$-invariant Dirac-semimetal in 5D}\label{sec:5Dsemimetal}
We now generalize the two-band crossings of the previous section to four-band crossings. These should involve crossings of \emph{pairs} of bands, in analogy to the dispersion of massless Dirac spinors. Again, the adjective ``Dirac'' is loaded and possibly misleading, and our constructions below are more analogous to Weyl fermions in 1+5 D (which is still misleading). For this Subsection, we take $T$ to be a compact oriented 5-manifold, and $\mathscr{S}=T\times\CC^4$ a trivial Hermitian vector bundle.

The bundle $\mathcal{H}_0$ of traceless Hermitian endomorphisms of $\mathscr{S}$ is now a trivial rank-15 real vector bundle whose sections are Bloch Hamiltonians $k\mapsto H(k)$ with inner product $\langle H_1(k), H_2(k) \rangle = \frac{1}{4}\text{tr}(H_1(k)\circ H_2(k))$.

For notational ease, we will sometimes suppress the dependence on $k$ when dealing with a single operator.

\subsubsection{Dirac-type Hamiltonians in 5D}\label{sec:Dirac5D}
Physicists often restrict attention to ``Dirac-type'' Hamiltonians $H(k)=\vect{h}(k)\cdot\vect{\gamma}$ where $\vect{\gamma}=\{\gamma_1,\ldots,\gamma_5\}$ is an orthonormal set of $4\times 4$ hermitian operators on $S$ satisfying the Clifford algebra relation $\gamma_i\gamma_j+\gamma_j\gamma_i=2\delta_{ij}$, and $\vect{h}$ is a 5-component vector field. The pragmatic justification is ease of manipulation, since the spectrum is again easily found to be $\pm|\vect{h}(k)|$, with doubly-degenerate eigenvalues that become four-fold degenerate precisely at the zeroes of $\vect{h}$ (Fig.\ \ref{fig:symmetric spectrum}). However, we should note that arbitrary gauge transformations will not preserve this particular form of $4\times 4$ Hermitian matrices. It is therefore useful to find a symmetry which reduces the allowed gauge transformations so that the Dirac-form of $H$ is automatically preserved. 

\begin{figure}
    \centering
    \includegraphics[width=0.4\textwidth]{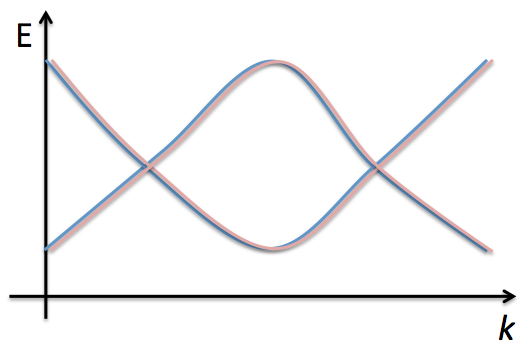}
    \caption{Four-band crossings for Dirac-type Hamiltonians, with symmetric spectrum.}
    \label{fig:symmetric spectrum}
\end{figure}

We do this by invoking a $\mathsf{T}$-symmetry constraint (a quaternionic structure) on the Hamiltonians, which will pick out the Dirac-type Hamiltonians as a rank-5 subbundle of $\mathcal{H}_0$. In fact, the observation that the proper gauge groups and Berry connections/phases for time-reversal invariant Hamiltonians are \emph{quaternionic}, had already been made in \cite{ASSS,ASSS2,Hatsugai}, and in the $4\times 4$ case such Hamiltonians can be constructed by considering $\CC^4$ as a spin-$\frac{3}{2}$ representation (so-called ``quadrupole Hamiltonians''). We give a more direct account, emphasizing the Clifford algebra and spin groups responsible for this structure.

Consider first a fibre $\CC^4\cong \CC^2\otimes\CC^2$ and define the antiunitary operator 
\begin{equation}\label{eqn:Toperator}
\Theta=(1\otimes \im\sigma_2)\circ \kappa,
\end{equation} where $\kappa$ is complex conjugation, which satisfies $\Theta^2=-1$. We verify that (real) linear combinations of the following mutually anticommuting $\gamma_i, i=1,\ldots,5$, \emph{commute} with $\Theta$:
\begin{equation}
\gamma_1=\sigma_2\otimes\sigma_1,\quad \gamma_2=\sigma_2\otimes\sigma_2,\quad\gamma_3=\sigma_2\otimes\sigma_3,\quad\gamma_4=\sigma_1\otimes 1,\quad\gamma_5=\sigma_3\otimes 1.\label{lineargamma}
\end{equation}

A more invariant way to see this is to notice that an anticommuting set of $4\times 4$ Hermitian operators $\gamma_1,\ldots,\gamma_4$ satisfying $\gamma_i\gamma_j+\gamma_j\gamma_i=2\delta_{ij}, i,j=1,\ldots,4$, generate (over the reals) a copy of $Cl_{0,4}\cong M_2(\HH)$ inside $M_4(\CC)$ which thereby commutes with a $\HH^\text{opp}$ action generated by $\im, \Theta$ (two anticommuting square roots of $-1$). The product of two or three different $\gamma_i,\,i=1,\ldots,4$ is not Hermitian, but $\gamma_5=\gamma_1\gamma_2\gamma_3\gamma_4$ is Hermitian and also commutes with $\Theta$. The remaining ten-dimensional subspace of traceless Hermitian operators is spanned by operators $\im\gamma_i\gamma_j, i\neq j$ and they \emph{anticommute} with $\Theta$ instead --- we will interpret this latter subspace of $4\times 4$ Hamiltonians as those compatible with a particle-hole symmetry $\mathsf{C}$, and these form an important class to be analysed in Section \ref{sec:Z2invariant}.

Therefore, traceless Hermitian operators of the form $H=\vect{h}\cdot\vect{\gamma}$ are precisely those which are compatible with a quaternionic structure $\Theta$, which we interpret as a time-reversal operator $\mathsf{T}$. Conversely, a quaternionic structure $\Theta$ identifies $\CC^4$ with $\HH^2$, and $M_2(\HH)\cong Cl_{0,4}$ is the real subalgebra of $M_2(\CC)$ which commutes with $\Theta$.

Globalising the above arguments, we deduce:
\begin{proposition}
A quaternionic structure $\Theta$ on a rank-$4$ Bloch bundle picks out a real orientable rank-5 subbundle $\mathcal{H}_{0,l}\subset\mathcal{H}_0$ of commuting Dirac-type Hamiltonians. The chirality condition $\gamma_5=\gamma_1\gamma_2\gamma_3\gamma_4$ for an orthonormal frame $\{\gamma_1,\ldots,\gamma_5\}$ for $\mathcal{H}_{0,l}$ defines the positive orientation.
\end{proposition}
A quaternionic structure always exists on a trivial Bloch bundle, and in this case, we identify $\mathcal{H}_{0,l}$ with the tangent bundle $\mathcal{E}$ as $\text{SO}(5)$-bundles. 
Thus $\{\gamma_1,\dots,\gamma_5\}$ corresponds to a positively-oriented orthonormal frame $\{e_1,\ldots,e_5\}$ for $\cE$. In terms of Clifford modules, the Clifford algebra $Cl_{0,5}$ has $S=\CC^4$ as the spinor representation in which the central chirality element $\tilde{c}(e_1e_2e_3e_4e_5)$ is the identity operator. A vector field $\vect{h}$ is quantized to an abstract Dirac-type Hamiltonian $c(\vect{h})$, which is represented concretely on $\mathscr{S}$ as the Hamiltonian $H=\tilde{c}(\vect{h})=\vect{h}\cdot \vect{\gamma}=\sum_{i=1}^5 h^i\gamma_i$ isometrically: $\langle{\vect{h}_1,\vect{h}_2}\rangle=\frac{1}{4}\text{tr}(\vect{h}_1\cdot\vect{h}_2+\ldots)=\vect{h}_1\cdot\vect{h}_2$.

\begin{remark}
The observation that the space $\mathcal{H}_{0,l}$ has structure group $\text{SO}(5)$ was made in Section 5 of \cite{ASSS2} (for trivial bundles $\mathscr{S}$). Generalizing Remark \ref{rem:gauge3D}, $S$ is the spinor representation for $\text{Spin}(5)\cong\text{Sp}(2)\subset\text{U}(4)$, and conjugating $H(k)=\vect{h}(k)\cdot\vect{\gamma}$ by a $\text{Spin}(5)$ unitary matrix $U(k)$ takes it to $H'(k)=\vect{h}'(k)\cdot\vect{\gamma}$ where $\vect{h}'(k)$ is rotated from $\vect{h}(k)$ by the $\text{SO}(5)$ matrix $R(k)$ covered by $U(k)$. The reduction of the gauge group is important: to stay within the class of Dirac-type Hamiltonians, we cannot allow all $\text{U}(4)$ gauge transformations of the Bloch bundle $\mathscr{S}$, but only those in $\text{Spin}(5)\cong\text{Sp}(2)$.
\end{remark}

\subsubsection{Quaternionic valence line bundle, local $\ZZ$-charge, and second Chern class}\label{sec:quaternionbundle}

For a 5-component vector field with finite isolated zeroes $W\subset T$, the local index at $w\in W$ is again given by the degree of the unit vector map $\hat{\vect{h}}:S^4_w\rightarrow S^4\subset\RR^5$ restricted to a small sphere around $w$. This $\ZZ$-valued local index measures the obstruction to ``gapping'' out a four-band crossing \emph{within the family of Dirac-type Hamiltonians} (or $\mathsf{T}$-invariant Hamiltonians). Away from $W$, $H(k)$ can be spectrally-flattened to $\text{sgn}(H(k))=\hat{\vect{h}}(k)\cdot\vect{\gamma}$, and the Fermi projection $\frac{1}{2}(1-\text{sgn}(H(k))$ projects onto a complex rank-2 subspace of $S=\CC^4$. This subspace is actually a \emph{quaternionic line}, i.e.\ an element of $\HH\PP^1$, due to the existence of $\Theta$ commuting with $H(k)$. Just as a unit 3-vector defines a point in $\CC\PP^1\cong S^2$ via $\vect{h}\leftrightarrow$ $(-1)$-subspace of $\vect{h}\cdot\vect{\sigma}$, we also have $\HH\PP^1\cong S^4$ via \{unit 5-vector $\vect{h}$\} $\leftrightarrow$ \{$(-1)$-subspace of $\vect{h}\cdot\vect{\gamma}$\}. 

The valence subbundle $\cE_F$ of a $\mathsf{T}$-invariant Dirac-type Hamiltonian in 5D is thus a quaternionic line bundle over $T\setminus W$. These quaternionic line bundles are pulled back from the tautological one over $\HH\PP^1=S^4$, and may have non-trivial \emph{symplectic} Pontryagin class/second Chern class in $H^4(\,\cdot\,,\ZZ)$ \cite{BH,ASSS}. 

\begin{remark}
The \emph{symplectic} Pontryagin classes $q_i$ for $\text{Sp}(n)$ bundles \cite{BH}, should be distinguished from the real Pontryagin classes $p_i$ for $\text{O(n)}$ bundles. The former comes from $H^{4i}(\text{BSp}(n),\ZZ)$ and equals $(-1)^ic_{2i}$ of the underlying complex bundle of a quaternionic bundle, while the latter comes from $H^{4i}(\text{BO}(n),\ZZ)$ and equals $(-1)^ic_{2i}$ of the complexification of the underlying real bundle of a complex bundle (e.g.\ underlying a quaternionic one). For the tautological quaternionic line bundle $\cH\rightarrow\HH\PP^1$, the symplectic Pontryagin class/second Chern class is a generator of $H^4(S^4,\ZZ)=\ZZ$.
\end{remark}

The $H^4$ part of the semimetal Mayer--Vietoris sequence 
\begin{equation}
\cdots 0\rightarrow H^4(T)\xrightarrow{\iota^*} H^4(T\setminus W) 
\xrightarrow{\beta} H^4(S^4_W)\xrightarrow{\Sigma} H^5(T)\rightarrow 0\ldots\label{MV5D}
\end{equation}
is slightly harder to interpret in terms of characteristic classes of $\cE_F$: unlike $H^2(\cdot,\ZZ)$ and complex line bundles, it is not generally true that $H^4(\cdot,\ZZ)$ classifies all quaternionic line bundles (equivalently $\text{Sp}(1)=\text{SU}(2)$-bundles) over a given space. For example, $\text{SU}(3)$ is a non-trivial $\text{SU}(2)$ bundle over $S^5$ but $H^4(S^5,\ZZ)=0$ so $c_2$ doesn't detect this bundle. Nevertheless, $H^4(\TT^5,\ZZ)\cong\ZZ^5$ with the generators coming from pulling back $\cH$ along the projection $\vect{h}:\TT^5\rightarrow\TT^4\rightarrow S^4=\HH\PP^1$ for each of the five independent choices of 4-subtori. Thus all elements of $H^4(\TT^5,\ZZ)$ do arise as the second Chern class of some valence quaternionic line bundle $\cE_F\rightarrow\TT^5$ (weak insulator invariants), and their restrictions to $H^4(\TT^5\setminus W)$ form the kernel of the local charge map $\beta$ in \eqref{MV5D}. The other elements $\omega\in H^4(\TT^5\setminus W)$ have some non-zero local charges under $\beta$. From a singular vector field $\vect{h}$ with those local charges $\beta(\omega)$, the valence bundle for $\vect{h}\cdot\vect{\gamma}$ has $c_2(\cE_F)=\omega$. Thus for $T=\TT^5$, the $H^4(\cdot,\ZZ)$ groups in \eqref{MV5D} can be interpreted as the weak invariants $c_2$ for the valence bundles $\cE_F$ of the Dirac-type Hamiltonians as constructed above.

\begin{remark}
Quaternionic line bundles are already stable for dim$(T)=5$ \cite{AR}, so we can classify the $\cE_F$ by (reduced) quaternionic $K$-theory $KSp$. Therefore, the semimetal MV-sequence in $KSp$ can be interpreted in terms of obstructions to extending a semimetal valence bundle (over $\TT^5\setminus W$) to an insulator one (over $\TT^5$).
\end{remark}

\subsection{$\mathsf{T}$-invariant and chiral symmetric Dirac-semimetal in 4D}
When we introduced $\mathsf{T}$ to $4\times 4$ Hamiltonians, we found that they had the form $H(k)=\vect{h}(k)\cdot\vect{\gamma}$ for some 5-component vector $\vect{h}(k)$. Now introduce a further (unitary) chiral symmetry $\mathsf{S}$, with $\mathsf{S}^2=1, \mathsf{S}\mathsf{T}=\mathsf{T}\mathsf{S}$ which a compatible $H$ must \emph{anticommute} with. Without loss, we can take $\mathsf{S}=\gamma_5$. Equivalently, this is an additional (antiunitary) $\mathsf{C}$ symmetry, with $\mathsf{C}\mathsf{T}=\mathsf{T}\mathsf{C}$ and $\mathsf{C}^2=-1$: take $\mathsf{C}=(-\sigma_3\otimes \im\sigma_2)\circ\kappa=-\mathsf{T}\circ\mathsf{S}$. The additional anticommutation requirement $\mathsf{S} H=-H\mathsf{S}$ (or $\mathsf{C}H=-H\mathsf{C}$) forces $H(k)$ to be a linear combination of only the first \emph{four} gamma matrices in \eqref{lineargamma}. Furthermore, the traceless condition on $H$ is now automatic.

Thus in 4D, we may identify the $(\mathsf{T},\mathsf{C})$-compatible (or $(\mathsf{T},\mathsf{S})$-compatible) $4\times 4$ Hamiltonians with the tangent bundle of $\TT^4$, via $H(k)=\vect{h}(k)\cdot\vect{\gamma}$, where $\vect{h}, \vect{\gamma}$ now have four components. Similarly, $S=\CC^4$ is the spinor representation for $Cl_{0,4}\otimes_\RR\CC\cong M_4(\CC)$ which has a quaternionic structure $\Theta=\mathsf{T}$ compatible with $Cl_{0,4}=M_2(\HH)$. The spectrum is again $\text{spec}(H(k))=\pm|\vect{h}(k)|$, vanishing at the zeroes $W\subset\TT^4$ of $\vect{h}$. The topological invariant for a four-band crossing at a zero $w\in W$ is the degree of the unit vector map $\hat{\vect{h}}:S_w^3\rightarrow S^3\in\RR^4$.

\begin{remark} There is also an analysis based on supercommutants. The requirement of having symmetry operators $\mathsf{T},\mathsf{C}$ is equivalent to having a \emph{$\ZZ_2$-graded representation} of the Clifford algebra $Cl_{3,0}$ --- $\mathsf{C},\im\mathsf{C},\im\mathsf{C}\mathsf{T}$ are the three anticommuting Clifford generators, acting irreducibly as odd operators on the superspace $\CC^4=\HH\otimes\RR^{1|1}$. The supercommutant, by the super-Schur lemma \cite{Var}, is a real superdivision algebra $Cl_{0,3}$. More explicitly, we can take
\begin{eqnarray}
Cl_{3,0}\, \text{generators}:& &\mathsf{C} = (-\sigma_3\otimes\im\sigma_2)\circ\kappa,\;\; \im\mathsf{C}=(\sigma_3\otimes\sigma_2)\circ\kappa,\;\;\im\mathsf{C}\mathsf{T} =\im\sigma_3\otimes 1=\im\gamma_5\nonumber\\
Cl_{0,3}\, \text{generators}: & & \gamma_1, \gamma_2, \gamma_3\nonumber\\
\text{Grading operator}: & & \gamma_4\nonumber
\end{eqnarray}
\end{remark}
The $(\mathsf{C},\mathsf{T})$-compatible Hamiltonians are spanned by the above $Cl_{0,3}$ generators and the grading operator.

\subsubsection{Semimetal gerbe from Hamiltonian}\label{sec:semimetalgerbe}
Having identified the $\mathsf{T}$ and $\mathsf{S}$-symmetric $4\times 4$ Hamiltonians as those which are linear combinations of $\gamma_1,\ldots, \gamma_4$, we now explain how a gerbe can be associated to such Hamiltonians.

First, it has become usual to associate a gapped chiral symmetric Hamiltonian $H$ with a unitary in half the number of dimensions as follows. In a basis in which $\mathsf{S}$ is $\begin{pmatrix} 1 & 0 \\ 0 & -1\end{pmatrix}$, a chiral symmetric Hamiltonian $H$ is off-diagonal as is $\text{sgn}(H)$. Thus $\text{sgn}(H)=\begin{pmatrix} 0 & U^\dagger \\ U & 0\end{pmatrix}$ where $U$ is unitary since $(\text{sgn}(H))^2=1$. Note that this association $H\leftrightarrow U$ is basis-dependant (under $\text{U}(2)\times\text{U}(2)$ transformations which respect $\mathsf{S}$), and for families of Hamiltonians $k\mapsto H(k)$, the topological invariants of the map $k\mapsto U(k)$, which is defined wherever $\text{sgn}(H(k))$ is defined (i.e.\ away from $W$), can change under ``large'' gauge transformations.

We require a further $\mathsf{T}$ symmetry such that $\mathsf{T}$ commutes with $\mathsf{S}$, so it is of the form $\begin{pmatrix} V & 0 \\ 0 & V'\end{pmatrix}\circ\kappa$ for some $2\times 2$ unitaries $V, V'$. Since $\mathsf{T}^2=-1$, up to a unitary transformation preserving $\mathsf{S}$, we can take $\mathsf{T}=\Theta=\begin{pmatrix} \im\sigma_2 & 0 \\ 0 & \im\sigma_2\end{pmatrix}\circ\kappa$ as in \eqref{eqn:Toperator}. Having done so, a basis for the $\mathsf{T}$ and $\mathsf{S}$ compatible Hamiltonians is $\gamma_1,\ldots, \gamma_4$ given in \eqref{lineargamma}. Note that the unitary transformations which preserve both $\mathsf{S}$ and $\mathsf{T}$ form a $\text{Sp}(1)\times\text{Sp}(1)\cong\text{Spin}(4)$ subgroup of $\text{U}(4)$, and the matrices $\gamma_1,\ldots, \gamma_4$ can be conjugated by such unitaries.

In the basis where
\begin{equation}
\gamma_1=\begin{pmatrix} 0 & -\im\sigma_1 \\ \im\sigma_1 & 0\end{pmatrix},\quad \gamma_2=\begin{pmatrix} 0 & -\im\sigma_2 \\ \im\sigma_2 & 0\end{pmatrix}, \quad \gamma_3=\begin{pmatrix} 0 & -\im\sigma_3 \\ \im\sigma_3 & 0\end{pmatrix},\quad \gamma_4=\begin{pmatrix} 0 & 1 \\ 1 & 0\end{pmatrix}\nonumber
\end{equation}
as in \eqref{lineargamma}, the spectrally-flattened Hamiltonian is $\text{sgn}(H(k))=\hat{\vect{h}}(k)\cdot\vect{\gamma}=\begin{pmatrix} 0 & U(k)^\dagger \\ U(k) & 0\end{pmatrix}$ with 
\begin{equation}
U(k)=\begin{pmatrix} h_4+\im h_3 & h_2+\im h_1 \\ -h_2+\im h_1 & h_4-\im h_3  \end{pmatrix}, \qquad h_1^2+h_2^2+h_3^2+h_4^2=1.\label{specialunitary}
\end{equation}
We recognise \eqref{specialunitary} as a parametrisation of $\text{SU}(2)$. Thus $k\mapsto U(k)$ is a map from $T\setminus W\rightarrow\text{SU}(2)$, from which we can associate a gerbe and analyse the gerbe extension problem to $T$, using the semimetal MV-sequence as in Section \ref{sec:extension}.

\begin{remark}
Gauge transformations live in $\text{Sp}(1)\times\text{Sp}(1)=\text{SU}(2)\times\text{SU}(2)$, and they effect $U(k)\mapsto V_1(k)U(k)V_2^\dagger(k)$ for $(V_1(k), V_2(k))\in\text{SU}(2)\times\text{SU}(2)$.
\end{remark}

\begin{remark}
The roles of $\mathsf{S}$ and $\text{sgn}(H)$ can be exchanged, in the sense that we can write $\text{sgn}(H)=\begin{pmatrix}1 & 0 \\ 0 & - 1\end{pmatrix}$ and $\mathsf{S}=\begin{pmatrix}0 & U^\dagger \\ U & 0\end{pmatrix}$ instead. From this point of view, $\text{sgn}(H)$ is the difference of the spectral projections $p_\pm=\frac{1\pm\text{sgn}(H)}{2}$ onto the positive and negative eigenbundles $\mathscr{S}_\pm$ (over $T\setminus W$). Since $\mathsf{S}p_\pm\mathsf{S}^{-1}=p_\mp$,  $U$ gives an isomorphism between $\mathscr{S}_+$ and $\mathscr{S}_-$, c.f.\ Section IV of \cite{Read}.
\end{remark}


\section{New classes of $\ZZ_2$ semimetals}\label{sec:torsionsemimetal}
\subsection{Bilinear Hamiltonians in 5D}
There is another interesting class of $4\times 4$ Bloch Hamiltonians in 5D: those \emph{quadratic}, or \emph{bilinear}, in $\gamma$. As explained in Section \ref{sec:Dirac5D}, a quaternionic structure $\Theta$ singles out $\mathcal{H}_{0,l}$ in $\mathcal{H}_0$ as the $\Theta$-commuting (traceless) Hamiltonians, and we identify $\mathcal{H}_{0,l}$ with the tangent bundle $\cE$ of $T$. Furthermore, the $\Theta$-\emph{anticommuting} Hamiltonians form a complementary rank-10 subbundle $\mathcal{H}_{0,q}$ to $\mathcal{H}_{0,l}$ in $\mathcal{H}_0$, also with structure group $\text{SO}(5)$. In terms of the $\gamma$-matrices, an orthonormal basis for $\mathcal{H}_{0,q}$ is given by $\im \gamma_I$ where $I=\{ij\},1\leq i<j\leq 5$ is a multi-index and e.g.\ $\gamma_{12}\coloneqq\gamma_1\gamma_2=\frac{1}{2}[\gamma_1,\gamma_2]$. We should think of $\gamma_I$ as coming from $\{e_I=e_i\wedge e_j\}_I\subset\bigwedge^2\mathcal{E}$, just as $\gamma_i$ came from an orthonormal frame $\{e_i\}_i\subset\cE$. In accordance with terminology in the physics literature, elements of $\mathcal{H}_{0,q}$ are said to be particle-hole symmetric, with $\mathsf{C}=\Theta$. Thus $\mathsf{C}$-symmetric Hamiltonians are exactly the ``$\gamma$-bilinear'' ones --- they are ($\im$ times of) quantizations of $\bigwedge^2\mathcal{E}$, and we will refer to such Hamiltonians simply as \emph{bilinear Hamiltonians}.

{\bf Hamiltonians parametrised by a pair of vector fields.} For example, a \emph{pair} of vector fields $(\vect{a},\vect{b})$ defines an abstract Hamiltonian $\im c(\vect{a}\wedge\vect{b})=\frac{\im}{2}[c(\vect{a}),c(\vect{b})]$ in $Cl(\cE,g)\otimes_\RR\CC$. This is represented on $\mathscr{S}$ as the concrete Hamiltonian $H_{\vect{a},\vect{b}}=\im\tilde{c}(\vect{a}\wedge\vect{b})=\frac{\im}{2}[\vect{a}\cdot\vect{\gamma},\vect{b}\cdot\vect{\gamma}]=\im(\vect{a}\wedge\vect{b})_I\gamma_I$. Note that $\vect{a}\wedge\vect{b}$ is geometrically just an oriented area element, and is invariant under rotations in the plane of $\vect{a},\vect{b}$. Furthermore, $\vect{a}\wedge\vect{b}=0$ iff $\vect{a},\vect{b}$ are linearly dependent. Following \cite{thomas2}, a pair of vector fields $(\vect{a},\vect{b})$ over $T$ which are linearly independent everywhere is called a \emph{tangent $2$-field}; if $(\vect{a},\vect{b})$ is defined only on the complement of a finite number of points $W$, it is called a \emph{tangent $2$-field over $T$ with finite singularities}, or a  \emph{tangent $2$-field over $T\setminus W$}.

\begin{remark}\label{rmk:twoterm}
The general element of  $\bigwedge^2\RR^5$ is not a simple tensor $\vect{a}\wedge\vect{b}$ but a sum of two such tensors $\vect{a}\wedge\vect{b}+\vect{c}\wedge\vect{d}$, which can be taken to be orthogonal to each other --- this is a linear algebraic fact about canonical forms of antisymmetric matrices 
$$A\overset{\text{SO}(5)}{=}\begin{pmatrix}0 & \lambda & & \\ -\lambda & 0 & & & \\ & & 0 & \mu & \\ & & -\mu & 0 & \\ & & & & 0\end{pmatrix}.$$
Thus we can more generally consider $H=H_{\vect{a},\vect{b}}+H_{\vect{c},\vect{d}}$ for mutually orthogonal pairs of vector fields $(\vect{a},\vect{b})$ and $(\vect{c},\vect{d})$.
\end{remark}

\subsection{A $\mathsf{C}$-invariant four-band $\ZZ_2$-semimetal in 5D}\label{sec:Z2invariant}

Interestingly, the spectrum of $\gamma$-bilinear Hamiltonians $H=H_{\vect{a},\vect{b}}+H_{\vect{c},\vect{d}}$ can be found in much the same way as the linear ones, utilising only the Clifford algebra relations. 
\begin{proposition}
Let $H=H_{\vect{a},\vect{b}}+H_{\vect{c},\vect{d}}$ be a bilinear Hamiltonian as in Remark \ref{rmk:twoterm}, and let $\lambda=|\vect{a}\wedge\vect{b}|,\, \mu=|\vect{c}\wedge\vect{d}|$.
Then $\text{spec}(H(k))=\pm(\lambda(k)\pm\mu(k))$, and all four signs occur.
\end{proposition}
\begin{proof}Let $\lambda=|\vect{a}\wedge\vect{b}|, \mu=|\vect{c}\wedge\vect{d}|$. At the Clifford algebra level,
\begin{eqnarray}
c(\vect{a}\wedge\vect{b}+\vect{c}\wedge\vect{d})^2&=&c(\vect{a}\wedge\vect{b})^2+c(\vect{c}\wedge\vect{d})^2+2c(\vect{a}\wedge\vect{b}\wedge\vect{c}\wedge\vect{d})\nonumber\\
&=& -\lambda^2-\mu^2+2c(\vect{a}\wedge\vect{b}\wedge\vect{c}\wedge\vect{d}).
\end{eqnarray}
Thus, the square of $H=\im\tilde{c}(\vect{a}\wedge\vect{b}+\vect{c}\wedge\vect{d})$ is $H^2=\lambda^2+\mu^2-2\tilde{c}(\vect{a}\wedge\vect{b}\wedge\vect{c}\wedge\vect{d})$. If $\lambda$ or $\mu$ is zero, we are done. Otherwise, $\vect{a}\wedge\vect{b}\wedge\vect{c}\wedge\vect{d}=\lambda\mu\,\tilde{c}(e_1e_2e_3e_4)=\lambda\mu\,\tilde{c}(e_5)$ for some positively-oriented orthonormal frame $\{e_i\}_{i=1,\ldots,5}$, where we have suppressed the wedge product notation and used the chirality condition $\tilde{c}(e_1e_2e_3e_4e_5)=1$. Since $\tilde{c}(e_5)$ has spectrum $\pm 1$ with each doubly-degenerate, it follows that $\text{spec}(H^2)=\lambda^2+\mu^2\pm 2\lambda\mu=(\lambda\pm\mu)^2$ and $\text{spec}(H)=\pm(\lambda\pm\mu)$, where we have suppressed the dependence on $k\in T$.
\end{proof}

The zeroes of $H$ are easily read off, and there are two types: (1) when $\lambda(k)=\mu(k)=0$, and (2) when $\lambda(k)=\mu(k)\neq 0$. The former involves a four-fold degenerate zero, so all four bands cross, whereas the latter involves crossing between only a pair of bands (Fig.\ \ref{fig:four-band-crossing}).

\begin{figure}
    \centering
    \includegraphics[width=0.7\textwidth]{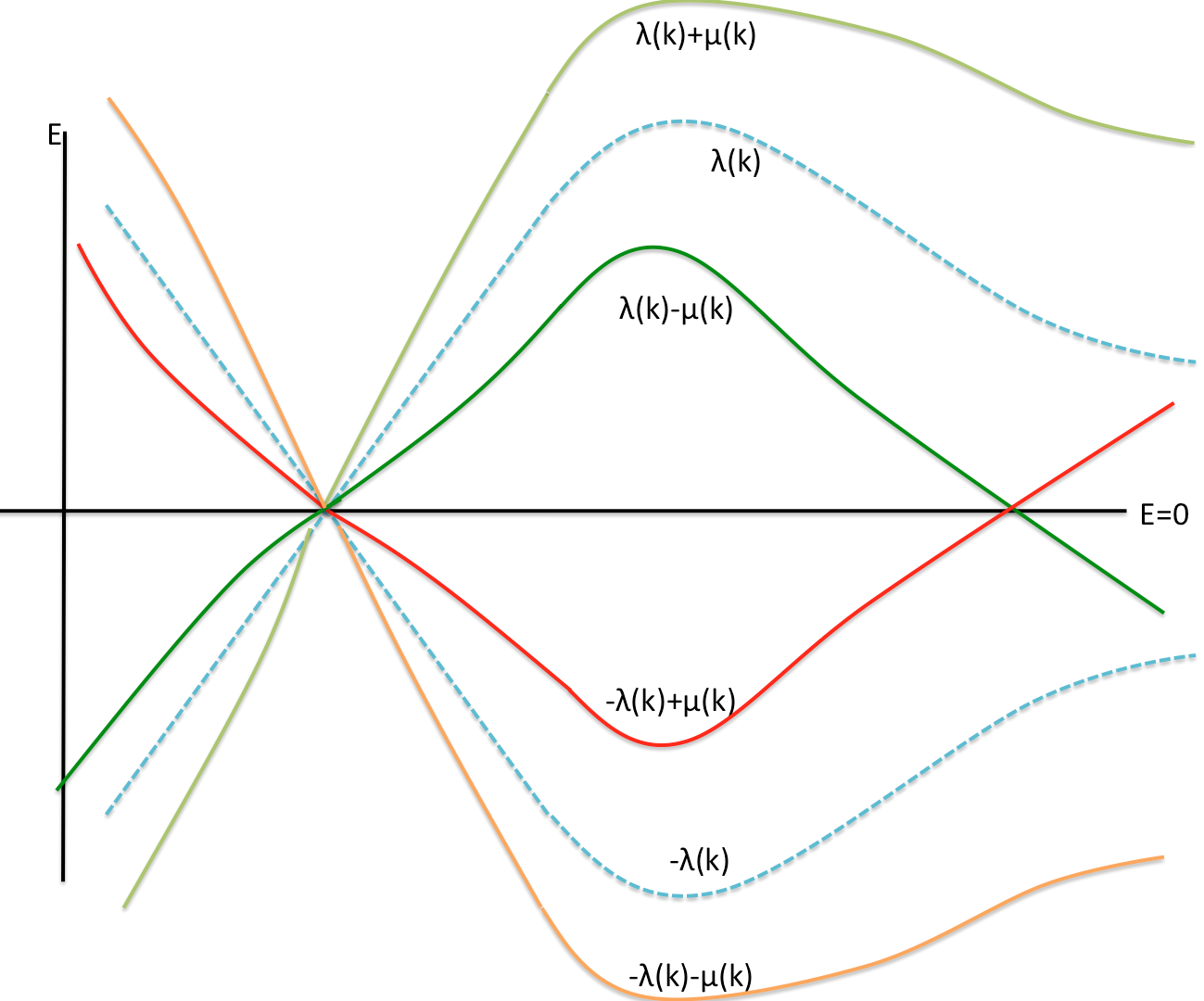}
    \caption{Spectrum of $\gamma$-quadratic Hamiltonians. If $\mu=0$, the spectrum degenerates to the dotted blue lines, which are each 2-fold degenerate.}
    \label{fig:four-band-crossing}
\end{figure}

When $\mu$ is zero, the four-band spectrum becomes a pair of doubly-degenerate bands, and these degenerate pairs cross exactly at the points where $\lambda=0$, much like the spectrum of the Dirac-type Hamiltonians which are linear in $\gamma$. We are interested in such ``bilinear Hamiltonians'', and they are precisely the pure tensor ones $H_{\vect{a},\vect{b}}(k)=\frac{\im}{2}[\vect{a}(k)\cdot\vect{\gamma},\vect{b}(k)\cdot\vect{\gamma})]=\im\sum_I (\vect{a}(k)\wedge\vect{b}(k))_I\gamma_I$, specified by a \emph{single} pair of vector fields $(\vect{a}\wedge\vect{b})$, with $\text{spec}(H_{\vect{a},\vect{b}}(k))=\pm|\vect{a}\wedge\vect{b}|(k)=\pm\lambda(k)$.

We will \emph{assume} that the singularities (points where $\vect{a}\wedge\vect{b}$ vanishes) form a finite set of isolated points $W$. In other words, $(\vect{a},\vect{b})$ defines a tangent 2-field over $T\setminus W$. At a singularity $w\in W$, there may be topological obstructions to extending $(\vect{a},\vect{b})$ across $w$, in which case, we cannot open up a gap at the four-band crossing at $w$ by simply deforming $\lambda$ into a strictly positive function. Despite the spectrum resembling that of a (linear) Dirac-type Hamiltonian, this type of four-band crossing can be topologically protected by a very different and much more subtle mechanism, see Section \ref{sec:ATinvariant}.

\begin{remark}
We may also consider bilinear Hamiltonians $H_{\vect{a},\vect{b}}$ specified by a pair of vector fields $(\vect{a},\vect{b})$ over a $d$-manifold with $d\neq 5$. However, it is no longer the case that a $\mathsf{C}$-symmetry singles out such Hamiltonians, and one needs to motivate such model Hamiltonians in a different way. They may nevertheless be studied as a toy model.
\end{remark}

\subsection{Atiyah--Thomas $\ZZ_2$-invariant}\label{sec:ATinvariant}
A detailed study of tangent 2-fields with finite singularities was carried out in \cite{thomas,thomas2,AD,Korbas}, where analogues of the Poincar\'{e}--Hopf theorem for vector fields were obtained. For this, we need the notion of the local index at a singularity $w\in W$ of a tangent 2-field $(\vect{a},\vect{b})$ on $T\setminus W$. Let $S_w^4$ be a small oriented 4-sphere surrounding $w\in W$, then $(\vect{a},\vect{b})$ gives a map $S_w^4\rightarrow W_{5,2}$ to the non-compact Stiefel manifold $W_{5,2}=\text{GL}(5,\RR)/\text{GL}(3,\RR)$ of 2-frames in $\RR^5$. There is a homotopy equivalence between $W_{5,2}$ and the (compact) Stiefel manifold $\mathcal{V}_{5,2}=\text{SO}(5)/\text{SO}(3)$ of orthonormal 2-frames in $\RR^5$. The \emph{local index of $(\vect{a},\vect{b})$ at $w$} is defined to be the homotopy class of $(\vect{a},\vect{b})|_{S_w^4}$ in $\pi_4(W_{5,2})\cong\pi_4(\mathcal{V}_{5,2})\cong\ZZ_2$, and is the local obstruction to deforming $(\vect{a},\vect{b})$ to remove the singularity at $w$.

Recall from the previous subsection the construction of a bilinear Dirac Hamiltonian $H_{\vect{a},\vect{b}}$ from a tangent 2-field $(\vect{a},\vect{b})$ over $T$ with finite singularities $W$, and that four-band crossings occur precisely at $W$. The local index at $w$ of $(\vect{a},\vect{b})$ has the physical interpretation as a $\ZZ_2$ local topological charge protecting the band crossing $w\in W$.

To understand $\pi_4(\mathcal{V}_{5,2})\cong \ZZ_2$ better, recall the fibration 
\begin{equation}
S^3=\mathcal{V}_{4,1}\rightarrow\mathcal{V}_{5,2}\rightarrow\mathcal{V}_{5,1}=S^4,\nonumber
\end{equation}
where the $S^4$ base parametrises the choice of a first normalised vector $\hat{\vect{a}}$, and the fiber parametrises (locally) the choice of $\hat{\vect{b}}$ orthonormal to $\hat{\vect{a}}$. From the homotopy exact sequence,
\begin{equation}
\ldots \rightarrow\pi_4(S^3)=\ZZ_2\rightarrow \pi_4(\mathcal{V}_{5,2})=\ZZ_2\rightarrow \pi_4(S^4)=\ZZ\rightarrow \pi_3(S^3)=\ZZ\rightarrow\ldots\nonumber
\end{equation}
we see that $\pi_4(\mathcal{V}_{5,2})=\ZZ_2$ comes from the inclusion of the fiber $S^3=\mathcal{V}_{4,1}$ which has the famous $\pi_4(S^3)=\pi_4(\text{SU}(2))=\ZZ_2$.\nonumber

Thus, we can construct an explicit example of a $\ZZ_2$-protected four-band crossing as follows. Choose some local Euclidean coordinates $k$ for a neighbourhood $\mathcal{O}$ of $w$ centred at 0, and a trivialization $e_1,\ldots,e_5$ of the tangent bundle above $\mathcal{O}$. Take $\hat{\vect{a}}=e_1$, and for $k\in S_w^4$, we choose $\hat{\vect{b}}$ orthonormal to $\hat{\vect{a}}$ by letting $k\mapsto\hat{\vect{b}}(k)\in S^4$ be given by a generator of $\pi_4(\mathcal{V}_{4,1})=\pi_4(S^3)$. Extend to orthogonal vectors $\vect{a}, \vect{b}$ to all of $\mathcal{O}$ by scaling by $|k|$. Thus $H_{\vect{a},\vect{b}}(k)=\im|k|\sum_I(\hat{\vect{a}}\wedge\hat{\vect{b}}(k))_I\gamma_I$ on $\mathcal{O}\setminus\{w\}$ and $H_{\vect{a},\vect{b}}(w)=0$. Then $w$ is a four-band crossing with non-trivial $\ZZ_2$-charge.

\begin{remark}
By taking $\frac{\vect{a}\wedge\vect{b}}{|\vect{a}\wedge\vect{b}|}$ we obtain a \emph{2-plane field}\footnote{2-plane fields are also called \emph{oriented $2$-distributions}, see \cite{Korbas} for a guide to $k$-distributions on manifolds.} defining for each $k\in T\setminus W$ a point in the oriented Grassmannian $\cG^+_{5,2}$ of 2-planes in $\RR^5$. For a 2-plane field over $T\setminus W$, there is similarly a map $S_w^4\rightarrow\mathcal{G}^+_{5,2}$, and a local index can be defined at $w\in W$. Note that $\mathcal{G}^+_{5,2}=\mathcal{V}_{5,2}/\text{SO}(2)$, and at the level of $\pi_4$ the exact sequence
\begin{equation}
\ldots \rightarrow 0=\pi_5(\text{SO}(2))\rightarrow\pi_4(\mathcal{V}_{5,2})\xrightarrow{\cong}\pi_4(\mathcal{G}^+_{5,2})\rightarrow 0=\pi_4(\text{SO}(2))\rightarrow\ldots\nonumber
\end{equation}
says that the index is in $\pi_4(\mathcal{V}_{5,2})\cong\pi_4(\mathcal{G}^+_{5,2})\cong\ZZ_2$. Although we can consider bilinear Hamiltonians specified by a field of oriented area elements, i.e. a 2-plane field along with areas (magnitudes), which will also have $\ZZ_2$ local indices, we only consider those specified by a tangent 2-field (with finite singularities) in this paper. 

\end{remark}

For tangent 2-fields in $d\geq 3$-dimensions, the local topological index at a point singularity is defined similarly, as an element of $\pi_{d-1}(\cV_{d,2})$. These homotopy groups are known to be $\ZZ_2$ when $d$ is odd and $\ZZ\oplus \ZZ_2$ when $d$ is even \cite{paechter}. Unlike the case of vector fields, the global constraint for these indices depends on $d$ (mod 4) \cite{thomas,thomas2,Atiyah,AD} and in the following subsection, we will restrict to $d=1$ (mod 4) .

\subsection{Kervaire semicharacteristic, Kervaire chains and $\ZZ_2$-Fermi arcs}\label{sec:Kervairestructures}
For this section, we assume that $T$ is a compact oriented $d$-manifold $T$ with $d=1$ (mod 4) and $d\geq 5$. Most of the constructions of Section \ref{sec:Euler} can be carried out for $\ZZ_2$ coefficients, and there is also a generalisation of the Poincar\'{e}--Hopf theorem for tangent 2-fields and 2-plane fields with finite singularities \cite{thomas,thomas2,Atiyah,AD}. This suggests that Hamiltonians parametrised by a tangent 2-field, such as the 5D $\ZZ_2$-Dirac semimetals introduced in Section \ref{sec:Z2invariant}, can have topologically protected crossings which are furthermore constrained by the global Atiyah--Dupont condition. If $\kappa(T)=0$, we expect there to be {\em $\ZZ_2$-charge cancellation} and also a notion of {\em $\ZZ_2$-Fermi arcs}. 

Let us simplify notation and write $\cP$ for the tangent 2-field $(\vect{a},\vect{b})$ over $T\setminus W$, and $\text{Ind}_w(\cP)$ for its local $\ZZ_2$ charge at $w\in W$, which is assumed to be a finite set as before. The sum $\text{Ind}(\cP)=\sum_{w\in W}\text{Ind}_w(\cP)$ is a global obstruction to deforming $\mathcal{P}$ to a tangent 2-field on all of $T$. In \cite{AD}, Atiyah--Dupont proved that $\text{Ind}(\cP)$ is equal to the \emph{real Kervaire semicharacteristic}
$\kappa(T)$, defined as
\begin{equation}
\kappa(T)=\sum_{q\,\text{even}} \text{dim}_\RR \,H^q(T,\RR)\;\;\text{mod }2.\nonumber
\end{equation}
Thus $\text{Ind}(\cP)=\kappa(T)$ for any $\cP$. 
It is known that the manifold $T$ admits tangent 2-fields $\cP$ with {\em finite} singularity set $W$ if and only if the Stiefel--Whitney class $w_{d-1}(T)$ vanishes,
(\cite{thomas2}), which holds for instance when $T=\TT^d$, in which case the Kervaire semicharacteristic vanishes, $\kappa(\TT^d)=0$.

\begin{remark}
Although we have assumed the bilinear Hamiltonians to be specified by tangent 2-fields, we could also work more generally with 2-fields for other vector bundles using a generalisation of the Atiyah--Dupont theory in \cite{TangZhang}. 
\end{remark}

\begin{definition}
Let $T$ be a compact oriented $d$-manifold $T$ with $d=1$ (mod 4) and $\kappa(T)=0$. For a 0-chain $\cW\in C_0(T,\ZZ_2)$ defined on $W$ and with total weight zero, a \emph{Kervaire chain} for $\cW$ is a 1-chain $l\in C_1(T,\ZZ_2)$ such that $\partial l=\cW$, thus $[l]\in H_1(T,W,\ZZ_2)$. We write $\kerv(T,\cW)$ for the subset of homology classes of Kervaire chains for $\cW$. 
\end{definition}

The $\ZZ_2$-versions of the semimetal MV sequence and the dual homology sequence are
\begin{equation}
\small
    \xymatrix{ 0\ar[r]&H^{d-1}(T,\ZZ_2)\ar[d]_{\text{PD}}\ar[r] & H^{d-1}(T\setminus W,\ZZ_2)\ar[d]_{\text{PD}}\ar[r]^{\beta} &H^{d-1}(S_W,\ZZ_2)\ar[d]_{\text{PD}}\ar[r]^{\Sigma} & H^d(T,\ZZ_2)\ar[d]_{\text{PD}}\ar[r] & 0\\
   0\ar[r] & H_1(T,\ZZ_2)\ar[r]& H_1(T,W,\ZZ_2)\ar[r]^{\partial} & H_0(W,\ZZ_2) \ar[r]^{\Sigma} & H_0(T,\ZZ_2)\ar[r] & 0
    }.\nonumber
\normalsize
\end{equation}
As before, we define $\cokerv(T,\cW)\subset H^{d-1}(T\setminus W,\ZZ_2)$ to be the Poincar\'{e} dual image of $\kerv(T,\cW)$. We will then explain how a tangent 2-field with finite singularity set $W$ has a topological invariant in $\cokerv(T,\cW)$ and therefore a Kervaire chain representation.

\begin{definition}
Let $T$ be a compact oriented $d$-manifold $T$ with $d=1$ (mod 4) and $\kappa(T)=0$, and $\cS $ be the sphere bundle for its tangent bundle. A \emph{Kervaire structure} for $T$ is a class in $H^{d-1}(\cS ,\ZZ_2)$ whose restriction to each fibre $\cS_k, k\in T$ generates $H^{d-1}(S^{d-1},\ZZ_2)$. Define $\widetilde{\cokerv}(T)$ to be the set of Kervaire structures for $T$.
\end{definition}
From the Gysin sequence \eqref{gysin} with $\ZZ_2$-coefficients, we see that $H^{d-1}(T,\ZZ_2)$ acts on $\widetilde{\cokerv}(T)$ freely and transitively by pullback and addition. Let $\mV_{d,2}T$ be the bundle of Stiefel manifolds $\cV_{d,2}$ associated to the orthonormal frame bundle of $T$ (endowed with some Riemannian metric). Corresponding to the fibration $S^{d-2}=\cV_{d-1,1}\rightarrow \cV_{d,2}\rightarrow \cV_{d,1}=S^{d-1}$, there is a fibration
\begin{equation}
	S^{d-2}\rightarrow\mV_{d,2}\overset{\rho}{\longrightarrow}\cS .\nonumber
\end{equation}
A (non-singular) tangent 2-field $\cP$ on $T$ gives rise to a Kervaire structure as follows. Let $\widehat{\cP}$ denote the section of $\mV_{d,2}$ determined by $\cP$, which is a $d$-cycle on $\mV_{d,2}$. The Poincar\'{e} dual is a $(2d-3)$-cocycle on $\mV_{d,2}$, and pushforward under $\rho$ gives a $(d-1)$-cocycle in $\cS $. Furthermore, the construction of the cocycle is such that its restriction to $\cS_k$ generates $H^{d-1}(S^{d-1},\ZZ_2)$ for each $k\in T$.

Now consider tangent 2-fields on $T$ with finite singularities $W$. As in Section \ref{sec:semimetalEulerchain}, we can define a local charge operator $$\tilde{\beta}:H^{d-1}(\cS|_{T\setminus W},\ZZ_2)\rightarrow H^{d-1}(S_W,\ZZ_2)\overset{\text{PD}}{\longleftrightarrow} H_1(W,\ZZ_2).$$ Then for each 0-chain $\cW$ with total charge zero, we define $\widetilde{\cokerv}(T,\cW)$ to be the subset of $H^{d-1}(\cS_{T\setminus W},\ZZ_2)$ whose restriction to each fibre over $k\in T\setminus W$ generates $H^{d-1}(S^{d-1},\ZZ_2)$, and whose local charges is $\text{PD}(\cW)$. Each $\widetilde{\cokerv}(T,\cW)$ is an affine space for $H^{d-1}(T,\ZZ_2)$. Suppose $\cP$ is a tangent 2-field on $T\setminus W$ with local charge 0-chain $\cW$, then it determines a $d$-cycle in $\mV_{d,2}|_{T\setminus D_W}$ relative to the boundary $\mV_{d,2}|_{S_W}$, which Poincar\'{e} dualises to a $(2d-3)$-cocycle on $\mV_{d,2}|_{T\setminus D_W}$. Pushforward under $\rho$ gives a $(d-1)$-cocycle on $\cS|_{T\setminus D_W}$ representing a class in $\widetilde{\cokerv}(T,\cW)$. Fix a reference non-singular vector field $\vect{h}_{\text{ref}}$, then $\vect{h}_{\text{ref}}^*$ identifies $\widetilde{\cokerv}(T,\cW)$ with $\cokerv(T,\cW)$. 

The following diagram summarises the various affines spaces for $H_1(T,\ZZ_2)\cong H^{d-1}(T,\ZZ_2)$:
\begin{equation}
\xymatrix{  & H_1(T,\ZZ_2)\ar[r] \ar[d]& \kerv(T,\cW)\ar[l] \ar[d]& \subset\quad\;\; H_1(T,W,\ZZ_2) \\
 & H^{d-1}(T,\ZZ_2)\ar[r]\ar[d] \ar[u]& \cokerv(T,\cW)\ar[l] \ar[d]\ar[u]& \subset  H^{d-1}(T\setminus W,\ZZ_2) \\
H^{d-1}(\cS,\ZZ_2)\supset & \widetilde{\cokerv}(T)\ar[r] \ar[u]& \widetilde{\cokerv}(T,\cW)\ar[l] \ar[u]& \subset H^{d-1}(\cS|_{T\setminus W},\ZZ_2)
}.\nonumber
\end{equation}

\begin{definition}
Let $H$ be a bilinear Hamiltonian specified by a tangent 2-field $\cP$ over a compact oriented 5-manifold $T$ with $\kappa(T)=0$, with finite singularities $W$. It has a cohomological topological invariant $[\omega_{\cP}]\in H^4(T\setminus W,\ZZ_2)$ defined as in the previous paragraph, and its \emph{Kervaire chain representation} is the Poincar\'{e} dual $[l_\cP]\in H_1(T,W,\ZZ_2)$.
\end{definition}

\subsubsection*{$\ZZ_2$-Fermi arcs}
The bulk-boundary correspondence proceeds exactly as before, taking a Kervaire chain for $\cW$ to a ``$\ZZ_2$-Fermi arc'' with topological invariant in $H_1(\widetilde{T},\widetilde{W},\ZZ_2)$. Unlike the usual $\ZZ$-Fermi arcs, the $\ZZ_2$-Fermi arcs cancel in pairs so that there is at most one arc between two projected band crossings, and the arc never fully winds around a 1-cycle in $\wt{T}$.

\vspace{1em}

\remark{It is known that $\mathsf{T}$-invariant insulators in 2D and 3D, where $\mathsf{T}$ now refers to ordinary time-reversal rather than the combination $\mathsf{T}\mathsf{P}$), are classified by certain $\ZZ_2$-valued \emph{Fu--Kane--Mele} (FKM) invariants \cite{FKM,dNG}. In ongoing work, we are applying the ideas of this paper to understand $\mathsf{T}$-invariant semimetals, which were experimentally realised in \cite{Lv,Xu1,Kauf}. In this setting, the band crossings appear at conjugate-momentum pairs, and the points $w_i, \tau(w_i)$ in each pair have the same \emph{integer} charge. The relation of these local $\ZZ$ charges to FKM $\ZZ_2$-invariants is a subtle generalisation of the relation to weak Chern numbers in the $\mathsf{T}$-broken case. We emphasize that the $\ZZ_2$ invariants there are mathematically \emph{different} from the ones due to Atiyah--Dupoint--Thomas and the Kervaire semicharacteristic, and arise under physically different symmetry conditions.}

\section{Outlook and acknowledgements}\label{sec:Outlook}
The topological semimetal-insulator relationship encoded in the semimetal MV sequence can, in principle, be refined using an MV sequence in \emph{differential cohomology} \cite{Simons}. Another refinement comes from additional spatial symmetries. For example, a spatial inversion symmetry $\mathsf{P}$ induces an action on the Brillouin torus $T$ under which $W$ should be invariant. The same action arises if $\mathsf{T}$ rather than $\mathsf{T}\mathsf{P}$ symmetry is present. The vector field $\vect{h}$ needs to be suitably equivariant, and the local index at $w\in W$ is an equivariant degree. Finally, the $K$-theory MV sequence can be used (with $T\setminus W$ replaced by $T\setminus D_W$) to analyse stable topological semimetal invariants. This paves the way for the development of a new type of T-duality transformation $K^*(T\setminus D_W)\Leftrightarrow K^{*+d}(T,W)$, which would provide a T-dual picture of semimetals along the lines of \cite{MT2,MT3,MT1} in the case of insulators. 

We end by posing the question: associated to the MV sequence analysis in the case of Kervaire semimetals (Section \ref{sec:Kervairestructures}), is there a geometric extension problem?

\bigskip

\noindent{\it Acknowledgements}. We thank Keith Hannabuss for some helpful discussions, and Siye Wu for his useful criticisms. This work was supported by the Australian Research Council via ARC Discovery Project grant DP150100008 and ARC DECRA grant DE170100149.


\begin{appendices}

\section{Hopf bundles, basic gerbe, Stiefel manifold}\label{appendix}

To be self-contained, we recall here some universal constructions of geometric objects with connections that are used in the main text.\\

\noindent {\bf The Hopf line bundle over $\CC \PP^1$.} (Sections \ref{sec:physback} and \ref{sec:extension}) Recall that the 2D-sphere $\CC \PP^1$ is defined as the space of complex lines in $\CC^2$. The Hopf line bundle $\cL$ over $\CC \PP^1$ is defined as the sub-bundle of the trivial bundle $\CC \PP^1 \times \CC^2$ given by
$$
\cL = \{(\ell, v) \in \CC \PP^1 \times \CC^2 \big| v\in \ell\}.
$$
As a sub-bundle of the trivial bundle $\CC \PP^1 \times \CC^2$
(which has the obvious Hermitian metric and unitary connection) $\cL$ comes equipped with a Hermitian metric and unitary connection $\nabla^\cL$.  Explicitly, a natural projection defining $\cL$ is given by $P(\vect{n})= \frac{1}{2}\left(1+ \vect{n}\cdot\vect{\sigma}\right)$,
where $\vect{n} \in S^2$ is the unit vector in $\RR^3$ and $\vect{\sigma}=(\sigma_1, \sigma_2, \sigma_3)$ are the Pauli matrices, $\sigma_j \in M_2(\CC)$ for $j=1,2,3$. 
Then  $\nabla^\cL= PdP$ and the curvature is $F_\cL = PdPdP.$

The 1st Chern class classifying $\cL$ up to isomorphism is, $$c_1(\cL) = \frac{1}{2\pi i}\int_{\CC \PP^1} F_{\cL} =1.$$ 

\noindent {\bf The basic gerbe over $S^3$.} (Sections  \ref{sec:semimetalgerbe}  and \ref{sec:extension})
We begin by briefly recalling the definition of bundle gerbes (cf. \cite{MS}), which are 
geometric examples of more abstract gerbes (cf. \cite{Brylinski}). They are higher analogs of line bundles,
and are classified up to (stable) equivalence by the Dixmier--Douady class in integer-valued degeree 3 cohomology (which is the analog of the 
Chern class for line bundles).
Bundle gerbes $(\cG, \cL)$ are specified by the following data:

\begin{itemize}
\item A smooth fiber bundle of manifolds
\begin{equation*}
\xymatrix{
Z\ar[r]&Y\ar[d]^\pi\\
& X.
}
\end{equation*}
where $Z, Y$ are not required to be finite dimensional.

\item A primitive line bundle  $\cL$ over $Y^{[2]} = Y\times_X Y$, in the sense that under
lifting by the three projection maps $\pi_j: Y^{[3]} \to Y^{[2]}$, which omits the 
$j$-th factor,
\begin{equation*}
\xymatrix{Y^{[3]}\ar@<2ex>[r]^{\pi_1}
\ar[r]^{\pi_2}
\ar@<-2ex>[r]^{\pi_3}& Y^{[2]}}
\end{equation*}
%
there is an isomorphism of line bundles on $Y^{[3]}$
\begin{equation}\label{gerbeprod}
\pi_3^*\cL\otimes\pi_1^*\cL\cong \pi_2^*\cL.
\end{equation}

\item The equation above can be viewed as a product, and this product is required to be associative, which is a compatibility condition on $Y^{[4]}$,
\begin{equation*}
\xymatrix{
 \cL_{(y_1,y_2)} \otimes \cL_{(y_2,y_3)}  \otimes \cL_{(y_3,y_4)}\ar[r]\ar[d]^-{\sigma}& \cL_{(y_1,y_3)}  \otimes \cL_{(y_3,y_4)} \ar[d]^-{\sigma}\\
\cL_{(y_1,y_2)}  \otimes \cL_{(y_2,y_4)} \ar[r]&\cL_{(y_1,y_4)},
}
\end{equation*}
for all $(y_1, y_2, y_3, y_4) \in Y^{[4]}$.
\end{itemize}
A {\em connection} on such a bundle gerbe is a unitary connection on the line bundle $\cL$ with curvature equal to $\pi_2^*(B)-\pi_1^*(B)$, for some 2-form $B$ on $Y$. Necessarily, $dB=\pi^*H$, where $H$ is a closed 3-form on $X$ called the curvature of the gerbe $(\cG,\cL)$, and is a de Rham representative of the Dixmier--Douady class.

Next we describe the {\em basic (bundle) gerbe} on $\text{SU}(2)$, used in the main text.
We identify $S^3\cong \text{SU}(2)$ and consider $Y \subset \text{U}(1)\times \text{SU}(2)$ defined by 
$$
Y = \{(z, g) \in \text{U}(1)\times \text{SU}(2) \big| z\ne 1, z \notin {\rm spec}(g)\}
$$
Then the projection to the 2nd factor is a submersion $pr_2:Y\longrightarrow \text{SU}(2)$ and the fibred product
$$
Y^{[2]} = \{(z_1, z_2, g) \in \text{U}(1)^2  \times \text{SU}(2) \big| z_1, z_2 \ne 1, z_1, z_2 \notin {\rm spec}(g)\}.
$$
For $(\lambda, g) \in \text{U}(1)\times \text{SU}(2), \, \lambda\ne 1$, let $E_{(\lambda, g)}$ denote the $\lambda$ eigenspace of $g$
and 
$$
E_{(z_1, z_2, g)} = \bigoplus E_{(\lambda, g)}
$$
where $\lambda$ is an eigenvalue lying in the open arc component in $\text{U}(1)$ joining $z_1, z_2$ that does not contain $1$. Let $\cL_{(z_1, z_2, g)} = \det(E_{(z_1, z_2, g)})$ denote the determinant line, which is the highest exterior power. Then 
$\cG=(\cL, Y)$ is a bundle gerbe description of the basic gerbe over $S^3\cong \text{SU}(2)$, see \cite{MS}, \cite{Gaw}. It comes with gerbe connection (Theorem 5.1 \cite{MS2}), with 2-form $B$ on $Y$ defined as,
\begin{equation} 
B_{(z,g)} = \frac{1}{8\pi^2}\oint_{C_{(z,g)}}
\log_z \xi \, {\rm tr}\left((\xi - g)^{-1}dg(\xi - g)^{-2}dg\right) d\xi.
\end{equation}
Here  for any $(z,g)\in Y$,  $C_{(z,g)} \subset \CC \setminus R_z$  is a choice of anti-clockwise 
closed contour which encloses ${\rm spec}(g)$, where for each complex 
number $z$ with $|z| = 1$, $R_z$ denotes the closed ray from the origin through $z$, and
one defines the branch of the logarithm, $\log_z \colon \CC \setminus R_z \to \CC $,  
by making the cut along $R_z$ and also setting $\log_z(1) = 0$.

Then the curvature $H$ of the gerbe $(\cG, \cL)$  is a universal multiple of the Wess--Zumino--Witten (WZW) integrand 
${\rm tr}((g^{-1} dg)^3)$ on  $\text{SU}(2)$, whose Dixmier--Douady number classifying the gerbe up to (stable) equivalence is,
$$
DD(\cG) = - \frac{1}{24\pi^2} \int_{\text{SU}(2)} {\rm tr}((g^{-1} dg)^3) =1.
$$
\\
\noindent {\bf The quaternionic Hopf line bundle over $\HH \PP^1$.} (Sections \ref{sec:5Dsemimetal} and \ref{sec:extension}) Recall that the 4D-sphere $\HH \PP^1$ is defined as the space of quaternionic lines in $\HH^2$ and that the quaternionic Hopf line bundle $\cH$ over $\HH \PP^1$ is defined as the sub-bundle of the trivial bundle $\HH \PP^1 \times \HH^2$ given by
$$
\cH = \{(\ell, v) \in \HH \PP^1 \times \HH^2 \big| v\in \ell\}.
$$
As a sub-bundle of the trivial bundle $ \HH \PP^1 \times \HH^2$ (which has the obvious Hermitian metric and unitary connection) $\cH$ comes equipped with a Hermitian metric and unitary connection  $\nabla^\cH$. 
Explicitly, a natural projection defining $\cH$ is given by $P(\vect{n})= \frac{1}{2}\left(1+ \vect{n}\cdot\vect{\gamma}\right)$,
where $\vect{n} \in S^4$ is the unit vector in $\RR^5$ and $\vect{\gamma}=(\gamma_1, \gamma_2, \gamma_3, \gamma_4, \gamma_5)$ are gamma matrices, with $\gamma_j \in M_2(\HH)\cong Cl_{0,4}$ for $j=1,2,3,4$ generating the Clifford algebra and $\gamma_5=\gamma_1\gamma_2\gamma_3\gamma_4$. 
Then  $\nabla^\cH= PdP$ and the curvature is $F_\cH = PdPdP.$

The 2nd Chern number classifying $\cH$ up to isomorphism is, 
$$c_2(\cH) =\frac{1}{4\pi^2} \int_{\HH \PP^1} {\rm tr}(F_{\cH}\wedge F_{\cH}) =1.$$ 
\\

\noindent {\bf The Stiefel manifold $\cV_{5,2}$.} (Section \ref{sec:torsionsemimetal}) It can be defined as the space of all orthonormal 2-frames in $\RR^5$. As a homogeneous space, $\cV_{5,2} = SO(5)/SO(3)$, where we notice that it is a compact, connected, oriented $7$-dimensional manifold,  which can be identified with the unit sphere $S(S^4)$ subbundle  of the tangent bundle of the sphere $S^4$. So 
$\cV_{5,2}$ is a non-principal $SO(4)/SO(3)=S^3 $-bundle over $S^4$ with Euler class equal to $2$.
The Pontryagin number of the bundle $S^3 \hookrightarrow \cV_{5,2} \twoheadrightarrow S^4$ is just the Pontryagin number of $S^4$, which is zero since 
the sphere $S^4$ is the oriented boundary of the $5$ dimensional ball. By \cite{Milnor}, the Euler and Pontryagin numbers together determine the bundle up to isomorphism.
Via the long exact sequence for the fibration $S^3 \hookrightarrow \cV_{5,2} \twoheadrightarrow S^4$, we see that there is a natural isomorphism, (this is related to \cite{LinYau})
$$\pi_4(\cV_{5,2}) \cong \pi_4(S^3)\cong \ZZ_2.$$
The generator of $\pi_4(S^3)$ determines the non-trivial $S^3$-bundle over $S^5$ via the clutching construction, $$SU(2) \hookrightarrow SU(3) \twoheadrightarrow S^5.$$
By the Pontryagin--Thom construction, $\pi_4(S^3)\cong \ZZ_2$ comes from $\pi_1(\text{SO}(3))$ and so relates to the t' Hooft--Polyakov $\text{SO}(3)$ monopole \cite{tHooft,Poly}.

\bigskip

\end{appendices}



\begin{thebibliography}{10}


\bibitem{Atiyah} 
Atiyah, M.F.: Vector fields on manifolds. Arbeitsgemeinschaft f\"{u}r Forschung des Landes Nordrhein-Westfalen, vol. 200, VS Verlag, Cologne pp. 7--26 (1970) 

\bibitem{AD}
Atiyah, M.F., Dupont, J.L.: Vector fields with finite singularities. Acta Math. {\bf 128}(1) 1--40 (1972)

\bibitem{AR}
Atiyah, M.F., Rees, E.: Vector Bundles on Projective 3-space. Invent. Math. {\bf 35} 131--153 (1976)

\bibitem{Avila}
Avila, J.C., Schulz-Baldes, H., Villegas-Blas, C.: Topological Invariants of Edge States for Periodic Two-Dimensional Models. Mathematical Physics, Analysis and Geometry {\bf 16}(2) 137--170 (2013)

\bibitem{ASSS}
Avron, J.E., Sadun, L., Segert, J., Simon, B.: Topological Invariants in Fermi Systems with Time-Reversal Invariance. Phys. Rev. Lett. {\bf 61} 1329 (1988)

\bibitem{ASSS2}
Avron, J.E., Sadun, L., Segert, J., Simon, B.: Chern numbers, quaternions, and Berry's phases in Fermi systems. Commun. Math. Phys. {\bf 124}(4) 595--627 (1989)

\bibitem{BH}
Borel, A., Hirzebruch, F.: Characteristic classes and homogeneous spaces, I. Amer. J. Math. {\bf 80}(2) 458--538 (1958)

\bibitem{BM}
Borel, A., Moore, J.C.: Homology theory for locally compact spaces. Michigan Math. J. {\bf 7}(2) 137--159 (1960)

\bibitem{BottTu}
Bott, R., Tu, L.W.: Differential Forms in Algebraic Topology. Grad. Texts in Math. {\bf 82}, Springer--Verlag, New York, (1982)

\bibitem{Bradlyn}
Bradlyn, B., Cano, J., Wang, Z., Vergniory, M.G., Felser, C., Cava, R.J., Bernevig, B.A.: Beyond Dirac and Weyl fermions: Unconventional quasiparticles in conventional crystals. Science {\bf 353}(6299) (2016)

\bibitem{Brylinski}
Brylinski, J-L.:
Loop spaces, characteristic classes and geometric quantization. 
Progress in Mathematics, {\bf 107}, Birkhauser Boston, Inc., Boston, MA, (1993)

\bibitem{Burg}
Burghelea, D., Haller, S.: Euler structures, the variety of representations and the Milnor--Turaev torsion. Geom. Topol. {\bf 10} 1185--1238 (2006)

\bibitem{Chriss}
Chriss, N., Ginzburg, V.: Representation Theory and Complex Geometry. Birkh\"{a}user Boston (1997)

\bibitem{Gaw}
Carpentier, D., Delplace, P., Fruchart, M., Gaw\k{e}dzki, K.: Topological index for periodically driven time-reversal invariant 2D systems. Phys. Rev. Lett. {\bf 114}(10) 106806 (2015)

\bibitem{dNG}
De Nittis, G., Gomi, K.: Classification of ``Quaternionic'' Bloch-Bundles Topological Quantum Systems of Type AII. Commun. Math. Phys. {\bf 339} 1--55 (2015)

\bibitem{Dupont}
Dupont, J.L.: Symplectic bundles and $KR$-theory. Math. Scand. {\bf 24}(1) 27--30 (1969)

\bibitem{Dwivedi2}
Dwivedi, V., Chua, V.: Of bulk and boundaries: Generalized transfer matrices for tight-binding models. Phys. Rev. B {\bf 93} 134304 (2016)

\bibitem{Dwivedi}
Dwivedi, V., Ramamurthy, S.T.: Connecting the dots: Time-reversal symmetric Weyl semimetals with tunable Fermi arcs. Phys. Rev. B {\bf 94} 245143 (2016)

\bibitem{FM}
Freed, D.S., Moore, G.W.: Twisted equivariant matter. Ann. Henri Poincar{\'e} {\bf 14}(8) 1927--2023 (2013) 

\bibitem{FKM}
Fu, L., Kane, C.L., Mele, E.J.: Topological insulators in three dimensions. Phys. Rev. Lett. {\bf 98}(10) 106803 (2007)


\bibitem{Gaw2}
Gaw\k{e}dzki, K.: 2$d$ Fu--Kane--Mele invariant as Wess--Zumino action of the sewing matrix. Lett. Math. Phys. {\bf 107}(4) 733--755 (2017)


\bibitem{Hatsugai2}
Hatsugai, Y.: Edge states in the integer quantum Hall effect and the Riemann surface of the Bloch function. Phys. Rev. B {\bf 48}(16) 11851--11862 (1993)

\bibitem{Hatsugai}
Hatsugai, Y.: Symmetry-protected $\mathbb{Z}_2$-quantization and quaternionic Berry connection with Kramers degeneracy. New J. Phys. {\bf 12} 065004 (2010)

\bibitem{Herring}
Herring, C.: Accidental Degeneracy in the Energy Bands of Crystals. Phys. Rev. {\bf 52} 365--373 (1937)


\bibitem{Huang}
Huang, S.-M. et al.: New type of Weyl semimetal with quadratic double Weyl fermions. Proc. Natl. Acad. Sci. USA {\bf 113}(5) 1180--1185 (2016)

\bibitem{Hutchings}
Hutchings, M.: Reidemeister torsion in generalized Morse theory. Forum Math. {\bf 14} 209--244 (2002)


\bibitem{Iversen}
Iversen, B.: Cohomology of Sheaves. Springer-Verlag, Berlin Heidelberg (1986)


\bibitem{Kato}
Kato, T.: Perturbation theory for linear operators. Classics in Mathematics, Springer-Verlag, Berlin (1995)

\bibitem{Kauf}
Kaufmann, R.M. Li, D., Wehefritz-Kaufmann, B.: Notes on topological insulators. Rev. Math. Phys. {\bf 28}(10) 1630003 (2016)


\bibitem{Korbas}
Korba\v{s}, J.: Distributions, vector distributions, and immersions of manifolds in Euclidean spaces. In: Handbook of Global Analysis (edited by D. Krupka and D. Saunders), Elsevier Science, Amsterdam, pp. 665--724 (2008)

\bibitem{Kraus}
Kraus, Y.E., Lahini, Y., Ringel, Z., Verbin, M., Zilberberg, O.: Topological States and Adiabatic Pumping in Quasicrystals. Phys. Rev. Lett. {\bf 109} 106402 (2012)

\bibitem{Kuchment}
Kuchment, P.: An overview of periodic elliptic operators. Bull. Amer. Math. Soc. {\bf 53} 343--414 (2016) 

\bibitem{LM}
Lawson, H.B., Michelsohn, M.-L.: Spin Geometry. Princeton Univ. Press, Princeton, NJ (1989)

\bibitem{LinYau}
Lin, H.,  Yau, S.-T.: On exotic sphere fibrations, topological phases, and edge states in physical systems. Internat. J. Modern Phys. B {\bf 27}(19) 1350107 (2013)

\bibitem{Lindner}
Lindner, N.H., Refael, G., Galitski, V.: Floquet topological insulator in semiconductor quantum wells. Nature Phys. {\bf 7} 490--495 (2011) 

\bibitem{LFF}
Liu, J., Fang, C., Fu, L.: Tunable Weyl fermions and Fermi arcs in magnetized topological crystalline insulators. arXiv:1604.03947

\bibitem{Lv}
Lv, B.Q. et al.: Experimental Discovery of Weyl Semimetal TaAs. Phys. Rev. X {\bf 5} 031013 (2015)

\bibitem{MT2}
Mathai, V., Thiang, G.C.: T-duality of topological insulators. J. Phys. A: Math. Theor. {\bf 48}(42) 42FT02 (2015). 
arXiv:1503.01206

\bibitem{MT3}
Mathai, V., Thiang, G.C.: T-duality simplifies bulk-boundary correspondence. Commun. Math. Phys. {\bf 345}(2) 675--701 (2016). arXiv:1505.05250

\bibitem{MT1}
Mathai, V., Thiang, G.C.: T-duality simplifies bulk-boundary correspondence: some higher dimensional cases. Ann. Henri Poincar\'{e}, {\bf 17}(12) 3399--3424 (2016).
arXiv:1506.04492

\bibitem{MTFermi}
Mathai, V., Thiang, G.C.: Global topology of Weyl semimetals and Fermi arcs. J. Phys. A: Math. Theor. (Letter) {\bf 50}(11) 11LT01 (2017). arXiv:1607.02242

\bibitem{Milnor}
Milnor, J.: On Manifolds Homeomorphic to the 7-Sphere. Ann. of Math. (2) {\bf 64}(2) 399--405 (1956)

\bibitem{MilnorDT}
Milnor, J.:  Topology from the differentiable viewpoint. Based on notes by David W. Weaver. The University Press of Virginia, Charlottesville, VA (1965)

\bibitem{Molina}
Molina, O.M.: Co-Euler structures on bordisms. Topology Appl. {\bf 193} 51--76 (2015)

\bibitem{MS2}
Murray, M., Stevenson, D.: The basic bundle gerbe on unitary groups. J. Geom. Phys. {\bf 58}(11) 1571--1590 (2008)

\bibitem{MS}
Murray, M., Stevenson, D.:
{Bundle gerbes: stable isomorphism and local theory}. 
J. London Math. Soc. (2) {\bf 62} no. 3, 925-937 (2000)


\bibitem{Nico}
Nicolaescu, L.I.: The Reidemeister Torsion of 3-Manifolds. De Gruyter Stud. Math. {\bf 30}, De Gruyter, Berlin (2003)

\bibitem{NN}
Nielsen, H.B., Ninomiya, M.: Absence of neutrinos on a lattice: (II). Intuitive topological proof. Nucl. Phys. B {\bf 193} 173--194 (1981)

\bibitem{paechter}
Paechter, G.F.: The groups $\pi_r(V_{n,m})$. Q. J. Math. {\bf 7}(1) 249--268 (1956)

\bibitem{Poly}
Polyakov, A.M.: Particle spectrum in quantum field theory. Pisma. Zh. Eksp. Theor. Fiz. {\bf 20} 430  (1974) [JETP Lett. {\bf 20} 194 (1974)]

\bibitem{Prodan}
Prodan, E.: Virtual topological insulators with real quantized physics. Phys. Rev. B {\bf 91} 245104 (2015)


\bibitem{Read}
Read, N.: Compactly-supported Wannier functions and algebraic $K$-theory. Phys. Rev. B {\bf 95} 115309 (2017)

\bibitem{Rechtsman}
Rechtsman, M.C.: Photonic Floquet topological insulators. Nature {\bf 496} 196--200 (2013)

\bibitem{RS}
Reed, M., Simon, B.: Methods of Modern Mathematica Physics. Vol. IV: Analysis of Operators, Elsevier (1978)

\bibitem{Simon}
Simon, B.: Holonomy, the quantum adiabatic theorem, and Berry's phase. Phys. Rev. Lett. {\bf 51}(24) 2167 (1983)

\bibitem{Simons}
Simons, J., Sullivan, D.: The Mayer--Vietoris property in differential cohomology. arXiv:1010.5269

\bibitem{Sol}
Soluyanov, A.A., Gresch, D., Wang, Z., Wu, Q., Troyer, M., Dai, X., Bernevig, B.A.: Type-II Weyl semimetals. Nature {\bf 527} 495--498 (2015)

\bibitem{ShawLever}
Shaw, R., Lever, J.: Irreducible multiplier corepresentations of the extended Poincar\'{e} group. Commun. Math. Phys. {\bf 38}(4) 279--297 (1974)

\bibitem{Steenrod}
Steenrod, N.: The topology of fiber bundles. Princeton Univ. Press, Princeton, New Jersey (1951)

\bibitem{tHooft}
t' Hooft, G.: Magnetic monopoles in unified gauge theories. Nucl Phys. B {\bf 79}(2) 276--284 (1974)

\bibitem{TangZhang}
Tang, Z., Zhang, W.: A generalization of the Atiyah--Dupont vector fields theory. Commun. Contemp. Math. {\bf 4}(4) 777--796 (2002)

\bibitem{Thiang}
Thiang, G.C.: On the $K$-theoretic classification of topological phases of matter. Ann. Henri Poincar\'{e}. {\bf 17}(4) 757--794 (2016)

\bibitem{thomas}
Thomas, E.: The index of a tangent 2-field. Comment. Math. Helv. {\bf 42}(1) 86--110 (1967)

\bibitem{thomas2}
Thomas, E.: Vector fields on manifolds. Bull. Amer. Math. Soc. {\bf 75}(4) 643--683 (1969)

\bibitem{TuraevEuler}
Turaev, V.: Euler structures, nonsingular vector fields, and torsions of Reidemeister type. Izv. Math. {\bf 34} (3) 627--662 (1990)

\bibitem{TuraevSpinc}
Turaev, V.: Torsion invariants of Spin${}^c$-structures on 3-manifolds. Math. Res. Lett. {\bf 4} 679--695 (1997)

\bibitem{Turner}
Turner, A.M., Vishwanath, A.: Beyond band insulators: topology of semimetals and interacting phases. In: M. Franz and L. Molenkamp (Eds.), Contemp. Concepts Cond. Mat. Sci. {\bf 6}, \emph{Topological Insulators}, pp. 293--324, Elsevier, Amsterdam (2013)

\bibitem{Var}
Varadarajan, V.S.: Supersymmetry for mathematicians: An introduction, Vol. 11 of Courant Lect. Notes Math. Courant Inst. Math. Sci., New York (2004)

\bibitem{vN}
von Neumann, J., Wigner, E.P.: \"{U}ber merkw\"{u}rdige diskrete Eigenwerte. Physik. Zeits. {\bf 30} 467--470 (1929)

\bibitem{Wan}
Wan, X., Turner, A.M., Vishwanath, A., Savrasov, S.Y.: Topological semimetal and Fermi-arc surface states in the electronic structure of pyrochlore iridates. Phys. Rev. B {\bf 83} 205101 (2011)


\bibitem{Wigner}
Wigner, E.P.: Unitary representations of the inhomogeneous Lorentz group including reflections. In Group Theoretical Concepts in Elementary Particle Physics, (ed.\ F. G\"{u}rsey), vol. 1. pp. 37--80, Gordon and Breach, New York, (1964)

\bibitem{Witten}
Witten, E.: Three Lectures On Topological Phases Of Matter. La Rivista del Nuovo Cimento {\bf 39}(7) 313--370 (2016)

\bibitem{Xu1}
Xu, S.-Y. et al.: Discovery of a Weyl Fermion Semimetal and Topological Fermi Arcs. Science {\bf 349} 613--617 (2015)

\bibitem{Xu2}
Xu, S.-Y. et al.: Discovery of a Weyl fermion state with Fermi arcs in niobium arsenide. Nature Physics {\bf 11} 748--754 (2015)

\bibitem{XutypeII}
Xu, Y., Zhang, F., Zhang, C.: Structured Weyl Points in Spin-Orbit Coupled Fermionic Superfluids. Phys. Rev. Lett. {\bf 115} 265304 (2015)

\bibitem{Zhang}
Zhang, C. et al.: Signatures of the Adler--Bell--Jackiw chiral anomaly in a Weyl fermion semimetal. Nature Comm. {\bf 7} 10735 (2016) 


\bibitem{ZhangLian}
Zhang, S.-C., Lian, B.: Five-dimensional generalization of the topological Weyl semimetal. Phys. Rev. B {\bf 94} 041105(R) (2016)

\bibitem{Zhao1}
Zhao, Y.X., Wang, Z.D.: Topological Classification and Stability of Fermi Surfaces. Phys. Rev. Lett {\bf 110} 240404 (2013)




\end{thebibliography}
\end{document}